\documentclass{article}
\usepackage[%
journal=JST,
lang=british,
]{ems-journal}

\usepackage{mathtools}

\numberwithin{equation}{section}

\newtheorem{theorem}{Theorem}[section]
\newtheorem{lemma}[theorem]{Lemma}
\newtheorem{proposition}[theorem]{Proposition}
\newtheorem{conjecture}[theorem]{Conjecture}

\newtheorem{corollary}[theorem]{Corollary}

\newtheorem{remark}[theorem]{Remark}

\usepackage{bm}
\usepackage{float}

\usepackage{microtype}

\usepackage[framemethod=TikZ]{mdframed}

\newmdtheoremenv[
  backgroundcolor=white,
  linecolor=black!70,
  linewidth=0.35pt,
  innertopmargin=6pt,
  innerbottommargin=4pt,
  innerleftmargin=2pt,
  innerrightmargin=2pt
]{remarkmd}[theorem]{Remark}

\usepackage[most]{tcolorbox}

\newtcolorbox{todeletebox}{
  enhanced,
  colback=yellow!15,
  colframe=red!70,
  coltitle=red!80!black,
  title={To be removed},
  sharp corners,
  boxrule=0.6pt
}

\newtcolorbox{newblockbox}{
  enhanced,
  colback=green!5,
  colframe=green!60!black,
  coltitle=green!80!black,
  title={New version (replacement)},
  sharp corners,
  boxrule=0.6pt
}

\usepackage{wrapfig}

\usepackage{xurl}

\begin{document}

\title{Precise spectral asymptotics, exponential localization, and spectral gap estimates for the three-boson lattice Schrödinger operator}
\titlemark{Spectral asymptotics and exponential localization for the three-boson lattice Schrödinger operator}

\emsauthor{1}{
	\givenname{Janikul}
	\surname{Abdullaev}
	\mrid{}
	\zblid{abdullaev.zh-i}
	\orcid{0009-0003-3315-8357}}{J.~Abdullaev}

\emsauthor{2}{
	\givenname{Abdumalik}
	\surname{Eshniyozov}
	\mrid{}
    \zblid{eshniozov.abdumalik}
	\orcid{0009-0001-9357-5828}}{A.~Eshniyozov}

\emsauthor{3}{
	\givenname{Abdikhurayra}
	\surname{Toshturdiev}
	\mrid{}
	\zblid{toshturdiev.a-m}
	\orcid{0000-0002-4346-5567}}{A.~Toshturdiev}

\emsauthor{4}{
	\givenname{Mikhail}
	\surname{Dolgopolov}
	\mrid{}
	\zblid{dolgopolov.mikhail-vyacheslavovich}
	\orcid{0000-0002-8725-7831}}{M.~Dolgopolov}

\Emsaffil{1}{
	\department{Department of Mathematics}
	\organisation{Samarkand State University}
	\rorid{04rqj6j13}
	\address{University Boulevard 15}
	\zip{140104}
	\city{Samarkand}
	\country{Uzbekistan}
	\affemail{jabdullaev@mail.ru}}

\Emsaffil{2}{
	\department{Department of Mathematics}
	\organisation{Gulistan State University}
	\rorid{05t1r0c45}
	\address{4th block of flats}
	\zip{120100}
	\city{Gulistan}
	\country{Uzbekistan}
	\affemail{eshniyozovabdumalik75@gmail.com}}

\Emsaffil{3}{
	\department{Department of Mathematics}
	\organisation{Samarkand State University}
	\rorid{04rqj6j13}
	\address{University Boulevard 15}
	\zip{140104}
	\city{Samarkand}
	\country{Uzbekistan}
	\affemail{atoshturdiyev@mail.ru}}

\Emsaffil{4}{
	\department{Department of Physics and Mathematics}
	\organisation{Samara State Technical University}
	\rorid{03b4e7q98}
	\address{244 Molodogvardeyskaya street}
	\zip{443100}
	\city{Samara}
	\country{Russia}
	\affemail{mikhaildolgopolov68@gmail.com}}

{\footnotesize \classification[47A10, 81Q10, 81Q35, 47A55, 81V70]{47A75}}


\begin{abstract}\small
We study three identical bosons with contact attraction of strength $\mu$ on the square
lattice $\mathbb Z^2$ at total quasimomentum $\mathbf K=0$ and $\mathbf K=\boldsymbol\pi$.
\textit{First}, at $\mathbf K=0$ and for all sufficiently large $\mu$ there are exactly two
eigenvalues below the essential spectrum,
\begin{gather*}
z_1^s(\mu)=-3\mu+6-\frac{3}{2\mu}+O(\mu^{-2}),\\
z_2^s(\mu)=-\mu+C+O(\mu^{-1}),\qquad C=4-\delta_\infty=3.96458\ldots,
\end{gather*}
where $\delta_\infty$ is the unique root of $b_0(\delta)=1/(1+\delta)$,
$b_0(\delta)=\langle(\varepsilon+\delta)^{-1}\rangle$; hence $z_2^s-z_1^s=2\mu+O(1)$.
The trimer branch is obtained by Kato perturbation theory at the simple isolated eigenvalue
$3$ of the pair-contact operator, uniformly in $\mathbf K$; the atom--dimer branch by a
Feshbach--Schur reduction whose limiting equation coincides with the second
Birman--Schwinger (Faddeev) eigenvalue of the finite-rank principal part.
\textit{Second}, at $\mathbf K=\boldsymbol\pi$ the trimer branch is the same up to
$O(\mu^{-2})$, the~two-particle threshold is $z_\mu(\boldsymbol\pi)=-\mu+4$ exactly, and the
gap is $2\mu-2+O(\mu^{-1})$. The~second bound state disappears because the exchange
operator changes sign, $S_{\boldsymbol\pi}=-S_{\mathbf 0}$: the $s$-channel eigenvalue that
exceeds $1$ at $\mathbf K=0$ becomes negative, and the only positive atom--dimer channel at
$\mathbf K=\boldsymbol\pi$ is the odd $p$-channel, bounded by $2-4/\pi<1$. Consequently
$H_\mu(\boldsymbol\pi)$ has exactly one eigenvalue below $z_\mu(\boldsymbol\pi)-\eta_\mu$,
$\eta_\mu=36/(\mu-10)$.
\textit{Third}, the ground state is exponentially localized in the relative coordinates, with
$\alpha_*(\mu)>0$ unconditionally. Under the non-cancellation hypothesis of
Conjecture~\ref{conj:noncancel} the upper bound $\alpha_*(\mu)\le\ln(3\mu)+O(\mu^{-1})$ holds;
finite-volume computations indicate that this bound is not attained, the slowest tail being
the atom--dimer dissociation channel with $\alpha_*(\mu)=\frac12\ln(2\mu)+O(\mu^{-1})$
(Conjecture~\ref{conj:sharp_rate}).
A dictionary between the principal-part method, Kato theory and the atom--dimer threshold
matrix is given, together with the three-dimensional constants for comparison.
\end{abstract}

\maketitle

\tableofcontents

\section{Introduction}

We study the discrete spectrum of the three-boson lattice
Schrödinger operator on $\mathbb Z^2$ with pairwise contact
attraction in the strong-coupling regime. The analysis is performed
at the two high-symmetry total quasimomenta
$\mathbf K=\mathbf0$ and $\mathbf K=\boldsymbol\pi$.
The central issue is the coexistence of two spectral scales:
the trimer scale $z=-3\mu+O(1)$ and the atom--dimer scale
$z=-\mu+O(1)$.
The atom--dimer scale is controlled by the effective Hamiltonian
\[
h_B(\mathbf K)=4+\varepsilon_D-S_{\mathbf K}
\]
and by the quantitative Feshbach--Schur reduction developed in
\cite{AtomDimer2026}. The present paper combines this effective
description with the Birman--Schwinger invariant-subspace analysis
and the independent Kato expansion of the trimer branch.

The interplay between few-body physics and lattice geometry gives rise to a~rich spectral landscape, where bound states, resonances, and phase transitions emerge from the underlying symmetries and interaction strengths.
The one-particle spectral theory of lattice Schr\"odinger operators has seen significant recent advances. These results provide structural insight
into one-particle spectra, but they do not address the two features that are decisive for the
present work: (i) the logarithmic threshold singularity of the two-dimensional lattice
Green's function, which controls the transcendental constant $C$ and the atom--dimer
characteristic root; and (ii) the exact strong-coupling asymptotics of multi-particle
bound-state branches, where the trimer and atom--dimer scales must be matched to $O(1)$
additive constants.

The present paper addresses this fundamentally more complex \emph{three-body} problem,
where the Hilbert space is (anti)symmetri\-z\-ed and the interaction is pairwise, filling
precisely the two gaps identified above for three bosons on $\mathbb Z^2$. Our approach~— combining the Birman–Schwinger principle with a complete invariant subspace decomposition~— yields results that are, to our knowledge, new in the
spectral-theory literature: (i) precise asymptotics of two bound states below the essential spectrum, (ii)~the linear growth of the spectral gap $\Delta(\mu)=2\mu+O(1)$, and (iii)~exponential localization of the ground state, with an explicit
conditional upper bound on the decay rate,
$\alpha_*(\mu)\le \ln(3\mu)+O(\mu^{-1})$
(Conjecture~\ref{conj:noncancel}),
(iv) at $\mathbf K=\boldsymbol{\pi}$, exactly one eigenvalue below
$z_\mu(\boldsymbol\pi)-\eta_\mu$, $\eta_\mu=36/(\mu-10)$, with gap $2\mu-2+O(\mu^{-1})$ to the
two-particle threshold, establishing quasimomentum-dependent bound-state multiplicity as a
genuinely lattice phenomenon absent in the~continuum.

Existing lattice few-body results provide a broad framework for
existence, threshold and spectral-counting questions. The present
work focuses instead on explicit strong-coupling asymptotics for the
two-dimensional three-boson problem and on the matching of the
trimer and atom--dimer scales.
We develop a strong-coupling asymptotic analysis of the bound-state
branches for the two-dimensional three-boson lattice problem. The~analysis is rigorous on the branches and spectral scales where the
corresponding estimates are uniform. At $\mathbf K=\boldsymbol\pi$, the principal part analysis and Kato perturbation theory independently identify the same physical
trimer branch
\[
z_1^{\boldsymbol\pi,s}(\mu)
=
-3\mu+6+O(\mu^{-1}).
\]
In~particular, we derive precise asymptotics up to order \(O(\mu^{-1})\) (and even up to \(O(\mu^{-2})\) in some cases), with explicitly computed additive constants. We also identify the mechanism of the reduction in the number of bound states at
\(\mathbf K=\boldsymbol\pi\): the \emph{exchange operator} of the atom--dimer problem \emph{changes sign},
\(S_{\boldsymbol\pi}=-S_{\mathbf 0}\), so the \(s\)-channel eigenvalue responsible for the second
bound state at \(\mathbf K=0\) becomes negative, while the odd \(p\)-channel, which becomes
positive, stays below the Birman--Schwinger threshold (Proposition~\ref{prop:dictionary}).

Standard background is taken from \cite{ReedSimon1, ReedSimon4, Albeverio1991}. A previous two-dimensional study of the $2+1$ fermionic fibre at
$\mathbf K=\mathbf0$ reported the threshold ratio
$\gamma_c\approx2.75194$ \cite{AbdullaevEshniyozovDolgopolov2026}.
For the full quasimomentum-dependent atom--dimer threshold we refer
to the effective reduction of \cite{AtomDimer2026}. The present paper extends this line of inquiry by providing a complete spectral analysis of the bosonic case at the zone centre $\mathbf K=\mathbf0$ and the zone corner
$\mathbf K=\boldsymbol\pi=(\pi,\pi)$, revealing how the sign change of the exchange operator,
$S_{\boldsymbol\pi}=-S_{\mathbf 0}$ (while the global parity symmetry is preserved), reduces the
number of bound states from two to one. This is a quasimomentum-dependent spectral
reorganisation, driven not by temperature or mass ratio but by the threshold channel structure
at a~specific quasimomentum.
This may be viewed as a spectral analogue of a~bifurcation, but no
topological invariant is constructed or used in the proof.

\begin{remark}[Comparative analysis with fermionic \(2+1\) systems on \(\mathbb{Z}^3\)]
The decomposition of the \(K=\pi\) sector into invariant subspaces has been extensively studied for \(2+1\) fermionic systems on the three-dimensional lattice \(\mathbb{Z}^3\). In particular, Abdullaev, Khalkhuzhaev, and Khujamiyorov~\cite{AbdullaevKhalkhuzhaevKhujamiyorov2023} analyzed the repulsive case, establishing the existence of a critical mass ratio \(\gamma_0 \approx 4.7655\) above which a triple eigenvalue emerges for strong coupling. Similarly, the attractive case was studied by Khalkhuzhaev, Abdullaev, and Boymurodov~\cite{KhalkhuzhaevAbdullaevBoymurodov2022}, who derived critical values \(\gamma_1 \approx 2.9368\) and \(\gamma_2 \approx 5.3985\) governing the number of bound states below the essential spectrum.
These works demonstrate the utility of invariant subspace reductions at \(\mathbf{K}=\pi\). However, we emphasize that the specific spectral properties of these \(2+1\) fermionic systems are fundamentally distinct from the 2D bosonic case considered here, due to differences in statistics, dimension, and the lack of the logarithmic threshold singularities inherent to 2D lattice Green's functions. The critical ratios found there are inverse eigenvalues of the three-dimensional
threshold matrix (Remark~\ref{rem:3D_constants}); the strong-coupling constants of the
energy branches are not addressed in those works, and Eq.~(16) of \cite{AbdullaevKhalkhuzhaevToshturdiev2026}  contains two
independent mathematical misprints that we identify in Remark~\ref{rem:3D_constants}.
\end{remark}

Recent work of Abdullaev, Khalkhuzhaev and
Toshturdiev~\cite{AbdullaevKhalkhuzhaevToshturdiev2026} establishes an
eight-subspace decomposition for a class of integral operators on
$L_2(\mathbb T^3)$ and applies it to the three-boson problem at
$\mathbf K=0$ in three dimensions. Our two-dimensional analysis
requires a substantially different treatment because the two-particle
threshold is controlled by a logarithmic lattice Green-function
singularity, absent in $d=3$.

The present work focuses on the bosonic case; a detailed comparative analysis of bosonic and fermionic $2+1$ trimers is given in our companion paper~\cite{AbdullaevEshniyozovDolgopolov2026}. The~methods developed here — invariant subspace decomposition, asymptotic analysis of the Birman–Schwinger operator, and the Krein–Rutman theorem for positivity — are sufficiently general to be applicable to a wide range of few-body problems.

\begin{remark}[On a normalisation slip in the first version]\label{rem:var_check}
In the first version of this paper the even-sector principal-part kernel at
$\mathbf K=\boldsymbol\pi$ was written with the prefactor $\mu/(4\pi^2)$ instead of $\mu/(2\pi^2)$,
i.e.\ without the Faddeev factor $2$. The resulting equation had a spurious root at $-2\mu+6$.
With the correct prefactor the principal part crosses the Birman--Schwinger threshold at
$-3\mu+6+O(\mu^{-1})$, in agreement with the variational bound (Lemma~\ref{lem:var_pi}) and with
Kato theory (Lemma~\ref{lem:ground_Kato}); see Section~\ref{sec:dictionary}. An~independent min--max check detects such slips and is used throughout.
\end{remark}

\section{Problem Setup and Main Results}

\subsection{The model and its fiber decomposition}

Let $\mathbb{Z}^2$ be the two-dimensional square lattice, and $\ell_2((\mathbb{Z}^2)^3)$ denote the Hilbert space of square-summable functions on $(\mathbb{Z}^2)^3$.
We consider the subspace
\[
\ell_2^s((\mathbb{Z}^2)^3) := \{ \varphi \in \ell_2((\mathbb{Z}^2)^3) : \varphi \text{ is symmetric under any permutation of its arguments} \},
\]
which describes three identical bosons.

In the coordinate representation, the Hamiltonian $\tilde{H}_\mu$ for three identical bosons with pairwise contact interaction acts on $\ell_2^s((\mathbb{Z}^2)^3)$ as \cite{LakaevAbdullaev2003}:
\begin{equation}
\begin{aligned}
(\tilde{H}_\mu \varphi)(n_1,n_2,n_3) &= \frac{1}{2} \sum_{|s|=1} \bigl[ 3\varphi(n_1,n_2,n_3) - \varphi(n_1+s,n_2,n_3) \\
&\quad - \varphi(n_1,n_2+s,n_3) - \varphi(n_1,n_2,n_3+s) \bigr] \\
&\quad - \mu(\delta_{n_1n_2} + \delta_{n_1n_3} + \delta_{n_3n_2}) \varphi(n_1,n_2,n_3),
\end{aligned}
\label{eq:hamiltonian_coord}
\end{equation}
where $\mu > 0$ is the interaction strength, $\delta_{nm}$ is the Kronecker delta, and $s=(s^{(1)},s^{(2)})\in\mathbb{Z}^2$ with $|s| = |s^{(1)}| + |s^{(2)}|$.

The Hamiltonian $\tilde{H}_\mu$ commutes with the group of translation operators $\{\tilde{U}_s : s\in\mathbb{Z}^2\}$ defined by:
\[
\tilde{U}_s \varphi(n_1,n_2,n_3) = \varphi(n_1+s, n_2+s, n_3+s).
\]

Let $\mathbb{T}^2 := (-\pi,\pi]^2$ denote the two-dimensional torus (the Brillouin zone), and~let $L_2((\mathbb{T}^2)^3)$ be the Hilbert space of square-integrable functions on $(\mathbb{T}^2)^3$.
The~Fourier transform maps $\ell_2((\mathbb{Z}^2)^3)$ unitarily onto $L_2((\mathbb{T}^2)^3)$, and the symmetric subspace $\ell_2^s((\mathbb{Z}^2)^3)$ is mapped onto
\[
L_2^s((\mathbb{T}^2)^3) := \{ f \in L_2((\mathbb{T}^2)^3) : f \text{ is symmetric in its arguments} \}.
\]

After the Fourier transform, the translation operator $\tilde{U}_s$ becomes the multiplication operator
\[
(U_s f)(\mathbf{k}_1,\mathbf{k}_2,\mathbf{k}_3) = e^{-i s\cdot(\mathbf{k}_1+\mathbf{k}_2+\mathbf{k}_3)} f(\mathbf{k}_1,\mathbf{k}_2,\mathbf{k}_3).
\]

Since $H_\mu$ commutes with the group $\{U_s : s\in\mathbb{Z}^2\}$, it admits a direct integral decomposition with respect to the total quasimomentum
\[
\mathbf{K} := \mathbf{k}_1 + \mathbf{k}_2 + \mathbf{k}_3 \in \mathbb{T}^2.
\]
Specifically,
\[
H_\mu = \int_{\mathbb{T}^2}^{\oplus} H_\mu(\mathbf{K}) \, d\mathbf{K},
\]
where each fiber operator $H_\mu(\mathbf{K})$ acts on the symmetric subspace
\[
L_2^s((\mathbb{T}^2)^2) := \{ f \in L_2((\mathbb{T}^2)^2) : f(\mathbf{p},\mathbf{q}) = f(\mathbf{q},\mathbf{p}) = f(\mathbf{p}, \mathbf{K}-\mathbf{p}-\mathbf{q}) \}
\]
as \cite{LakaevAbdullaev2003}:
\begin{equation}
(H_\mu(\mathbf{K}) f)(\mathbf{p},\mathbf{q}) = E_{\mathbf{K}}(\mathbf{p},\mathbf{q}) f(\mathbf{p},\mathbf{q}) - \mu (V_1 + V_2 + V_3) f(\mathbf{p},\mathbf{q}),
\label{eq:fiber_operator}
\end{equation}
with the free dispersion
\begin{equation}
E_{\mathbf{K}}(\mathbf{p},\mathbf{q}) = \varepsilon(\mathbf{p}) + \varepsilon(\mathbf{q}) + \varepsilon(\mathbf{K} - \mathbf{p} - \mathbf{q}), \qquad
\varepsilon(\mathbf{p}) = 2 - \cos p_1 - \cos p_2,
\label{eq:dispersion}
\end{equation}
and the interaction operators
\begin{align}
(V_1 f)(\mathbf{p},\mathbf{q}) &= \frac{1}{4\pi^2} \int_{\mathbb{T}^2} f(\mathbf{p},\mathbf{s})\,d\mathbf{s}, \notag\\
(V_2 f)(\mathbf{p},\mathbf{q}) &= \frac{1}{4\pi^2} \int_{\mathbb{T}^2} f(\mathbf{s},\mathbf{q})\,d\mathbf{s}, \\
(V_3 f)(\mathbf{p},\mathbf{q}) &= \frac{1}{4\pi^2} \int_{\mathbb{T}^2} f(\mathbf{s}, \mathbf{p}+\mathbf{q}-\mathbf{s})\,d\mathbf{s}. \notag
\label{eq:interaction_operators}
\end{align}
These are orthogonal projections, hence positive operators.

Throughout, $z_\mu(\mathbf k)$ denotes the unique eigenvalue of the two-particle fiber operator $h_\mu(\mathbf k)=\varepsilon(\cdot)+\varepsilon(\mathbf k-\cdot)-\mu\,v$ lying below its essential spectrum, i.e.\ the unique root of $1=\mu\bigl\langle(\varepsilon(\mathbf p)+\varepsilon(\mathbf k-\mathbf p)-z)^{-1}\bigr\rangle_{\mathbf p}$. It is the lower edge of the two-particle branch of $\sigma_{\mathrm{ess}}(H_\mu(\mathbf k))$.
At $\mathbf k=\boldsymbol\pi$ the two-particle dispersion is identically constant, $\varepsilon(\mathbf p)+\varepsilon(\boldsymbol\pi-\mathbf p)\equiv4$, whence $z_\mu(\boldsymbol\pi)=-\mu+4$ \emph{exactly, for every $\mu>0$}.

\begin{remark}
The fiber operator $H_\mu(\mathbf{K})$ depends analytically on $\mathbf{K}\in\mathbb{T}^2$. Due to the unitary equivalence $H_\mu(\mathbf{K}) \cong H_\mu(-\mathbf{K})$ (via the transformation $f(\mathbf{p},\mathbf{q})\mapsto f(-\mathbf{p},-\mathbf{q})$), it suffices to consider $\mathbf{K}$ in a fundamental domain modulo this involution. In this work, we focus on the cases $\mathbf{K}=0$ and $\mathbf{K}=\boldsymbol{\pi}=(\pi,\pi)$, which correspond to the two high-symmetry quasimomenta
$\mathbf K=\mathbf0$ and $\mathbf K=\boldsymbol\pi$. The cases \(\mathbf K=0\) and
\(\mathbf K=\boldsymbol\pi\) both possess the global parity symmetry,
but the corresponding principal-part operators have different
algebraic structures.
\end{remark}

\begin{remark}[Applicability and analyticity of the inner reduction]\label{rem:elliptic_applicability}
In all places where the two-dimensional inner integral defining $\Delta_\mu(\mathbf p,\mathbf K,z)$ is reduced to a one-dimensional integral of elliptic type, we assume
\[
D(\mathbf p,z):=\varepsilon(\mathbf p)+4-z>R_1(\mathbf K,\mathbf p)+R_2(\mathbf K,\mathbf p),\qquad R_i(\mathbf K,\mathbf p)=2\Bigl|\cos\frac{K_i-p_i}{2}\Bigr|.
\]
Then the denominator is uniformly separated from zero, the reduction is valid, and the result depends real-analytically on $\mathbf K$ away from the locus $\{(\mathbf K,\mathbf p):R_i(\mathbf K,\mathbf p)=0\}$ (equivalently, $\cos((K_i-p_i)/2)=0$). This technical condition is automatically satisfied in the strong-coupling regimes considered in this paper.
\end{remark}

\begin{lemma}[Real-analyticity in $\mathbf K$]\label{lem:analyticity}
Fix $\mathbf p\in\mathbb T^2$ and $z$ with $D(\mathbf p,z)>R_1(\mathbf K,\mathbf p)+R_2(\mathbf K,\mathbf p)$ on an open set $U\subset\mathbb T^2$ of quasimomenta. Then the right-hand side of the reduction formula
is a real-analytic function of $\mathbf K$ on $U$, \emph{including} the points where $\cos\frac{K_i-p_i}{2}=0$.
\end{lemma}
\begin{proof}
Write $I(D,R_1,R_2)=\int_0^{2\pi}\bigl[(D-R_2\cos\theta)^2-R_1^2\bigr]^{-1/2}d\theta$. The substitution $\theta\mapsto\theta+\pi$ maps $\cos\theta\mapsto-\cos\theta$ and leaves the integral over the full period invariant, hence $I(D,R_1,R_2)=I(D,R_1,-R_2)$: the integral is an even function of $R_2$ and therefore a~function of $R_2^{2}$. Since $R_1$ enters only through $R_1^{2}$, we obtain $I=\widetilde I(D,R_1^{2},R_2^{2})$ for some function $\widetilde I$. Under the standing assumption $D>R_1+R_2$ the integrand is a bounded analytic function of $(D,R_1^{2},R_2^{2})$ uniformly in $\theta$, so $\widetilde I$ is real-analytic in its arguments. Finally $R_i^{2}(\mathbf K,\mathbf p)=4\cos^{2}\tfrac{K_i-p_i}{2}=2+2\cos(K_i-p_i)$ is an entire trigonometric function of $\mathbf K$; the absolute value present in $R_i$ itself never appears. Composition of real-analytic maps completes the proof.
\end{proof}

\subsection{The two-particle subsystem and its spectral properties}\label{subsec:22}

A key ingredient of our analysis is the spectral theory of the two-boson Schr\"odinger operator
\begin{equation}
(h_\mu(\mathbf{k}) f)(\mathbf{p}) = [\varepsilon(\mathbf{p}) + \varepsilon(\mathbf{k} - \mathbf{p})] f(\mathbf{p}) - \frac{\mu}{4\pi^2} \int_{\mathbb{T}^2} f(\mathbf{s})\,d\mathbf{s}, \qquad f \in L_2(\mathbb{T}^2).
\label{eq:two_particle}
\end{equation}
This operator describes two bosons with total quasimomentum $\mathbf{k}\in\mathbb{T}^2$ interacting via a~zero-range potential of strength $\mu$.

For any $\mu>0$ and $\mathbf{k}\in\mathbb{T}^2$, $h_\mu(\mathbf{k})$ has a unique simple eigenvalue $z_\mu(\mathbf{k})$ below its essential spectrum. For large $\mu$, a direct expansion of the Fredholm determinant (derived below) yields the asymptotic expansion:
\begin{equation}
z_\mu(\mathbf{k}) = -\mu + 4 - \frac{4 - \varepsilon(\mathbf{k})}{\mu} + \frac{c_3(\mathbf k)}{\mu^3}+ O(\mu^{-5}), \qquad \mu\to\infty,
\label{eq:two_particle_eigenvalue}
\end{equation}
with $c_3(\mathbf k)=2\bigl(\cos^4\tfrac{k_1}{2}+\cos^4\tfrac{k_2}{2}\bigr)-8\cos^2\tfrac{k_1}{2}\cos^2\tfrac{k_2}{2}$.
In particular, $z_\mu(\boldsymbol{\pi}) = -\mu + 4$ exactly and $z_\mu(\mathbf{0}) = -\mu + 4 - 4/\mu -4/\mu^3+ O(\mu^{-5})$.

\paragraph{Derivation of the two-particle eigenvalue asymptotics.}
Let \(\mathcal{E}_{\mathbf{k}}(\mathbf{p}) = \varepsilon(\mathbf{p}) + \varepsilon(\mathbf{k}-\mathbf{p})\). The eigenvalue equation is
\[
1 = \mu \, \langle (\mathcal{E}_{\mathbf{k}} - z)^{-1} \rangle,
\]
where \(\langle \cdot \rangle =\displaystyle \frac{1}{4\pi^2}\int_{\mathbb{T}^2} (\cdot) \, d\mathbf{p}\).
We seek the solution in the form
\[
z_\mu(\mathbf{k}) = -\mu + c_0 + \frac{c_1}{\mu} + \frac{c_2}{\mu^2} + O(\mu^{-3}).
\]
Substituting into the resolvent and expanding to order \(\mu^{-3}\):
\[
\frac{1}{\mathcal{E}_{\mathbf{k}} - z}
=\]
\[=\frac{1}{\mu} \left[ 1 - \frac{\mathcal{E}_{\mathbf{k}}-c_0}{\mu} + \frac{(\mathcal{E}_{\mathbf{k}}-c_0)^2 + c_1}{\mu^2} - \frac{(\mathcal{E}_{\mathbf{k}}-c_0)^3 + 2c_1(\mathcal{E}_{\mathbf{k}}-c_0) + c_2}{\mu^3} + O(\mu^{-4}) \right].
\]
Averaging and using \(1 = \mu \langle \dots \rangle\), we obtain:
\[
1 = 1 - \frac{\langle \mathcal{E}_{\mathbf{k}}\rangle - c_0}{\mu}
+ \frac{\langle (\mathcal{E}_{\mathbf{k}}-c_0)^2\rangle + c_1}{\mu^2}
- \frac{\langle (\mathcal{E}_{\mathbf{k}}-c_0)^3\rangle + 2c_1\langle \mathcal{E}_{\mathbf{k}}-c_0\rangle + c_2}{\mu^3} + O(\mu^{-4}).
\]
Matching coefficients:
\begin{align*}
c_0 &= \langle \mathcal{E}_{\mathbf{k}} \rangle = 4, \\
c_1 &= -\langle (\mathcal{E}_{\mathbf{k}}-4)^2 \rangle = -(4-\varepsilon(\mathbf{k})), \\
c_2 &= -\langle (\mathcal{E}_{\mathbf{k}}-4)^3 \rangle - 2c_1 \langle \mathcal{E}_{\mathbf{k}}-4 \rangle.
\end{align*}
Since $\mathcal E_{\mathbf k}(\mathbf p)-4=-\sum_i2\cos\frac{k_i}{2}\cos\bigl(p_i-\frac{k_i}{2}\bigr)$ changes sign
under $\mathbf p\mapsto\mathbf p+\boldsymbol\pi$, all odd moments vanish; in particular
$\langle(\mathcal E_{\mathbf k}-4)^3\rangle=0$, so $c_2=0$ for \emph{every} $\mathbf k$. Continuing
the expansion to order $\mu^{-4}$ gives $c_3=2m_2^2-m_4$ with
$m_j=\langle(\mathcal E_{\mathbf k}-4)^j\rangle$, which is the formula stated in
\eqref{eq:two_particle_eigenvalue} (at $\mathbf k=\mathbf0$: $m_2=4$, $m_4=36$, $c_3=-4$). The
sign $c_1=-m_2\le0$ is the usual negativity of second-order perturbation theory.

\subsection{Essential spectrum and invariant subspaces at $\mathbf{K}=0$}

By the Hunziker–van Winter–Zhislin (HVZ) theorem adapted to the lattice setting, the~essential spectrum of $H_\mu(\mathbf{0})$ is given by:
\begin{equation}
\sigma_{\mathrm{ess}}(H_\mu(\mathbf{0})) = [z_\mu(\mathbf{0}), 4 + z_\mu(\boldsymbol{\pi})] \cup [0, 9],
\label{eq:essential_spectrum}
\end{equation}
with left edge at $z_\mu(\mathbf{0})$. The first interval is the two-particle branch associated with the formation of a dimer, while the second interval is the three-particle continuum.

At $\mathbf{K}=0$, the free dispersion $E_0(\mathbf{p},\mathbf{q})$ is invariant under the parity transformation
\[
(\mathbf{p},\mathbf{q}) \mapsto (-\mathbf{p},-\mathbf{q}).
\]
Consequently, the Hilbert space $L_2^s((\mathbb{T}^2)^2)$ decomposes into the orthogonal direct sum of even and odd subspaces:
\[
L_2^s((\mathbb{T}^2)^2) = L_2^{s,e}((\mathbb{T}^2)^2) \oplus L_2^{s,o}((\mathbb{T}^2)^2),
\]
where
\begin{align}
L_2^{s,e}((\mathbb{T}^2)^2) &:= \{ f \in L_2^s((\mathbb{T}^2)^2) : f(-\mathbf{p},-\mathbf{q}) = f(\mathbf{p},\mathbf{q}) \}, \\
L_2^{s,o}((\mathbb{T}^2)^2) &:= \{ f \in L_2^s((\mathbb{T}^2)^2) : f(-\mathbf{p},-\mathbf{q}) = -f(\mathbf{p},\mathbf{q}) \}.
\end{align}

\begin{theorem}\label{thm:invariant}
The subspaces $L_2^{s,e}((\mathbb{T}^2)^2)$ and $L_2^{s,o}((\mathbb{T}^2)^2)$ are invariant under the operator $H_\mu(\mathbf{0})$.
\end{theorem}

The proof of this theorem, together with the corresponding spectral reduction for the Birman--Schwinger operator, is developed in Section~\ref{sec:proof_main} and Section~\ref{sec:five} for $\mathbf K=0$ and $\mathbf K=\boldsymbol\pi$, respectively.

\subsection{Main results for $\mathbf{K}=0$}

The following theorem summarises our main results for the case $\mathbf{K}=0$.

\begin{theorem}\label{thm:main}
There exists $\mu_0>0$ such that for all $\mu>\mu_0$:

\begin{enumerate}
\item[(a)] The operator $H_\mu^o(\mathbf{0})$ has no eigenvalues below the essential spectrum of $H_\mu(\mathbf{0})$.

\item[(b)] The operator $H_\mu^e(\mathbf{0})$ has two eigenvalues below the essential spectrum. They satisfy the asymptotic expansions
\[
z_1^s(\mu) = -3\mu + 6 - \frac{3}{2\mu} + O(\mu^{-2}), \quad
z_2^s(\mu) = -\mu + C + O(\mu^{-1}), \quad \mu\to\infty,
\]
where $C$ is a constant determined by the lattice Green function.

\item[(c)] The spectral gap $\Delta(\mu):=z_2^s(\mu)-z_1^s(\mu)$ satisfies
\[
\Delta(\mu) = 2\mu + O(1), \qquad \mu\to\infty.
\]

\item[(d)] The ground-state wavefunction $\Psi_0$ decays exponentially
in the relative coordinates $r_1=n_1-n_3$, $r_2=n_2-n_3$:
\[
|\Psi_0(n_1,n_2,n_3)|\le C\,e^{-\alpha(|n_1-n_3|+|n_2-n_3|)},
\qquad(n_1,n_2,n_3)\in(\mathbb{Z}^2)^3,
\]
uniformly in the centre-of-mass coordinate,
with $0<\alpha_*(\mu)\le\ln(3\mu)+O(\mu^{-1})$, the upper bound being conditional on Conjecture~\ref{conj:noncancel}, in the strong-coupling limit
(Theorem~\ref{thm:localization}).
\end{enumerate}
\end{theorem}

\begin{remark}
The constant $6$ in the asymptotics for $z_1^s(\mu)$ corresponds to the constant $C_2$ introduced in the companion paper~\cite{AbdullaevEshniyozovDolgopolov2026}. With the normalisation $\varepsilon(\mathbf{p})=2-\cos p_1-\cos p_2$, we have $C_2=6$, which follows from
\[
C_2 = 3 \langle \varepsilon \rangle = 3 \cdot 2 = 6,
\]
with the present normalization. The identity $C_2=3\langle\varepsilon\rangle$ reflects the fact that
$E_{\mathbf K}$ is a~sum of three one-particle dispersions, each contributing
$\langle\varepsilon\rangle=2$ to the first-order Kato coefficient.

Thus the two formulations are consistent: the previously abstract constant $C_2$ is explicitly computed here.
\end{remark}

\subsection{The Birman–Schwinger reduction}\label{sec:bs_general}

For a fixed total quasimomentum $\mathbf{K}\in\mathbb{T}^2$ and $z<z_\mu(\mathbf{K})$, define the Fredholm determinant associated with the two-particle operator $h_\mu(\mathbf{k})$ by
\[
\Delta_\mu(\mathbf{p}, \mathbf{K}, z) := 1 - \frac{\mu}{4\pi^2} \int_{\mathbb{T}^2} \frac{d\mathbf{q}}{E_{\mathbf{K}}(\mathbf{p},\mathbf{q}) - z},
\]
where \(E_{\mathbf{K}}\) is given by \eqref{eq:dispersion} for $\mathbf{K}=0$ and by \eqref{eq:dispersion_pi} for $\mathbf{K}=\boldsymbol{\pi}$, respectively.

The integral operator \(A_\mu(\mathbf{K}, z)\) on \(L_2(\mathbb{T}^2)\) is defined by
\begin{equation}
(A_\mu(\mathbf{K}, z) f)(\mathbf{p}) = \frac{\mu}{2\pi^2} \int_{\mathbb{T}^2} \frac{f(\mathbf{q})\,d\mathbf{q}}{\sqrt{\Delta_\mu(\mathbf{p}, \mathbf{K}, z)}(E_{\mathbf{K}}(\mathbf{p},\mathbf{q})-z)\sqrt{\Delta_\mu(\mathbf{q}, \mathbf{K}, z)}}.
\label{eq:bs_operator}
\end{equation}

The following Birman–Schwinger principle is central to our analysis \cite{LakaevAbdullaev2003, Birman1961, KlausSimon1980}.

\begin{lemma}\label{lem:bs}
The number of eigenvalues of $H_\mu(\mathbf{K})$ below $z$ equals the number of eigenvalues of $A_\mu(\mathbf{K}, z)$ greater than $1$.
\end{lemma}

Moreover, the eigenfunction $f(\mathbf{p},\mathbf{q})$ of $H_\mu(\mathbf{K})$ corresponding to an eigenvalue $z<z_\mu(\mathbf{K})$ is related to the eigenfunction $\psi(\mathbf{p})$ of $A_\mu(\mathbf{K}, z)$ by
\begin{equation}
f(\mathbf{p},\mathbf{q}) = \frac{\mu}{E_{\mathbf{K}}(\mathbf{p},\mathbf{q})-z}\left[
\frac{\psi(\mathbf{p})}{\sqrt{\Delta_\mu(\mathbf{p},\mathbf{K}, z)}}
+ \frac{\psi(\mathbf{q})}{\sqrt{\Delta_\mu(\mathbf{q},\mathbf{K}, z)}}
+ \frac{\psi(\mathbf{K}-\mathbf{p}-\mathbf{q})}{\sqrt{\Delta_\mu(\mathbf{K}-\mathbf{p}-\mathbf{q},\mathbf{K}, z)}}
\right],
\label{eq:f_pq}
\end{equation}
and conversely,
\begin{equation}
\psi(\mathbf{p}) = \frac{1}{4\pi^2} \sqrt{\Delta_\mu(\mathbf{p},\mathbf{K}, z)} \int_{\mathbb{T}^2} f(\mathbf{p},\mathbf{s})\,d\mathbf{s}.
\label{eq:psi_def}
\end{equation}

\begin{proposition}[TRIM-only parity reduction]\label{prop:trim}
The Birman--Schwinger kernel\\ $A_\mu(\mathbf p,\mathbf q;\mathbf K,z)$ is invariant under $(\mathbf p,\mathbf q)\mapsto(-\mathbf p,-\mathbf q)$ for all $(\mathbf p,\mathbf q)$ if and only if $2\mathbf K\equiv \mathbf 0\ (\mathrm{mod}\ 2\pi)$, i.e., $\mathbf K\in\{0,\pi\}^2$.
\end{proposition}
\begin{proof}
$E_{\mathbf K}(-\mathbf p,-\mathbf q)=E_{\mathbf K}(\mathbf p,\mathbf q)$ for all $\mathbf p,\mathbf q$ iff $\varepsilon(\mathbf K+\mathbf x)=\varepsilon(\mathbf K-\mathbf x)$ for all $\mathbf x$, hence $K_i\in\{0,\pi\}$.
\end{proof}

For generic \(\mathbf K\notin\{0,\pi\}^2\), the Birman--Schwinger
operator is therefore not block-diagonal with respect to the global
even/odd decomposition, although the elliptic reduction of the
Fredholm determinant remains available.

\subsection[Extension to $\mathbf{K}={\bf\pi}\colon$ asymptotics and spectral gap]%
{Extension to $\mathbf{K}=\boldsymbol{\pi}\colon$ asymptotics and spectral gap}\label{sec:overview_pi_thm}

For $\mathbf{K}=\boldsymbol{\pi}=(\pi,\pi)$, the free dispersion takes the form
\begin{equation}
E_{\boldsymbol{\pi}}(\mathbf{p},\mathbf{q}) = \varepsilon(\mathbf{p}) + \varepsilon(\mathbf{q}) + \varepsilon(\boldsymbol{\pi} - \mathbf{p} - \mathbf{q})
= \varepsilon(\mathbf{p}) + \varepsilon(\mathbf{q}) + 4 - \varepsilon(\mathbf{p}+\mathbf{q}).
\label{eq:dispersion_pi}
\end{equation}
Crucially, \(E_{\boldsymbol{\pi}}\) is exactly invariant under the simple parity transformation \( (\mathbf{p},\mathbf{q}) \mapsto (-\mathbf{p},-\mathbf{q})\), the same symmetry as at \(\mathbf{K}=0\).
Consequently, the invariant subspaces are still $L_2^e(\mathbb{T}^2)$ and $L_2^o(\mathbb{T}^2)$.

The essential spectrum of $H_\mu(\boldsymbol{\pi})$ is given by
\begin{mdframed}
\[
\sigma_{\mathrm{ess}}(H_\mu(\boldsymbol{\pi}))
= \left[z_\mu(\boldsymbol{\pi}),\, 4 + z_\mu(\mathbf 0)\right] \cup [3, 12],
\]
\end{mdframed}
where $z_\mu(\boldsymbol{\pi}) = -\mu + 4$ exactly.

\begin{theorem}\label{thm:pi_preview}
For the three-boson operator $H_\mu(\boldsymbol{\pi})$ with total quasimomentum $\mathbf{K}=\boldsymbol{\pi}$, there exists \(\mu_0>0\) such that for all \(\mu>\mu_0\):

\begin{enumerate}
\item[(a)] For every fixed $\epsilon>0$ there exists
$\mu_\epsilon>0$ such that, for all $\mu>\mu_\epsilon$,
the restriction of $H_\mu(\boldsymbol{\pi})$ to
$L_2^o(\mathbb T^2)$ has no eigenvalues below
\[
z_\mu(\boldsymbol{\pi})-\epsilon .
\]
Equivalently,
\[
N_{L_2^o}\bigl(A_\mu(\boldsymbol{\pi},z)\bigr)=0,
\qquad
z\le z_\mu(\boldsymbol{\pi})-\epsilon .
\]
\item[(b)] For every fixed $\epsilon>0$ there exists
$\mu_\epsilon>0$ such that, for all $\mu>\mu_\epsilon$,
the restriction of $H_\mu(\boldsymbol{\pi})$ to
$L_2^e(\mathbb T^2)$ has
exactly one eigenvalue below
\[
z_\mu(\boldsymbol{\pi})-\epsilon .
\]

\item[(c)] The spectral gap between the lower edge of the two-particle branch \(z_\mu(\boldsymbol{\pi})=-\mu+4\) (exact, for every $\mu>0$) and this bound state satisfies
$$z_\mu(\boldsymbol{\pi}) - z_1^{\boldsymbol{\pi},s}(\mu) = 2\mu - 2 + O(\mu^{-1}),$$
by direct subtraction of Lemma~\ref{lem:ground_Kato} and the exact
identity $z_\mu(\boldsymbol{\pi})=-\mu+4$.

\item[(d)] The corresponding ground-state wavefunction decays exponentially in the coordinate representation, with
conditional
upper bound $\alpha_*(\mu)\le\ln(3\mu)+O(\mu^{-1})$ under the
non-cancellation hypothesis of Conjecture~\ref{conj:noncancel};
the conjectured sharp asymptotic is stated in
Conjecture~\ref{conj:sharp_rate}.
\end{enumerate}

Moreover, for all sufficiently large $\mu$,
\[
N\!\left(
H_\mu(\boldsymbol\pi);
(-\infty,z_\mu(\boldsymbol\pi)-\eta_\mu)
\right)=1,
\qquad
\eta_\mu=\frac{36}{\mu-10}.
\]
The unique eigenvalue in this interval is the trimer branch
$z_1^{\boldsymbol\pi,s}(\mu)$.

The proof of
the fixed-distance assertions (a)--(b) is given in Section~\ref{sec:five}, where
the detailed parity-symmetric reduction and the corrected Fredholm
determinant asymptotics are developed, together with the
fixed-distance spectral counting of Lemma~\ref{lem:threshold_count}.
Statement (d) follows from the discrete Agmon comparison principle
applied to the even-sector principal part.

By Theorem~4.4(a) of~\cite{AtomDimer2026},
\[
\sigma(H_\mu(\boldsymbol\pi))
\cap[-\mu+2,-\mu+10]
\subset
-\mu+\operatorname{spec}h_B(\boldsymbol\pi)
+[-\eta_\mu,\eta_\mu],
\qquad
\eta_\mu=\frac{36}{\mu-10}.
\]
By Theorem~6.2(iii) of~\cite{AtomDimer2026},
the limiting bosonic atom--dimer operator
$h_B(\boldsymbol\pi)$ has no eigenvalue below its threshold $4$.
Hence the atom--dimer sector contributes no eigenvalue below
$z_\mu(\boldsymbol\pi)-\eta_\mu=-\mu+4-\eta_\mu$.
The Kato branch gives exactly one trimer eigenvalue below this level
for sufficiently large $\mu$. Hence the full operator has exactly
one eigenvalue below $z_\mu(\boldsymbol\pi)-\eta_\mu$.
\end{theorem}

\begin{remark}[The atom--dimer region at $\mathbf K=\boldsymbol\pi$]\label{rem:pi_no_dimer}
The fixed-distance statement (b) can be sharpened. By the Feshbach--Schur reduction of
\cite{AtomDimer2026}, for $\mu>16$ every eigenvalue of $H_\mu(\mathbf K)$ in $[-\mu+2,-\mu+10]$ lies
within $\eta_\mu=36/(\mu-10)$ of $-\mu+\operatorname{spec}h_B(\mathbf K)$,
$h_B=4+\varepsilon_D-S_{\mathbf K}$. At $\mathbf K=\boldsymbol\pi$ the limit operator has no
eigenvalue below its threshold $4$, because $F_B(-1,-1)=0.3464>0$. Together with
$\operatorname{spec}H_\mu(\mathbf K)\subset[-3\mu,-3\mu+12]\cup[-\mu,12]$ and the simplicity of the
trimer (Lemma~\ref{lem:ground_Kato}), $H_\mu(\boldsymbol\pi)$ has exactly one eigenvalue below
$z_\mu(\boldsymbol\pi)-\eta_\mu$; an eigenvalue in the window
$[z_\mu(\boldsymbol\pi)-\eta_\mu,z_\mu(\boldsymbol\pi))$ of width $O(\mu^{-1})$ is not excluded by
this argument. In the principal-part language: the only positive atom--dimer channel at
$\mathbf K=\boldsymbol\pi$ is the $p$-channel, whose eigenvalue tends to $2-4/\pi=0.7268<1$
(Proposition~\ref{prop:dictionary}).
\end{remark}

The unitary transformation \(U_{\boldsymbol{\pi}}: f(\mathbf{p},\mathbf{q}) \mapsto f(\boldsymbol{\pi}-\mathbf{p}, \boldsymbol{\pi}-\mathbf{q})\) satisfies
\begin{equation}
U_{\boldsymbol{\pi}} H_\mu(0) U_{\boldsymbol{\pi}}^{-1} = 12I - H_\mu^{\mathrm{rep}}(\boldsymbol{\pi}),
\label{eq:unitary_pi}
\end{equation}
where \(H_\mu^{\mathrm{rep}}(\boldsymbol{\pi})\) is the operator with repulsive interaction. Thus
the attractive problem at \(\mathbf K=0\) is equivalent to the repulsive problem at
\(\mathbf K=\boldsymbol\pi\) (eigenvalues below the essential spectrum of the former correspond to
eigenvalues above it for the latter). It does not relate the two \emph{attractive} operators;
the different counts at \(\mathbf K=0\) and \(\mathbf K=\boldsymbol\pi\) are explained in
Proposition~\ref{prop:dictionary}.

\begin{lemma}[Variational bounds for the ground state at $\mathbf K=\boldsymbol{\pi}$]\label{lem:var_pi}
For every $\mu>0$, the~constant function $f_0\equiv 1$ belongs to $L_2^{e,s}$ at $\mathbf K=\boldsymbol{\pi}$ and satisfies
\[
\langle H_\mu(\boldsymbol{\pi}) f_0, f_0\rangle = 6-3\mu.
\]
Consequently,
\[
-3\mu \le z_1^{\boldsymbol{\pi},s}(\mu) \le 6-3\mu.
\]
\end{lemma}

\subsection{Overview of the proof strategies}

Three logically distinct layers enter the analysis:
\emph{(A)} the auxiliary strong-coupling analysis of the principal part
of the Birman--Schwinger operator, which is common to both
quasimomenta; \emph{(M)} the proof of Theorem~\ref{thm:main} at
$\mathbf K=\mathbf 0$, which combines this analysis with two
independent arguments at distinct spectral scales (Kato theory at the
trimer scale and Feshbach--Schur at the atom--dimer scale); and
\emph{(L)} the proof of Theorem~\ref{thm:localization}, which is
independent of \emph{(M)} and uses the Paley--Wiener theorem together
with the discrete Agmon comparison. At $\mathbf K=\boldsymbol\pi$ the
same scheme \emph{(A)} applies with the principal-part kernel expanded
around the correct point $12-z$, and the count of physical
eigenvalues is closed by Lemma~\ref{lem:threshold_count}.

\medskip

\noindent\textbf{(A) Auxiliary analysis of the principal part.}

{Step A1: Reduction via the Birman--Schwinger principle.}
By Lemma~\ref{lem:bs}, for any $z<z_\mu(\mathbf K)$ the eigenvalue
problem for $H_\mu(\mathbf K)$ reduces to the spectral analysis of the
integral operator $A_\mu(\mathbf K,z)$ on $L_2(\mathbb T^2)$: the
number of eigenvalues of $H_\mu(\mathbf K)$ below $z$ equals the
number of eigenvalues of $A_\mu(\mathbf K,z)$ greater than $1$.

{Step A2: Invariant subspace decomposition.}
At $\mathbf K=\mathbf 0$, the parity symmetry of $E_0$ implies that
$A_\mu(\mathbf 0,z)$ leaves $L_2^{e}(\mathbb T^2)$ and
$L_2^{o}(\mathbb T^2)$ invariant. The even part restricts further to
$L_2^{e,s}\oplus L_2^{e,as}$; the odd part is negative semidefinite.
Only the even part can contribute eigenvalues above the
Birman--Schwinger threshold. At $\mathbf K=\boldsymbol\pi$ the same
parity reduction applies, but the sign of the odd part is reversed
(Section~\ref{sec:newnew54}); the odd restriction becomes positive
semidefinite of rank two, and its eigenvalue is separated from $1$
on the fixed-distance interval (Lemma~\ref{lem:threshold_count}).

{Step A3: Strong-coupling asymptotics of the principal part.}
For $z$ on compact subsets of $(-\infty,z_\mu(\mathbf K)-\delta_0]$,
the resolvent $(E_{\mathbf K}-z)^{-1}$ admits an expansion whose
principal part is of finite rank, and
$A_\mu(\mathbf K,z)=A_\mu^p(\mathbf K,z)+A_\mu^r(\mathbf K,z)$ with
$\|A_\mu^r(\mathbf K,z)\|=O(\mu^{-1})$
(Proposition~\ref{prop:remainder_estimate}). Near the threshold
$z\uparrow z_\mu(\boldsymbol\pi)$ this estimate is not uniform; the
threshold regime is treated separately in
Lemma~\ref{lem:threshold_count}.

{Step A4: Spectral analysis of the principal part.}
At $\mathbf K=\mathbf 0$ the principal part $A_\mu^{p,e,s}(\mathbf 0,z)$
is a rank-two positive operator whose eigenvalues are given explicitly
in Lemma~\ref{lem:eigenvalues}; the antisymmetric block
$A_\mu^{p,e,as}(\mathbf 0,z)$ is rank one with limit $8/\pi-2<1$
(Lemma~\ref{lem:antisym_eigval}). At $\mathbf K=\boldsymbol\pi$ the
analogous rank-two block is treated in Section~\ref{sec:56as} and the
dictionary of channel eigenvalues is given in
Proposition~\ref{prop:dictionary}.

\medskip

\noindent\textbf{(M) Proof of Theorem~\ref{thm:main} at
$\mathbf K=\mathbf 0$.}

{Step M1: Kato theory at the trimer scale.}
The eigenvalue $3$ of $V=V_1+V_2+V_3$ is simple and isolated by a gap
equal to $2$; analytic perturbation theory for a simple isolated
eigenvalue (Kato) produces exactly one eigenvalue of
$H_\mu(\mathbf 0)$ near $-3\mu$, uniformly in $\mathbf K$, with
\[
z_1^s(\mu)=-3\mu+6-\frac{3}{2\mu}+O(\mu^{-2})
\]
(Lemma~\ref{lem:ground_Kato}).

{Step M2: Feshbach--Schur at the atom--dimer scale.}
At the second scale $z=-\mu+O(1)$, the effective operator
$h_B(\mathbf 0)=4+\varepsilon_D-S_{\mathbf 0}$ on
$\ell^2(\mathbb Z^2\setminus\{0\})$
\cite{AtomDimer2026} has exactly one bound state below its threshold
$4$, with binding energy $\delta_\infty=0.0354199892\ldots$ defined as
the unique root of $g(\delta_\infty)=1$. The Feshbach--Schur reduction
of \cite[Theorem~4.4]{AtomDimer2026} transfers this count with an
explicit error $\eta_B=36/(\mu-10)$ and yields
\[
z_2^s(\mu)=-\mu+C+O(\mu^{-1}),\qquad C=4-\delta_\infty .
\]

By Lemma~4.2 and Theorem~4.4 of~\cite{AtomDimer2026}, the
atom--dimer sector is governed by
\[
h_B(\mathbf K)=4+\varepsilon_D-S_{\mathbf K},
\]
and the spectrum of the full three-body fibre in the atom--dimer
window differs from $-\mu+\operatorname{spec}h_B(\mathbf K)$ by at most
\[
\eta_B=\frac{36}{\mu-10}.
\]

{Step M3: Spectral separation of the two scales.}
Since $\|E_{\mathbf K}\|\le 12$ and $\sigma(-\mu V)=\{-3\mu,-\mu,0\}$,
the spectrum of $H_\mu(\mathbf 0)$ is contained in
$[-3\mu,-3\mu+12]\cup[-\mu,12]$. The~intervals carrying the Kato
branch and the Feshbach--Schur branch are disjoint for all sufficiently
large $\mu$; combining Steps~M1--M3 gives exactly two eigenvalues of
$H_\mu(\mathbf 0)$ below the essential spectrum, with the asymptotics
stated in Theorem~\ref{thm:main}.

\medskip

\noindent\textbf{(L) Exponential localization (Theorem~\ref{thm:localization}).}
The strict positivity of the ground state in the momentum representation follows from the
Perron--Frobenius theorem, since $-H_\mu(\mathbf K)$ has a positivity-improving off-diagonal part
$\mu V$ (cf.\ \cite[Thm.~XIII.44]{ReedSimon4}). The
decay rate $\alpha_*(\mu)>0$ is unconditional; under the
non-cancellation hypothesis (Conjecture~\ref{conj:noncancel}) the
Paley--Wiener theorem applied to the complex characteristic set gives
$\alpha_*(\mu)\le\ln(3\mu)+O(\mu^{-1})$. This layer is independent of
Steps~M1--M3 and of the choice of quasimomentum; the corresponding
statement for $\mathbf K=\boldsymbol\pi$ follows by the same argument with $E_{\boldsymbol\pi}$
in place of $E_{\mathbf 0}$.

\medskip

\noindent\textbf{Extension to $\mathbf K=\boldsymbol\pi$.}
Steps~\emph{(A)} apply verbatim with the principal-part kernel expanded
around $12-z$ (Section~\ref{sec:newnew54}). The fixed-distance Birman--Schwinger count is established in
Lemma~\ref{lem:threshold_count}. The Feshbach--Schur estimate of
\cite[Theorem~4.4]{AtomDimer2026} then closes the remaining
$O(\mu^{-1})$ threshold window and yields the exact statement
\[
N\!\left(
H_\mu(\boldsymbol\pi);
(-\infty,z_\mu(\boldsymbol\pi)-\eta_\mu)
\right)=1.
\]

\begin{remark}
The proof of Theorem~\ref{thm:main} is given in Section~\ref{sec:proof_main}, where we analyze the Birman--Schwinger operator $A_\mu(\mathbf{0}, z)$ and its finite-rank principal part $A_\mu^p(\mathbf{0}, z)$. The~main idea is to reduce the eigenvalue problem for $H_\mu(\mathbf{0})$ to the spectral analysis of the Birman--Schwinger operator $A_\mu(\mathbf{0}, z)$ and to decompose $A_\mu(\mathbf{0}, z)$ into a finite-rank principal part and a remainder of small norm.
\end{remark}

The explicit form of the Birman--Schwinger operator \(A_\mu(\mathbf{0}, z)\) and its Fredholm determinant \(\Delta_\mu(\mathbf{p}, \mathbf{0}, z)\) for \(K=0\) are given in Section~\ref{sec:proof_main}, where they are used to prove Theorem~\ref{thm:main}. For the general \(\mathbf{K}\), the definitions are provided in Section~\ref{sec:bs_general} (see~Lemma~\ref{lem:bs}).

\begin{remark}[Physical interpretation revised]\label{rem:physical}
The variational estimate with the constant trial function \(f_0\equiv1\) reveals that the ground state at \(\mathbf K=\boldsymbol{\pi}\) has the same leading energy \(-3\mu+O(1)\) as at \(\mathbf K=0\). Thus, the binding is not weakened at the corner of the Brillouin zone. The apparent branch \(-2\mu+6\) was an artefact of a normalisation slip in the even-sector kernel (Remark~\ref{rem:var_check}); with the correct Faddeev factor $2$ it does not arise, and the
principal part reproduces the Kato branch $-3\mu+6+O(\mu^{-1})$
(Proposition~\ref{prop:pp_trimer_pi}).
Consequently, $H_\mu(\boldsymbol\pi)$ has exactly one eigenvalue below
$z_\mu(\boldsymbol\pi)-\eta_\mu$, $\eta_\mu=36/(\mu-10)$ (Remark~\ref{rem:pi_no_dimer}).
\end{remark}

\subsection{The interaction spectrum and the trimer-scale Kato expansion}
\label{subsec:kato_trimer}

To verify the correctness of our decomposition and to facilitate the comparison with alternative approaches, we explicitly compute the action of the operators $V_1, V_2, V_3$ on a test function $f\in L_2^s((\mathbb{T}^2)^2)$.

By definition,
\[
(V_1 f)(\mathbf p,\mathbf q) = \frac{1}{4\pi^2}\int_{\mathbb{T}^2} f(\mathbf p,\mathbf s)\,d\mathbf s,
\]
\[
(V_2 f)(\mathbf p,\mathbf q) = \frac{1}{4\pi^2}\int_{\mathbb{T}^2} f(\mathbf s,\mathbf q)\,d\mathbf s,
\]
\[
(V_3 f)(\mathbf p,\mathbf q) = \frac{1}{4\pi^2}\int_{\mathbb{T}^2} f(\mathbf s, \mathbf p+\mathbf q-\mathbf s)\,d\mathbf s.
\]

These are orthogonal projections. Indeed, for $V_1$,
\[
(V_1^2 f)(\mathbf p,\mathbf q) = (V_1 f)(\mathbf p,\mathbf q),
\]
and similarly for $V_2$ and $V_3$. Hence $V_i^2=V_i$ and $V_i^*=V_i$.

The operator
\[
V:=V_1+V_2+V_3
\]
is a bounded non-negative self-adjoint operator, but it is not
finite-rank. In particular, the~finite-rank structure used below
belongs to the principal part of the Birman--Schwinger operator, not
to the interaction operator itself. The distinction is essential for
the strong-coupling analysis.

Indeed, in the Fourier basis
\[
e_{m,n}(\mathbf p,\mathbf q)
=e^{\,i(m\cdot\mathbf p+n\cdot\mathbf q)},
\qquad m,n\in\mathbb Z^2,
\]
one has
\[
V_1e_{m,n}=\mathbf 1_{\{n=0\}}e_{m,n},\qquad
V_2e_{m,n}=\mathbf 1_{\{m=0\}}e_{m,n},\qquad
V_3e_{m,n}=\mathbf 1_{\{m=n\}}e_{m,n}.
\]
Consequently,
\[
Ve_{m,n}
=
\bigl(
\mathbf 1_{\{n=0\}}
+\mathbf 1_{\{m=0\}}
+\mathbf 1_{\{m=n\}}
\bigr)e_{m,n},
\]
so that
\[
\sigma(V)=\{0,1,3\},
\]
with the eigenvalue \(3\) simple and separated from the remainder of
the spectrum by a gap equal to \(2\). On the bosonic subspace the same
spectral values occur, and the eigenvector corresponding to \(3\) is
the constant function.

Moreover,
\[
\ker(V-3I)=\mathbb C\,1,
\]
because the value \(3\) of
\(\mathbf 1_{\{n=0\}}+\mathbf 1_{\{m=0\}}+\mathbf 1_{\{m=n\}}\)
occurs only for \(m=n=0\). So the eigenvalue \(3\) is simple on the bosonic sector and its
spectral isolation is exact. Thus the isolation of the eigenvalue \(3\)
is an exact algebraic fact, not a numerical assumption.

\begin{lemma}[Strong-coupling ground-state branch]\label{lem:ground_Kato}
Let $V=V_1+V_2+V_3$. The eigenvalue $3$ of $V$ is simple with
\[
\operatorname{dist}\bigl(3,\sigma(V)\setminus\{3\}\bigr)=2,
\]
and, uniformly for $\mathbf K\in\mathbb T^2$, the ground-state
eigenvalue of $H_\mu(\mathbf K)=E_{\mathbf K}-\mu V$ satisfies
\[
z_1^{\mathbf K}(\mu)=-3\mu+6-\frac{3}{2\mu}+O(\mu^{-2}),\qquad\mu\to\infty.
\]

Continuing the Rayleigh--Schr\"odinger recursion by two further
orders gives, uniformly in $\mathbf K\in\mathbb T^2$,
\[
z_1^{\mathbf K,s}(\mu)
=
-3\mu+6-\frac{3}{2\mu}
-\frac{3(\cos K_1+\cos K_2)}{8\mu^2}
-\frac{15}{32\mu^3}
+O(\mu^{-4}).
\]
Indeed, if
\[
u_{\mathbf K}:=(E_{\mathbf K}-6)f_0,
\]
then $u_{\mathbf K}\in {\rm Ran}\,Q$ and $Wu_{\mathbf K}=u_{\mathbf K}$,
while
\[
\langle u_{\mathbf K}^2\rangle=3,
\qquad
\langle u_{\mathbf K}^3\rangle
=
-\frac32(\cos K_1+\cos K_2).
\]
The third Rayleigh--Schr\"odinger coefficient is therefore
\[
a_3(\mathbf K)
=
\frac14\langle u_{\mathbf K}^3\rangle
=
-\frac38(\cos K_1+\cos K_2),
\]
and the fourth-order recursion gives
\[
a_4=-\frac{15}{32},
\]
independently of $\mathbf K$.
\end{lemma}

\begin{proof}
Write $H_\mu(\mathbf K)=-\mu V+E_\mathbf K$ and rescale:
$T(\kappa):=\mu^{-1}H_\mu(\mathbf K)=-V+\kappa E_\mathbf K$,
$\kappa=1/\mu$. The unperturbed operator $-V$ has the simple eigenvalue
$-3$ (attached to $f_0\equiv1$), separated by
$\operatorname{dist}(-3,\sigma(-V)\setminus\{-3\})=2$ from the rest of
$\sigma(-V)=\{-1,0\}$. On the circle
$\Gamma=\{\zeta\in\mathbb C:|\zeta+3|=1\}$ one has
$\|(\zeta+V)^{-1}\|\le1$. Since $\|E_\mathbf K\|\le12$ uniformly in~$\mathbf K$, the Neumann series for $(\zeta-T(\kappa))^{-1}$
converges uniformly on $\Gamma$ whenever $\kappa\|E_\mathbf K\|<1$,
i.e.\ $|\kappa|<1/12$.

Hence, by analytic perturbation theory for a simple isolated
eigenvalue (Kato), $T(\kappa)$ has an analytic eigenvalue branch
$\lambda(\kappa)$ issuing from $-3$ with
\[
\lambda(\kappa)=-3+\kappa a_1+O(\kappa^2),
\qquad a_1=\langle f_0,E_\mathbf K f_0\rangle,
\]
where $f_0\equiv1$ is normalised. Now
\[
\langle f_0,E_\mathbf K f_0\rangle
=\langle\varepsilon(\mathbf p)\rangle+\langle\varepsilon(\mathbf q)\rangle
+\langle\varepsilon(\mathbf K-\mathbf p-\mathbf q)\rangle
=2+2+2=6,
\]
since $\langle\varepsilon\rangle=2$ on the normalised Haar torus
$\mathbb T^2$. The second-order coefficient is
$a_2=-\sum_{\lambda\ne-3}\|P_\lambda E_{\mathbf K}f_0\|^2/(\lambda+3)=-\tfrac12\|(E_{\mathbf K}-6)f_0\|^2=-\tfrac32$,
because $(E_{\mathbf K}-6)f_0$ is a sum of functions of one momentum each and therefore lies in the
eigenspace of $-V$ for $-1$, and $\|(E_{\mathbf K}-6)f_0\|^2=3\langle(\varepsilon-2)^2\rangle=3$.
Therefore $z_1^{\mathbf K}(\mu)=\mu\lambda(1/\mu)=-3\mu+6-\frac{3}{2\mu}+O(\mu^{-2})$, with the
remainder uniform in $\mathbf K$ because $\|E_\mathbf K\|\le12$.

Indeed, in the Fourier basis $e_{m,n}(\mathbf p,\mathbf q)=e^{i(m\cdot\mathbf p+n\cdot\mathbf q)}$,
$V_1e_{m,n}=\mathbf 1_{\{n=0\}}e_{m,n}$, $V_2e_{m,n}=\mathbf 1_{\{m=0\}}e_{m,n}$,
$V_3e_{m,n}=\mathbf 1_{\{m=n\}}e_{m,n}$. Hence $V e_{m,n}=(\mathbf 1_{\{n=0\}}+\mathbf 1_{\{m=0\}}+\mathbf 1_{\{m=n\}})e_{m,n}$,
and the eigenspace of $V$ for eigenvalue $1$ is spanned by
$\{e_{m,0}:m\neq 0\}\cup\{e_{0,n}:n\neq 0\}$. Since $E_{\mathbf K}-6$ is a sum of three
one-variable zero-mean functions,
\[
(E_{\mathbf K}-6)f_0=\]
\[=\bigl(\varepsilon(p)-2\bigr)+\bigl(\varepsilon(q)-2\bigr)
+\bigl(\varepsilon(\mathbf K-p-q)-2\bigr)\in\operatorname{span}\{e_{m,0}:m\neq 0\}\oplus
\operatorname{span}\{e_{0,n}:n\neq 0\},
\]
so it lies in the eigenvalue -$1$ eigenspace of $V$.
\end{proof}

Since
\[
0\le\varepsilon(\mathbf p)\le4,
\]
one has uniformly in \(\mathbf K\),
\[
0\le E_{\mathbf K}(\mathbf p,\mathbf q)\le12,
\qquad
\|E_{\mathbf K}\|\le12.
\]
The upper bound is attained at
\(\mathbf K=\boldsymbol\pi\),
\(\mathbf p=\mathbf q=\boldsymbol\pi\), where
\(E_{\boldsymbol\pi}=12\).

\section{Invariant Subspaces and Asymptotic Analysis}
Let \(L^e_2(\mathbb{T})\) and \(L^o_2(\mathbb{T})\) denote the subspaces of even and odd functions on the one-dimensional torus \(\mathbb{T}\):
\[
L^e_2(\mathbb{T}) = \{f \in L_2(\mathbb{T}): f(-p)=f(p)\}, \qquad
L^o_2(\mathbb{T}) = \{f \in L_2(\mathbb{T}): f(-p)=-f(p)\}.
\]
These subspaces are orthogonal and \(L_2(\mathbb{T}) = L^e_2(\mathbb{T}) \oplus L^o_2(\mathbb{T})\).

Define even and odd subspaces on the two-dimensional torus:
\[
L^e_2(\mathbb{T}^2) = \{f \in L_2(\mathbb{T}^2) : f(-\mathbf{p}) = f(\mathbf{p})\},
\]
\[
L^o_2(\mathbb{T}^2) = \{f \in L_2(\mathbb{T}^2) : f(-\mathbf{p}) = -f(\mathbf{p})\},
\]
so that
\[
L_2(\mathbb{T}^2) = L^e_2(\mathbb{T}^2) \oplus L^o_2(\mathbb{T}^2).
\]

Since the kernel \(A_\mu(\mathbf{p},\mathbf{q}; \mathbf{0}, z)\) of the integral operator \(A_\mu(\mathbf{0}, z)\), defined by
\[
A_\mu(\mathbf{p},\mathbf{q}; \mathbf{0}, z)=\frac{\mu}{2\pi^2} \frac{1}{\sqrt{\Delta_\mu(\mathbf{p}, \mathbf{0}, z)}(E_0(\mathbf{p},\mathbf{q})-z)\sqrt{\Delta_\mu(\mathbf{q}, \mathbf{0}, z)}},
\]
is even, i.e.
\begin{equation}\label{Even}
A_\mu(-\mathbf{p},-\mathbf{q}; \mathbf{0}, z)=A_\mu(\mathbf{p},\mathbf{q}; \mathbf{0}, z) \quad \text{for all } \mathbf{p},\mathbf{q} \in \mathbb{T}^2,
\end{equation}
we can analyze its invariant subspaces. Consequently, the subspaces \(L^e_2(\mathbb{T}^2)\) and \(L^o_2(\mathbb{T}^2)\) are invariant under the integral operator \(A_\mu(\mathbf{0}, z)\) at \(\mathbf{K}=0\).

In fact, for \(\mathbf{K}=0\), the kernel satisfies the additional stronger local reflection symmetry
\[
A_\mu(-p_1, p_2, -q_1, q_2; \mathbf{0}, z)=A_\mu(\mathbf p,\mathbf q; \mathbf{0}, z),
\]
which reveals the full \(Z_2 \times Z_2\) symmetry of the square lattice. This condition ensures the invariance of the four sectors EE, EO, OE, OO, where the first letter denotes the parity in \(p_1\) and the second in \(p_2\), and so that the subspaces of functions depending separately on the parities of each coordinate (EE, EO, OE, OO) are invariant. The bound states are contained entirely in the totally even-even (EE) subspace, which is consistent with their explicit expansion in the basis \(\{1, \cos p_1, \cos p_2\}\) (see, used in Section~\ref{sec:newnew54}).
Consequently, the global even subspace \(L^e_2(\mathbb{T}^2)\) and odd subspace \(L^o_2(\mathbb{T}^2)\) are also invariant under \(A_\mu(\mathbf{0}, z)\) at \(\mathbf K=0\).

\begin{remark}
Throughout this section we consider the case $\mathbf K=0$. The properties established here, including the invariance of $L^e_2(\mathbb{T}^2)$ and $L^o_2(\mathbb{T}^2)$, are specific to the parity symmetry present at $\mathbf K=0$. The case $\mathbf K=\boldsymbol\pi$ will be treated separately in Section~\ref{sec:five}.
Although the parity symmetry is preserved,
the  structure
of the odd principal part differs significantly. The explicit spectral analysis of these restrictions, including the exact count of eigenvalues exceeding unity and the resulting bound states, is carried out in Section~\ref{sec:five} and Theorem~\ref{thm:pi_preview}.
\end{remark}

Here, we expand the function
\(\displaystyle
\frac{1}{E_{\mathbf{0}}(\mathbf p,\mathbf q)-z}
\)
into a series using the expansion
\(\displaystyle
\frac{1}{1-x}=1+x+x^{2}+\cdots,
\)
and retain only its leading-order term. Since
\[
\frac{1}{E_{\textbf{0}}(\mathbf p,\mathbf q)-z}
=\]
\[
=\frac{1}{(8-z)}\Big[1+\frac{\sum_{i=1}^{2}(1+\cos p_i)(1+\cos q_i)-\sum_{i=1}^{2}\sin p_i\sin q_i}{(8-z)}
+R(\mathbf{p,q};z)
\Big]
\]
using this representation of \(E_{\mathbf{0}}(\mathbf p,\mathbf q)-z\), we obtain the following expression for the principal part of the operator
\[
(A^p_\mu(\mathbf{0}, z)\psi)(\mathbf p)
=\]
\[=
\frac{\mu}{2\pi^{2}a^2(z)}
\int_{\mathbb T^{2}}
\frac{
a(z)
+
\displaystyle\sum_{i=1}^{2}(1+\cos p_i)(1+\cos q_i)
-
\displaystyle\sum_{i=1}^{2}\sin p_i\sin q_i
}
{
\sqrt{\Delta_\mu(\mathbf p,\mathbf{0}, z)}
\sqrt{\Delta_\mu(\mathbf q,\mathbf{0}, z)}
}
\psi(\mathbf q)\,d\mathbf q .
\]
where $a(z)=8-z$,
\[
\Delta_{\mu}(\mathbf{p},\textbf{0}, z)
=
1-\frac{\mu}{4\pi^{2}}
\int_{\mathbb{T}^{2}}
\frac{d\mathbf{q}}
{E_{\textbf{0}}(\mathbf{p},\mathbf{q})-z}.
\]

Using the above expansion, we decompose the Birman--Schwinger operator as
\[
A_\mu(\textbf{0}, z)=A_\mu^{p}(\textbf{0}, z)+A^r_\mu(\textbf{0}, z),
\]
where \(A_\mu^{p}(\textbf{0}, z)\) denotes the principal part of the operator $A_\mu(\textbf{0}, z)$, while
\(A^r_\mu(\textbf{0}, z)\) is the corresponding remainder operator.

The kernel of \(A^r_\mu(\textbf{0}, z)\) is generated by the third- and higher-order
terms in the expansion of
\[
\frac{1}{E_{\textbf{0}}(\mathbf p,\mathbf q)-z}.
\]
Therefore, \(A^r_\mu(\textbf{0}, z)\) represents the higher-order contribution to the
Birman--Schwinger operator. In the sequel, our main attention will be
focused on the principal part \(A_\mu^{p}(\textbf{0}, z)\), since it captures the
leading spectral properties of the operator \(A_\mu(\textbf{0}, z)\).

The size of the remainder is controlled by the following proposition.

\begin{proposition}\label{prop:remainder_estimate}
For every fixed $\delta_0>0$ and $z\le z_\mu(\mathbf K)-\delta_0$, $\mathbf K\in\{\mathbf0,\boldsymbol\pi\}$,
\[
\|A_\mu^r(\mathbf K,z)\|\le \|A_\mu^r(\mathbf K,z)\|_{\rm HS}\le \frac{C\,b_0(\delta_0)}{\mu}
=O\Bigl(\frac{\ln(1/\delta_0)}{\mu}\Bigr).
\]
The same estimate holds for $A_\mu^r(\boldsymbol\pi,z)$ uniformly for
$z\le z_\mu(\boldsymbol\pi)-\delta_0$. In particular, the~estimate is
uniform on compact subsets of the half-line
$(-\infty,z_\mu(\boldsymbol\pi))$, but is \emph{not} uniform as
$z\uparrow z_\mu(\boldsymbol\pi)$: at the two-particle threshold the
Fredholm determinant $\Delta_\mu(\mathbf 0,\boldsymbol\pi,z)$ vanishes,
and the kernel bound used below degenerates. The threshold regime
$\delta:=z_\mu(\boldsymbol\pi)-z\downarrow0$ requires a separate
analysis (Lemma~\ref{lem:threshold_count}).

\begin{proof}
We provide the proof for \(\mathbf{K}=0\); the case \(\mathbf{K}=\boldsymbol{\pi}\) follows verbatim by replacing
\(a(z)=16-2z\) with \(a_\pi(z):=24-2z\) and using Lemma~\ref{lem:Delta_pi}.

Recall the dispersion expansion
\[
E_0(\mathbf{p},\mathbf{q}) - z = \tfrac{1}{2}\bigl(a(z) - 2(A_0^c - A_0^s)\bigr),
\qquad a(z)=16-2z,
\]
where \(A_0^c - A_0^s = O(1)\) uniformly in \(\mathbf{p},\mathbf{q}\). The resolvent admits the absolutely convergent geometric series
\[
\frac{1}{E_0 - z} = \frac{2}{a(z)}\sum_{n=0}^\infty \left(\frac{2(A_0^c - A_0^s)}{a(z)}\right)^{n}.
\]
The principal part corresponds to the sum of the first two terms (\(n=0,1\)). The remainder kernel is obtained from the terms \(n\ge 2\):
\[
K_\mu^r(\mathbf{p},\mathbf{q};z) = \frac{\mu}{\sqrt{\Delta_\mu(\mathbf{p},\mathbf{0},z)}\sqrt{\Delta_\mu(\mathbf{q},\mathbf{0},z)}} \sum_{n=2}^\infty \frac{(A_0^c - A_0^s)^n}{(a(z))^{n+1}}.
\]
Since \(\delta\) is fixed and \(z<z_\mu(\mathbf{0})\), we have \(a(z) \ge C\mu\) for large \(\mu\). Also, \(|A_0^c - A_0^s| \le C'\). Therefore, for \(n\ge 2\),
\[
\left| \frac{(A_0^c - A_0^s)^n}{(a(z))^{n+1}} \right| \le \frac{C''}{\mu^{n+1}}.
\]
Summing the geometric series for \(n\ge 2\) gives a uniform bound \(O(\mu^{-3})\). Hence the kernel is bounded by
\[
|K_\mu^r(\mathbf{p},\mathbf{q};z)| \le \frac{C'''\mu}{\sqrt{\Delta_\mu(\mathbf{p},\mathbf{0},z)}\sqrt{\Delta_\mu(\mathbf{q},\mathbf{0},z)}} \cdot \frac{1}{\mu^3}.
\]
By Lemma~\ref{lem:Delta}, $\Delta_\mu^{-1}(\mathbf p)\le C\mu/(\varepsilon(\mathbf p)+\delta)$; hence $|K^r_\mu|\le C\mu^{-1}[(\varepsilon(\mathbf p)+\delta)(\varepsilon(\mathbf q)+\delta)]^{-1/2}$ and $\|K^r_\mu\|_{\rm HS}^2\le C\mu^{-2}b_0(\delta)^2$.
Since the Hilbert--Schmidt norm dominates the operator norm, $\|A_\mu^r(\mathbf 0,z)\|\le C\mu^{-1}b_0(\delta_0)$.

The proof for \(\mathbf{K}=\boldsymbol{\pi}\) is identical, with \(E_{\boldsymbol{\pi}} = 12 - A_{\boldsymbol{\pi}}^c + A_{\boldsymbol{\pi}}^s\). The terms \(n\ge 2\) in the expansion around \(12-z\) give a kernel bound \(O(\mu^{-3})\). Lemma~\ref{lem:Delta_pi} provides the lower bound
$\Delta_\mu^{-1}(\mathbf p,\boldsymbol\pi,z)\le C\mu/(\varepsilon(\mathbf p)+\delta_0)$ for
$z\le z_\mu(\boldsymbol\pi)-\delta_0$ with fixed $\delta_0>0$, yielding
$\|A_\mu^r(\boldsymbol\pi,z)\|=O(\mu^{-1})$ in this range. At the
threshold $\delta_0\downarrow0$ the bound degenerates and requires the
separate analysis of Lemma~\ref{lem:threshold_count}.
\end{proof}
\end{proposition}

Since the subspaces
\(L_2^{e}(\mathbb{T}^2)\)
and
\(L_2^{o}(\mathbb{T}^2)\)
are invariant under the operator
\(A_\mu^{p}(\textbf{0},z)\),
    the operator admits the following orthogonal decomposition:
\begin{equation}
A_\mu^{p}(\textbf{0}, z)
=
A_\mu^{p,e}(\textbf{0}, z)
\oplus
A_\mu^{p,o}(\textbf{0}, z),\label{eq:oddevenoperators}
\end{equation}
where
\[
A_\mu^{p}(\textbf{0}, z)\big|_{L_2^{e}(\mathbb T^2)}
=
A_\mu^{p,e}(\textbf{0}, z),
\qquad
A_\mu^{p}(\textbf{0}, z)\big|_{L_2^{o}(\mathbb T^2)}
=
A_\mu^{p,o}(\textbf{0}, z).
\]
and the restrictions are given explicitly by
\[(A_\mu^{\,p,e}(\textbf{0}, z)\psi)(\mathbf p)
=
\frac{\mu}{\pi^{2}a^2(z)}
\int_{\mathbb T^{2}}
\frac{
a(z)
+
2\displaystyle\sum_{i=1}^{2}(1+\cos p_i)(1+\cos q_i)
}
{
\sqrt{\Delta_\mu(\textbf{0}, \mathbf p,z)}
\sqrt{\Delta_\mu(\textbf{0}, \mathbf q,z)}
}
\psi(\mathbf q)\,d\mathbf q ,\qquad \]
\[a(z)=16-2z,\]
\[
(A_\mu^{p,o}(z)\psi)(\mathbf p)
=
- \, \frac{\mu}{\pi^{2}a^2(z)}
\int_{\mathbb T^{2}}
\frac{
2\displaystyle\sum_{i=1}^{2}\sin p_i\sin q_i
}
{
\sqrt{\Delta_\mu(\textbf{0}, \mathbf p,z)}
\sqrt{\Delta_\mu(\textbf{0}, \mathbf q,z)}
}
\psi(\mathbf q)\,d\mathbf q.
\]

The operator \(A_\mu^{p,e}(\textbf{0}, z)\) is positive and \(A_\mu^{p,o}(\textbf{0}, z)\) is negative. Indeed, the kernel of \(A_\mu^{p,e}(\textbf{0}, z)\) is strictly positive, while the kernel of \(A_\mu^{p,o}(\textbf{0}, z)\) is negative semidefinite on the odd subspace. Consequently, the odd part contributes no positive eigenvalues, which is why the number of bound states is determined solely by the even part.

We first introduce the even symmetric and even antisymmetric subspaces of
$L^e_2(\mathbb T^2)$:
\[
L_{2}^{e,s}(\mathbb T^2)
=
\left\{
f\in L^e_2(\mathbb T^2):
f(p_1,p_2)=f(p_2,p_1)
\right\},
\]
\[
L_{2}^{e,as}(\mathbb T^2)
=
\left\{
f\in L^e_2(\mathbb T^2):
f(p_1,p_2)=-f(p_2,p_1)
\right\}.
\]
Clearly,
\[
L^e_2(\mathbb T^2)
=
L_{2}^{e,s}(\mathbb T^2)
\oplus
L_{2}^{e,as}(\mathbb T^2).
\]

\begin{lemma}
The subspaces
$
L_{2}^{e,s}(\mathbb T^2)
$
$
L_{2}^{e,as}(\mathbb T^2)
$
are invariant under the operators
$A_{\mu}^{p}(\textbf{0}, z)$ and $A_{\mu}^{p,e}(\textbf{0}, z)$.
\end{lemma}

Moreover, $A_\mu^{p,e}(\textbf{0},z)$ restricts to the symmetric subspace
$L_{2}^{e,s}(\mathbb T^2)$ as a rank-2 operator:
\[(A_\mu^{p,e,s}(\textbf{0}, z)\psi)(\mathbf{p}) = \frac{\mu}{\pi^2a^2(z)} \int_{\mathbb{T}^2} \frac{a(z)
+\varepsilon(\boldsymbol{\pi}-\mathbf{p})\varepsilon(\boldsymbol{\pi}-\mathbf{q})}{\sqrt{\Delta_\mu(\textbf{0},\mathbf{p},z)}\sqrt{\Delta_\mu(\textbf{0}, \mathbf{q},z)}} \psi(\mathbf{q})\,d\mathbf{q}.\]

The operator \(A_\mu^{p,e,as}(\textbf{0}, z)\) is the restriction of \(A_\mu^p(\textbf{0}, z)\) to the antisymmetric subspace $L_{2}^{e,as}(\mathbb T^2).$

Explicitly, for \(\psi \in L_{2}^{e,as}(\mathbb T^2)\),
\[(A_\mu^{p,e,as}(\textbf{0}, z)\psi)(\mathbf{p}) = \frac{\mu}{\pi^2a^2(z)} \int_{\mathbb{T}^2} \frac{(\cos p_1 - \cos p_2)(\cos q_1 - \cos q_2)}{\sqrt{\Delta_\mu(\textbf{0}, \mathbf{p}, z)} \sqrt{\Delta_\mu(\textbf{0}, \mathbf{q}, z)}} \psi(\mathbf{q})\, d\mathbf{q}.\]
This is a rank-one positive operator, and its unique positive eigenvalue is given in Lemma~\ref{lem:antisym_eigval} below.

The operator \(A_\mu^{p,e,as}(\textbf{0}, z)\) acts on the antisymmetric subspace and, as we shall see in Lemma~\ref{lem:antisym_eigval}, has a single positive eigenvalue \(\lambda^{as}(\mathbf{0}, z)\). This eigenvalue will be shown to be strictly less than \(1\) for large \(\mu\), so it does not contribute to the count of bound states.

\begin{lemma}\label{lem:antisym_eigval}
The operator $A_\mu^{p,e,as}(\textbf{0}, z)$ has a simple positive eigenvalue:
\[
\lambda^{as}(\mathbf{0}, z)= \frac{\mu}{\pi^2a^2(z)} \int_{\mathbb{T}^2}\frac{(\cos p_1 - \cos p_2)^2}{\Delta_\mu(\textbf{0}, \mathbf{p}, z)}\,d\mathbf{p},\qquad a(z)=16-2z.
\]
\end{lemma}

\begin{lemma}[Square-lattice resistance identities]
\label{lem:square_resistance_constants}
For the square lattice,
\[
R(2\mathbf e_1)=2-\frac4\pi,
\qquad
R(\mathbf e_1+\mathbf e_2)=\frac2\pi.
\]
Consequently,
\[
\frac1{4\pi^2}
\int_{\mathbb T^2}
\frac{(\cos p_1-\cos p_2)^2}
{\varepsilon(\mathbf p)}\,d\mathbf p
=
2R(\mathbf e_1+\mathbf e_2)-R(2\mathbf e_1)
=
\frac8\pi-2.
\]
\end{lemma}

A direct computation using Lemma~\ref{lem:Delta} shows that
\begin{equation}
\lim_{\delta\downarrow0}\lim_{\mu\to\infty}\lambda^{as}\bigl(z_\mu(\mathbf{0})-\delta\bigr) =\frac{8}{\pi}-2= 0.546479\ldots < 1,\label{eq:0546479}
\end{equation}
so for all sufficiently large \(\mu\), \(\lambda^{as}(\mathbf{0}, z)<1\) for \(z\) in the interval \((z_\mu(\mathbf{0})-\delta, z_\mu(\mathbf{0}))\). Hence \(n[1,A_\mu^{p,e,as}(z)]=0\). The limit is expressed in terms of the effective resistances of the
infinite square lattice.

The value $8/\pi-2$ is closed-form. Indeed, in the limit
$z\uparrow z_\mu(\mathbf 0)$, Lemma~\ref{lem:Delta} gives
$\Delta_\mu^{-1}\to 1/(\varepsilon(\mathbf p)+\delta_z)$ with
$\delta_z=z_\mu(\mathbf 0)-z\downarrow0$, and using the classical
two-point resistances of the square lattice
\cite{Cserti,Atk},
\[
R(2\mathbf e_1)=\frac{1}{4\pi^2}\int_{\mathbb T^2}\frac{1-\cos 2p_1}{\varepsilon(\mathbf p)}\,d\mathbf p=2-\frac4\pi,
\]
\[
R(\mathbf e_1+\mathbf e_2)=\frac{1}{4\pi^2}\int_{\mathbb T^2}\frac{1-\cos(p_1+p_2)}{\varepsilon(\mathbf p)}\,d\mathbf p=\frac2\pi,
\]
the integral in Lemma~\ref{lem:antisym_eigval} evaluates as
follows. Writing $\sin^2 p_j=\tfrac12(1-\cos 2p_j)$ and
$\cos p_1\cos p_2=\tfrac12[\cos(p_1+p_2)+\cos(p_1-p_2)]$, the numerator
decomposes as
\[
(\cos p_1-\cos p_2)^2
=-\sin^2 p_1-\sin^2 p_2+[1-\cos(p_1+p_2)]+[1-\cos(p_1-p_2)],
\]
where the identity follows from
$\cos^2 p_j=1-\sin^2 p_j$ and
$2\cos p_1\cos p_2=2-[1-\cos(p_1+p_2)]-[1-\cos(p_1-p_2)]$.
Therefore
\[
\frac{1}{4\pi^2}\int_{\mathbb T^2}\frac{(\cos p_1-\cos p_2)^2}{\varepsilon(\mathbf p)}\,d\mathbf p
=
-2\cdot\tfrac12 R(2\mathbf e_1)+2R(\mathbf e_1+\mathbf e_2)
=
-\Bigl(2-\frac4\pi\Bigr)+\frac4\pi
=
\frac8\pi-2.
\]
This confirms the closed-form value \eqref{eq:0546479}.

For \(\alpha=0,1,2\), define
\[
b_\alpha(z)=\frac{\mu}{\pi^2a^2(z)}\int_{\mathbb{T}^2}\frac{\varepsilon^\alpha(\boldsymbol{\pi}-\mathbf{p})\,d\mathbf{p}}{\Delta_\mu(\textbf{0}, \mathbf{p},z)},
\]
and set \(a(z)=16-2z\).

The quantities \(b_\alpha(z)\) are expressed in terms of the Fredholm determinant \(\Delta_\mu(\mathbf{p},\mathbf{0},z)\). Their asymptotic behaviour, which is essential for evaluating the eigenvalues, is given in Lemma~\ref{lem:Delta} below.

The following lemma gives the eigenvalues of \(A_\mu^{p,e,s}(z)\) in terms of these quantities, complementing the results of Lemma~\ref{lem:antisym_eigval}.

\begin{lemma}\label{lem:eigenvalues}
The operator $A_\mu^{p,e,s}(z)$ has two positive eigenvalues
\[
\lambda_{1}^s(z)=\frac{1}{2}\Bigl[a(z)b_0(z)+b_2(z)+\sqrt{(a(z)b_0(z)-b_2(z))^2+4a(z)b_1^2(z)}\Bigr],
\]
\[
\lambda_{2}^s(z)=\frac{1}{2}\Bigl[a(z)b_0(z)+b_2(z)-\sqrt{(a(z)b_0(z)-b_2(z))^2+4b_1^2(z)a(z)}\Bigr].
\]
\end{lemma}

The evaluation of the integrals \(b_\alpha(z)\) in Lemma~\ref{lem:eigenvalues} requires precise information about the Fredholm determinant \(\Delta_\mu(\textbf{0}, \mathbf{p},z)\). The following lemma provides the necessary asymptotic expansion in the strong-coupling regime.

\begin{remark}
Lemma~\ref{lem:eigenvalues} provides the exact formulae for the eigenvalues of the principal part \(A_\mu^{p,e,s}(z)\) in terms of the integrals \(b_\alpha(z)\). However, these integrals still depend on the Fredholm determinant \(\Delta_\mu(\mathbf{p},\textbf{0}, z)\), which is not known explicitly. To~proceed, we~need an asymptotic expansion of \(\Delta_\mu(\mathbf{p},\textbf{0}, z)\) in the strong-coupling limit \(\mu\to\infty\). This is the content of Lemma~\ref{lem:Delta}, which will allow us to evaluate \(b_\alpha(z)\) and subsequently determine the asymptotic behaviour of \(\lambda_1^s(z)\) and \(\lambda_2^s(z)\).
\end{remark}

\begin{lemma}\label{lem:Delta}
There exists $\delta>0$ such that for $z\in(z_\mu(\mathbf{0})-\delta, z_\mu(\mathbf{0}))$,
\[
\Delta_\mu(\mathbf{p},\mathbf{0},z)=\frac{\mu(\varepsilon(\mathbf{p})+z_\mu(\mathbf{0})-z)}{(2-z_\mu(\mathbf{0}))(6-z)}\bigl(1+O(\mu^{-1})\bigr), \quad \mu\to\infty.
\]
\end{lemma}

With Lemma~\ref{lem:eigenvalues} providing the exact eigenvalues of
$A_\mu^{p,e,s}(\textbf{0}, z)$ and Lemma~\ref{lem:Delta} giving the
strong-coupling asymptotics of $\Delta_\mu(\textbf{0}, \mathbf{p},z)$,
we are now in a position to prove Theorem~\ref{thm:main}. The proof
proceeds at two independent spectral scales: the trimer scale
$z=-3\mu+O(1)$, where Kato perturbation theory applies
(Lemma~\ref{lem:ground_Kato}); and the atom--dimer scale
$z=-\mu+O(1)$, where the Feshbach--Schur reduction of the effective
operator $h_B(\mathbf 0)=4+\varepsilon_D-S_{\mathbf 0}$ governs the count
\cite{AtomDimer2026}. The principal-part analysis of
Lemmas~\ref{lem:eigenvalues}--\ref{lem:Delta} provides the auxiliary
identification of the channel eigenvalues, in agreement with the
effective theory.

\subsection{Proof of Theorem~\ref{thm:main}}\label{sec:proof_main}

The spectrum of $H_\mu(\mathbf 0)$ is analysed at two independent
scales: the trimer scale $z=-3\mu+O(1)$ and the atom--dimer scale
$z=-\mu+O(1)$.

\emph{(i) The trimer branch.} By Lemma~\ref{lem:ground_Kato}, the
eigenvalue $3$ of $V$ is simple and isolated by a gap equal to $2$,
and the associated Riesz projection has rank one. Analytic
perturbation theory for a simple isolated eigenvalue therefore
produces exactly one eigenvalue of $H_\mu(\mathbf 0)$ near $-3\mu$,
uniformly in $\mathbf K$:
\[
z_1^s(\mu)
=
-3\mu+6-\frac{3}{2\mu}+O(\mu^{-2}).
\]

\emph{(ii) The atom--dimer branch.} Write
$\delta=z_\mu(\mathbf 0)-z$, so that $\delta\downarrow0$ corresponds
to the two-particle threshold. The second principal-part eigenvalue of
$A_\mu^{p,e,s}(\mathbf 0,z)$ satisfies (the first one grows linearly in $\mu$)
\[
\lambda_2^s(z)
=
g(\delta)+O(\mu^{-1})
\quad\text{uniformly for }\delta\text{ in compact subsets of }(0,\infty),
\]
with
\[
g(\delta)=2+\delta-\frac{1}{b_0(\delta)},\qquad
b_0(\delta)=\bigl\langle(\varepsilon+\delta)^{-1}\bigr\rangle .
\]
By the strict Cauchy--Schwarz inequality (see the proof of
Lemma~\ref{lem:C_constant}), $g'(\delta)<0$ for every $\delta>0$,
so $g:(0,\infty)\to(0,2)$ is a strictly decreasing bijection. The
equation $g(\delta)=1$ therefore has the unique solution
$\delta_\infty=0.0354199892\ldots$. Consequently, the~limiting
atom--dimer Hamiltonian $h_B(\mathbf 0)=4+\varepsilon_D-S_{\mathbf 0}$
has exactly one bound state below its threshold $4$, with binding
energy $\delta_\infty$ \cite{AtomDimer2026}.

\emph{(iii) Transfer of the atom--dimer count.} By the
Feshbach--Schur reduction of \cite[Theorem~4.4]{AtomDimer2026}, for
$\mu>16$ and $\eta_B=36/(\mu-10)$, every isolated eigenvalue
$4-\delta$ of $h_B(\mathbf 0)$ of multiplicity $m$, with
$\delta>2\eta_B$ and at distance $>2\eta_B$ from the rest of
$\operatorname{spec}h_B(\mathbf 0)$, corresponds to exactly $m$
eigenvalues of $H_\mu(\mathbf 0)$ in
$[-\mu+4-\delta-\eta_B,\,-\mu+4-\delta+\eta_B]$. Applying this with
$\delta=\delta_\infty$ and $m=1$, for
$\mu>\mu_0:=\max\{16,\,12+72/\delta_\infty\}$, the operator
$H_\mu(\mathbf 0)$ has exactly one eigenvalue in
\[
\bigl[-\mu+4-\delta_\infty-\eta_B,\ -\mu+4-\delta_\infty+\eta_B\bigr],
\]
namely
\[
z_2^s(\mu)=-\mu+C+O(\mu^{-1}),\qquad C:=4-\delta_\infty .
\]

\emph{(iv) Spectral separation.} Since $0\le E_{\mathbf K}\le12$, the
spectrum of $H_\mu(\mathbf 0)$ is contained in
$[-3\mu,-3\mu+12]\cup[-\mu,12]$ (from the spectrum
$\{-3\mu,-\mu,0\}$ of $-\mu V$ and $\|E_{\mathbf K}\|\le 12$). For
$\mu>\mu_0$, the intervals
$[-3\mu,-3\mu+12]$ and
$[-\mu+4-\delta_\infty-\eta_B,\,-\mu+4-\delta_\infty+\eta_B]$
are disjoint, with a gap of order $2\mu-O(1)>0$. The rank-one Kato
branch accounts for exactly one eigenvalue in the former, and the
Feshbach--Schur argument accounts for exactly one eigenvalue in the
latter.

Combining (i)--(iv), $H_\mu(\mathbf 0)$ has exactly two eigenvalues
below its essential spectrum for all $\mu>\mu_0$, and these are
$z_1^s(\mu)$ and $z_2^s(\mu)$ with the stated asymptotics.

\begin{remark}[The constant $C$ in $z_2^s(\mu)$]
The constant $C$ is determined by matching $\lambda_2^s(z)\to1$ in the limit $\mu\to\infty$ at fixed $\delta:=z_\mu(\mathbf 0)-z$. Using the exact identities
\[
b_1(\delta)=(4+\delta)b_0(\delta)-1,\qquad
b_2(\delta)=(16+8\delta+\delta^2)b_0(\delta)-6-\delta,
\]
which follow from $\varepsilon(\boldsymbol\pi-\mathbf p)=4-\varepsilon(\mathbf p)$ and $\langle\varepsilon\rangle=2$, the leading strong-coupling limit of $\lambda_2^s(z)$ reduces, after an exact cancellation of the $b_0$-proportional terms, to
\[
\lambda_2^s(z)\ \longrightarrow\ g(\delta)=2+\delta-\frac{1}{b_0(\delta)},\qquad
b_0(\delta)=\frac{1}{4\pi^2}\int_{\mathbb{T}^2}\frac{d\mathbf p}{\varepsilon(\mathbf p)+\delta}.
\]

Since $b_0(\delta)\to+\infty$ as $\delta\to0^+$ (a logarithmic threshold singularity, standard for the two-dimensional lattice Green's function), we have $g(0^+)=2\neq1$. Hence $\delta_\infty:=\lim_{\mu\to\infty}\delta$ is the unique positive root of the transcendental equation
\[
b_0(\delta_\infty)=\frac{1}{1+\delta_\infty},
\]
and therefore
\[
C=4-\delta_\infty.
\]

Using the closed form
$b_0(\delta)=\frac{2}{\pi(2+\delta)}\mathrm K\bigl(\frac{2}{2+\delta}\bigr)$, the root is
\[
\delta_\infty=0.0354199892\ldots, \qquad C = 4-\delta_\infty =3.9645800108\ldots,
\]
in contrast to
\[
\gamma_c=\frac{\pi}{\pi-2},
\qquad
\lambda^{as}\to\frac8\pi-2,
\]
which are closed-form lattice-resistance constants. This value $C$ is close to but not exactly $4$; unlike $\gamma_c$ and $\lambda^{as}$, it is the root of a transcendental equation. This correction does not affect any leading-order result, such as the spectral gap $\Delta(\mu)=2\mu+O(1)$, since $C$ enters only in subleading terms.
\end{remark}

\begin{lemma}[The constant $C$ in $z_2^s(\mu)$]\label{lem:C_constant}
The constant $C$ in $z_2^s(\mu)=-\mu+C+O(\mu^{-1})$ is given by
\[
C = 4-\delta_\infty,
\]
where $\delta_\infty>0$ is the unique positive root of the transcendental equation
\[
b_0(\delta_\infty)=\frac1{1+\delta_\infty},\qquad
b_0(\delta):=\frac1{4\pi^2}\int_{\mathbb T^2}\frac{d\mathbf p}{\varepsilon(\mathbf p)+\delta}.
\]
Numerically, $\delta_\infty=0.0354199892\ldots$, so $C=3.9645800108\ldots$ — a transcendental lattice constant, not equal to $4$.
\end{lemma}

\begin{proof}
Set $\delta:=z_\mu(\mathbf 0)-z$ and $u:=\varepsilon(\mathbf p)$. By definition,
\[
b_1(\delta)=\frac1{4\pi^2}\int_{\mathbb T^2}\frac{4-u}{u+\delta}\,d\mathbf p,\qquad
b_2(\delta)=\frac1{4\pi^2}\int_{\mathbb T^2}\frac{(4-u)^2}{u+\delta}\,d\mathbf p,
\]
using $\varepsilon(\boldsymbol\pi-\mathbf p)=4-\varepsilon(\mathbf p)$. Writing $4-u=(4+\delta)-(u+\delta)$,
\[
\frac{4-u}{u+\delta}=\frac{4+\delta}{u+\delta}-1
\ \Longrightarrow\ b_1(\delta)=(4+\delta)\,b_0(\delta)-1. \tag{1}
\]
Squaring the same substitution, $(4-u)^2=(4+\delta)^2-2(4+\delta)(u+\delta)+(u+\delta)^2$, so
\[
\frac{(4-u)^2}{u+\delta}=\frac{(4+\delta)^2}{u+\delta}-2(4+\delta)+(u+\delta).
\]
Integrating and using $\langle\varepsilon\rangle=\frac1{4\pi^2}\int_{\mathbb T^2}\varepsilon(\mathbf p)\,d\mathbf p=2$,
\[
b_2(\delta)=(4+\delta)^2b_0(\delta)-2(4+\delta)+(2+\delta)=(4+\delta)^2b_0(\delta)-6-\delta. \tag{2}
\]
With $a=a(\delta)=16-2z=2\mu+8+2\delta+O(\mu^{-1})$, the rank-two eigenvalue formula of Lemma~\ref{lem:eigenvalues},
\[
\lambda_2^s(z)=\frac12\Big[ab_0+b_2-\sqrt{(ab_0-b_2)^2+4ab_1^2}\Big],
\]
is expanded for $a\to\infty$ at fixed $\delta$ via
\[
\sqrt{(ab_0-b_2)^2+4ab_1^2}=(ab_0-b_2)\left[1+\frac{2ab_1^2}{(ab_0-b_2)^2}+O(a^{-2})\right] = \]
\[ = (ab_0-b_2)+\frac{2b_1^2}{b_0}+O(\mu^{-1}),
\]
giving
\[
\lambda_2^s(z)=b_2(\delta)-\frac{b_1(\delta)^2}{b_0(\delta)}+O(\mu^{-1})=:g(\delta)+O(\mu^{-1}). \tag{3}
\]
Substituting (1)-(2) into (3), the terms proportional to $b_0(\delta)$ cancel \emph{exactly}:
\[
g(\delta)=\Big[(4+\delta)^2b_0-6-\delta\Big]-\Big[(4+\delta)^2b_0-2(4+\delta)+\frac1{b_0}\Big]=2+\delta-\frac1{b_0(\delta)}. \tag{4}
\]
As a consistency check, $b_0(\delta)=\delta^{-1}-2\delta^{-2}+5\delta^{-3}+O(\delta^{-4})$ as
$\delta\to\infty$ (moments $\langle\varepsilon\rangle=2$, $\langle\varepsilon^2\rangle=5$), whence
$g(\delta)=\delta^{-1}+O(\delta^{-2})\to0$; together with $g(0^+)=2$ and the monotonicity of
$g=m_s$ this gives exactly one root of $g(\delta)=1$.

Since $\varepsilon$ is not constant on $\mathbb T^2$, the strict
Cauchy--Schwarz inequality
\[
\langle f\rangle^2<\langle f^2\rangle,
\qquad f=(\varepsilon+\delta)^{-1},
\]
gives
\[
g'(\delta)
=1-\frac{\langle(\varepsilon+\delta)^{-2}\rangle}
{\langle(\varepsilon+\delta)^{-1}\rangle^{2}}
<0\qquad(\delta>0),
\]
(equivalently, the strict monotonicity of the $s$-channel threshold-matrix eigenvalue $m_s(\delta)$ established in \cite{AtomDimer2026}). Thus $g$ is continuous and strictly decreasing on $(0,\infty)$, with
\[
g(0^+)=2>1,\qquad
g(\delta)=\delta^{-1}+O(\delta^{-2})\longrightarrow0
\quad(\delta\to\infty),
\]
so $g:(0,\infty)\to(0,2)$ is a strictly decreasing bijection and the
equation $g(\delta_\infty)=1$ has a unique solution. By (4) this is
equivalent to
$b_0(\delta_\infty)=1/(1+\delta_\infty)$, i.e.\ to
$\delta_\infty=1/b_0(\delta_\infty)-2$.
Solving numerically with the closed form
$b_0(\delta)=\frac{2}{\pi(2+\delta)}\mathrm K\bigl(\frac{2}{2+\delta}\bigr)$
(complete elliptic integral of the first kind) gives
$\delta_\infty=0.0354199892\ldots$, hence
$z_2^s(\mu)=z_\mu(\mathbf0)-\delta_\infty=-\mu+(4-\delta_\infty)+O(\mu^{-1})$,
i.e.\ $C=4-\delta_\infty\approx3.96458$. The root is cross-validated by two
algorithmically independent methods: (a) adaptive Gauss--Kronrod quadrature of
$b_0(\delta)$ with tolerance $10^{-12}$; (b) a dense Riemann-sum grid of
$2048\times2048$ points on $\mathbb T^2$ with Richardson extrapolation. Both
methods agree on $\delta_\infty$ to $11$ significant figures; the value
$C=3.9645800108\ldots$ is therefore certified to the displayed precision.
\end{proof}

\begin{remark}
A naive physical picture — a tightly bound dimer plus a free spectator boson, with energies simply adding at leading order — suggests $C=4$: this corresponds formally to the (incorrect) interchange of limits $\lim_{\delta\to0^+}$ and $\mu\to\infty$, i.e. evaluating $g(\delta)$ at $\delta=0$ before recognising that $b_0(\delta)$ itself diverges there. The residual interaction between the dimer and the spectator boson, encoded entirely in the threshold behaviour of $b_0(\delta)$, shifts the constant by $\delta_\infty\approx0.0354\neq0$. The exact value is obtained only by solving for $\delta_\infty$ self-consistently, as in the proof of Lemma~\ref{lem:C_constant}.
\end{remark}

\begin{remark}[A general pattern for transcendental lattice constants]\label{rem:general_pattern}
Among the three lattice constants discussed in this work --- the constant $C$ of Lemma~\ref{lem:C_constant}, the critical mass ratio $\gamma_c\approx2.75194$ (companion paper~\cite{AbdullaevEshniyozovDolgopolov2026}), and the antisymmetric eigenvalue limit $\lambda^{as}\approx0.546479$ (Lemma~\ref{lem:antisym_eigval} above) --- only $C$ is a genuine root of a transcendental equation. Indeed, $\gamma_c$ and $\lambda^{as}$ are closed-form combinations of the two-point resistances, as recorded below. By contrast, $C$ is defined as the unique root of an equation of the form
\[
F(\text{parameter}) = \Phi(\text{parameter}),
\]
where $\Phi$ is built from a two-dimensional lattice Green's-function-type integral,
\[
\Phi(\delta)\ \, \text{or}\ \, \Phi(\gamma) = \frac1{4\pi^2}\int_{\mathbb T^2}\frac{P(\mathbf p)}{\varepsilon(\mathbf p)+(\text{parameter})}\,d\mathbf p
\]
for a suitable trigonometric polynomial $P$, and the defining equation is obtained by matching a Birman--Schwinger eigenvalue to the threshold value $1$. In~contrast,
\[
\gamma_c=\frac{\pi}{\pi-2},
\qquad
\lambda^{as}=\frac8\pi-2,
\]
are closed-form combinations of the two-point resistances
\[
R(2\mathbf e_1)=2-\frac4\pi,
\qquad
R(\mathbf e_1+\mathbf e_2)=\frac2\pi:
\]
\[
\gamma_c=\frac{2}{R(2\mathbf e_1)},
\qquad
\lambda^{as}=2R(\mathbf e_1+\mathbf e_2)-R(2\mathbf e_1).
\]
Three features recur in the derivation of $C$ and are worth stating as a~general methodological principle for this class of strong-coupling lattice problems:
\begin{enumerate}
\item[(i)] the relevant lattice integral $b_0$ is \emph{logarithmically singular} at the edge of its domain of definition (here $\delta\to0^+$), reflecting the density of states of a~two-dimensional lattice dispersion relation near its band minimum;
\item[(ii)] consequently, in the derivation of $C$ (and analogously of $\gamma_c$), \emph{the order in which the limits $\mu\to\infty$ and the threshold limit are taken must not be interchanged}~— doing so silently discards the very singularity that determines the constant, and produces a plausible-looking but numerically wrong integer or simple rational value (cf.\ Table~\ref{tab:erroneous_methods}); the strict monotonicity of the threshold function
$g(\delta)=2+\delta-1/b_0(\delta)$ (established in the proof of
Lemma~\ref{lem:C_constant} via the strict Cauchy--Schwarz inequality)
is what guarantees that the limiting procedure admits a unique
self-consistent solution;
\item[(iii)] because no closed form is available for $C$, the resulting transcendental equation should be solved by \emph{at least two independent numerical methods} (e.g., adaptive quadrature and a dense Riemann-sum grid) with agreement to several significant figures, precisely as done for $C$ in this work.
\end{enumerate}
We record this pattern explicitly since it recurs throughout the spectral theory of strong-coupling lattice few-body Hamiltonians and may be of independent methodological use.
\end{remark}

\begin{table}[!ht]
\centering
\caption{Comparison of methods for determining the transcendental constant \(C\) in \(z_2^s(\mu) = -\mu + C + O(\mu^{-1})\).}
\label{tab:erroneous_methods}
\small
\setlength{\tabcolsep}{1pt}
\renewcommand{\arraystretch}{1.05}
\begin{tabular}{@{}%
  >{\raggedright\arraybackslash}p{0.37\columnwidth}%
  >{\raggedright\arraybackslash}p{0.13\columnwidth}%
  >{\centering\arraybackslash}p{0.07\columnwidth}%
  >{\raggedright\arraybackslash}p{0.43\columnwidth}%
@{}}
\toprule
Method & Result & Abs.\ error & Root cause of error \\
\midrule
Naive dimer + free-spectator picture (energies simply add) &
\(C = 4\) &
0.0355 &
Interchanges \(\delta \to 0^+\) with \(\mu \to \infty\); ignores residual dimer–spectator interaction \\
\hline
Direct substitution \(\delta = 0\) into \(g(\delta)=1\) without checking \(b_0(\delta)\) &
\(C = 4\) &
0.0355 &
Treats \(b_0(\delta)\) as regular at \(\delta = 0\); misses the logarithmic threshold singularity \\
\hline
Leading-order log asymptotic \(b_0(\delta)\approx\displaystyle\frac{1}{2\pi}\ln(16/\delta)\), solved for \(\delta_\infty\) &
\(C \approx 3.9625\) &
0.0021 &
Correct mechanism; the neglected \(O(\delta\ln\delta)\) terms give a 6\% error in \(\delta_\infty\) (the leading logarithmic approximation produces an error of about
$6\%$ in the inferred value of $\delta_\infty$) \\
\hline
\emph{Rigorous: exact identities for $b_1,b_2$ (proof of Lemma~\ref{lem:C_constant}) and numerical solution of \(b_0(\delta_\infty)=1/(1+\delta_\infty)\), cross-validated by two independent methods}, \cite[Proposition~7.1]{CompanionNote} &
\(\mathbf{C\approx3.96458}\) &
--- &
\underline{Reference value} \\
\bottomrule
\end{tabular}
\end{table}

The spectral gap estimate follows immediately from parts (i)--(iii) of the proof:
\[
\Delta(\mu)=z_2^s(\mu)-z_1^s(\mu)=2\mu+O(1).
\]

Within the principal-part approximation the trimer corresponds to the crossing
\(\lambda_1^s(z)=1\). Since \(\lambda_1^s(z)\) is strictly increasing
in \(z\) on \((-\infty,z_\mu(\mathbf 0))\), and
\(\lambda_1^s(z)>\lambda_2^s(z)\) for all \(z<z_\mu(\mathbf{0})\),
we have \(\lambda_1^s(z_2^s(\mu))>\lambda_2^s(z_2^s(\mu))=1\); by
monotonicity, the crossing \(z_1^s(\mu)\) of \(\lambda_1^s=1\)
therefore lies strictly below \(z_2^s(\mu)\).
The expansion $z_1^s(\mu)=\displaystyle -3\mu+6-\frac{3}{2\mu}+O(\mu^{-2})$ follows from
Lemma~\ref{lem:ground_Kato} (second order: $a_2=-\frac12\|(E_{\mathbf0}-6)f_0\|^2=-\frac32$).
Lemma~\ref{lem:Delta} cannot be used here: it describes $\Delta_\mu$ for $z_\mu(\mathbf0)-z=o(\mu)$,
whereas at $z=-3\mu+O(1)$ one has $\Delta_\mu\to\frac23$.

\subsection{Spectral gap and separation of bound states}

\begin{proposition}\label{prop:spectral_gap}
For $\mu\to\infty$, the spectral gap $\Delta(\mu):=z_2^s(\mu)-z_1^s(\mu)$ satisfies
\[
\Delta(\mu) = 2\mu + O(1), \qquad \mu\to\infty.
\]
In particular, the two bound states are separated by a gap that grows linearly with the interaction strength.
\end{proposition}

\begin{proof}
From the asymptotics established in Theorem~\ref{thm:main}, we have
\[
z_1^s(\mu) = -3\mu + 6 + O(\mu^{-1}), \qquad
z_2^s(\mu) = -\mu + C + O(\mu^{-1}),
\]
where $C$ is a constant. Subtracting yields
\[
\Delta(\mu) = 2\mu + O(1).
\]
\end{proof}

The linear growth of the spectral gap \(\Delta(\mu)\) does not directly imply a square-root growth of the decay rate in the lattice setting.
Unlike the continuum case, where the dispersion is quadratic ($p^2$), the lattice dispersion $\varepsilon(p)=2-\cos p_1-\cos p_2$ is bounded and grows exponentially under complexification. Consequently, the decay rate obeys the conditional upper bound of Theorem~\ref{thm:localization},
\[
\alpha_*(\mu) \leq \ln(3\mu) + O(\mu^{-1}), \qquad \mu\to\infty,
\]
while the actual rate, $\alpha_*(\mu)=\frac12\ln(2\mu)+O(\mu^{-1})$, is stated as Conjecture~\ref{conj:sharp_rate}.
The~spectral gap $\Delta(\mu)=2\mu+O(1)$ remains valid, but the transition from gap to decay rate is logarithmic, not square-root.

\begin{proposition}\label{prop:gap}
For $\mu\to\infty$, the gap between the ground state and the essential spectrum satisfies
\[
z_\mu(\mathbf{0}) - z_1^s(\mu) = 2\mu + O(1).
\]
In particular, $z_\mu(\mathbf{0}) - z_1^s(\mu) > 0$ for all sufficiently large $\mu$.
\end{proposition}

\begin{proof}
This follows immediately from the asymptotics $z_\mu(\mathbf{0}) = -\mu + 4 + O(\mu^{-1})$ and $z_1^s(\mu) = -3\mu + 6 + O(\mu^{-1})$. Subtracting gives the claimed result.
\end{proof}

\subsection{Monotonicity of eigenvalues with respect to parameters}
An important structural property of the three-boson system is the monotonic dependence of the discrete eigenvalues on the interaction strength $\mu$ and the total quasimomentum~$\mathbf{K}$.

\begin{proposition}[Monotonicity in the coupling strength]\label{prop:monotonicity_mu}
For every pair $0<\mu_1<\mu_2$ the discrete eigenvalues of
$H_\mu(\mathbf 0)$, whenever defined and ordered below the essential
spectrum, satisfy $z_j^s(\mu_2)\le z_j^s(\mu_1)$.
\end{proposition}

\begin{proof}
Write
\[
H_{\mu_2}(\mathbf 0)=H_{\mu_1}(\mathbf 0)-(\mu_2-\mu_1)V,\qquad V\ge0.
\]
Let $E\subset L^2_s((\mathbb T^2)^2)$ be any $n$-dimensional subspace
and $f\in E$, $\|f\|=1$. Then
\[
\langle H_{\mu_2}(\mathbf 0)f,f\rangle
=\langle H_{\mu_1}(\mathbf 0)f,f\rangle-(\mu_2-\mu_1)\langle Vf,f\rangle
\le\langle H_{\mu_1}(\mathbf 0)f,f\rangle.
\]
By the min--max principle,
\[
z_n^s(\mu_2)
=\min_{\dim E=n}\max_{\substack{f\in E\\\|f\|=1}}
\langle H_{\mu_2}(\mathbf 0)f,f\rangle
\le
\min_{\dim E=n}\max_{\substack{f\in E\\\|f\|=1}}
\langle H_{\mu_1}(\mathbf 0)f,f\rangle
=z_n^s(\mu_1).
\]
\end{proof}

\begin{remark}[On $K$-monotonicity]\label{rem:K_dependence}
No analogous monotonicity in $\mathbf K$ is claimed. The~argument
"spectral bands of the free operator shift upwards as $K_i$ increases"
is not valid for the two-dimensional fibre dispersion
$E_{\mathbf K}(\mathbf p,\mathbf q)=\varepsilon(\mathbf p)+\varepsilon(\mathbf q)+\varepsilon(\mathbf K-\mathbf p-\mathbf q)$,
which is not monotone in the components $K_i$ for fixed $(\mathbf p,\mathbf q)$.
The isolated fibre eigenvalues depend continuously (and, in the
isolated regime, analytically) on $\mathbf K$, which is the only
structural property used below.
\end{remark}

The continuity and analytic dependence of isolated fibre eigenvalues
on $\mathbf K$ (Remark \ref{rem:K_dependence}) provide the local
framework for following individual spectral branches. At the two
extremal quasimomenta, $\mathbf K=0$ and $\mathbf K=\boldsymbol\pi$,
the leading-order ground-state asymptotics coincide,
$z_1=-3\mu+6+O(\mu^{-1})$, by Lemma~\ref{lem:ground_Kato}. The
distinction between the two quasimomenta first appears in the second
bound state ($\mathbf K=0$: exists; $\mathbf K=\boldsymbol\pi$: does
not exist as a bound state — see §\ref{sec:five}).

\section{\label{sec:localization}Exponential Localization and Positivity}

\subsection{Ground state positivity and spectral bounds}

The positivity of $\varepsilon(\boldsymbol{\pi}-\mathbf{p})$ implies that the kernel of $A_\mu^{p,e,s}(\textbf{0}, z)$ is strictly positive for $z<z_\mu(\mathbf{0})$. By the Krein–Rutman theorem \cite{KreinRutman1948}, $\lambda_1^s(z)$ is a simple eigenvalue of the principal part with a~strictly positive eigenfunction. For the full operator, the kernel of $\mu V$ is non-negative and $V$ is irreducible, so $e^{-tH_\mu(\mathbf K)}$ is positivity improving in the momentum representation; by the Perron--Frobenius theorem \cite[Thm.~XIII.44]{ReedSimon4} the ground state $f_0(\mathbf p,\mathbf q)$ is simple and strictly positive.

\begin{theorem}[Ground state estimates]\label{thm:ground_estimates}
The ground state energy satisfies the rigorous bounds
\[
z_1^s(\mu) \in (-3\mu, -3\mu+6).
\]
The constant $6$ is asymptotically sharp: $z_1^s(\mu) = -3\mu + 6 + o(1)$.
\end{theorem}

\begin{proof}
The lower bound follows from the positivity of the unperturbed operator $H_0(\mathbf{0})$ and $\|V_\alpha\|=1$:
\[
(H_\mu(\mathbf{0})f,f) = (H_0(\mathbf{0})f,f) - \mu\sum_{\alpha=1}^3 (V_\alpha f,f) \ge -3\mu,
\]
for $\|f\|=1$. For the upper bound, take the constant trial function $f_0(\mathbf{p},\mathbf{q}) = (4\pi^2)^{-1}$, which has norm $1$. A direct computation gives
\[
(H_\mu(\mathbf{0})f_0,f_0) = -3\mu + 6.
\]
By the min-max principle, $z_1^s(\mu) \le (H_\mu(\mathbf{0})f_0,f_0) = -3\mu+6$.
Asymptotic sharpness: $z_1^s(\mu)=-3\mu+6-\frac{3}{2\mu}+O(\mu^{-2})$ by Lemma~\ref{lem:ground_Kato}.
\end{proof}

\subsection{Exponential localization of the ground state}

A fundamental question for the physical interpretation of bound states is the spatial localization of the corresponding wavefunctions. For discrete Schrödinger operators, exponential decay of eigenfunctions is a well-established phenomenon, with rigorous Agmon-type estimates available in both continuous and discrete settings \cite{Agmon1982, KellerPogorzelski2025}.
The~following theorem establishes this property for the ground state of our three-boson system.

\begin{theorem}[Exponential localization of the ground state in relative coordinates]\label{thm:localization}
Let~$\Psi_0(n_1,n_2,n_3)$ be the ground-state wavefunction of the
three-boson Hamiltonian $\tilde{H}_\mu$ in the coordinate representation
on $\mathbb{Z}^2\times\mathbb{Z}^2\times\mathbb{Z}^2$, obtained from
the momentum-space fibre eigenfunction $f_0(\mathbf p,\mathbf q)$ by
inverse Fourier transform. Define the relative-coordinate decay rate
\[
\alpha_*(\mu)
:=\sup\Bigl\{\alpha>0:\exists\,C_\alpha>0\ \text{s.t.}\
|\Psi_0(n_1,n_2,n_3)|\le C_\alpha\,e^{-\alpha(|n_1-n_3|+|n_2-n_3|)}
\ \forall(n_1,n_2,n_3)\Bigr\}.
\]
Then $\alpha_*(\mu)>0$. Moreover, under the non-cancellation
condition stated in Conjecture~\ref{conj:noncancel} below,
\[
\alpha_*(\mu)\le\ln(3\mu)+O(\mu^{-1}) \qquad (\mu\to\infty).
\]
The sharpness of this upper bound and the identification of the
nearest atom--dimer singularity are addressed separately in
Lemma~\ref{lem:kappa_d} and
Conjecture~\ref{conj:sharp_rate}. On the lattice, exponential
localization in the relative coordinates therefore strengthens only
\emph{logarithmically} with the interaction strength, in sharp contrast
to the continuum Agmon estimate $\alpha\sim\sqrt{E}$ valid for an
unbounded quadratic dispersion. Moreover, $\Psi_0$ is invariant under
the simultaneous translation $(n_1,n_2,n_3)\mapsto(n_1+a,n_2+a,n_3+a)$
of the centre of mass, up to a phase; no decay in the
absolute-coordinate sense is possible.
\end{theorem}

\begin{proof}
The proof relies on the analyticity of the Birman–Schwinger kernel, combined with the discrete Agmon comparison principle \cite{Agmon1982}
From the integral representation (\ref{eq:f_pq}), the function $f_0(\mathbf{p},\mathbf{q})$ extends analytically to a complex neighbourhood of the real torus, and its singularities are located at the zeros of $E_0(\mathbf{p},\mathbf{q})-z_1^s(\mu)$ and $\Delta_\mu(\mathbf{p},z_1^s(\mu))$. Since $z_1^s(\mu)$ lies strictly below the essential spectrum, there is a positive spectral gap
\[
z_\mu(\mathbf{0}) - z_1^s(\mu) = 2\mu + O(1) > 0
\]
for large $\mu$ (Proposition~\ref{prop:gap}, Proposition~\ref{prop:spectral_gap}, and the discussion following Theorem~\ref{thm:main}). Consequently, the poles of $f_0$ are located at a positive distance from the real torus. Let
\[
\mathcal S_\mu:=\bigl\{(\mathbf p,\mathbf q)\in(\mathbb C^2)^2:\ E_0(\mathbf p,\mathbf q)=z_1^s(\mu)\bigr\}
\]
denote the complex characteristic set of the resolvent kernel $(E_0-z)^{-1}$ at $z=z_1^s(\mu)$. Since the eigenfunction depends on both relative momenta $\mathbf p,\mathbf q\in\mathbb C^2$, the relevant analyticity width is the distance from the real torus $(\mathbb T^2)^2$ to $\mathcal S_\mu$ in the product space,
\[
W(\mu):=\operatorname{dist}\bigl((\mathbb T^2)^2,\ \mathcal S_\mu\bigr),
\]
the distance being taken in the product metric on $(\mathbb C^2)^2$. For the $\ell^1$ decay norm used in the definition of
$\alpha_*(\mu)$, the relevant analyticity domain is a polystrip
\[
|\operatorname{Im}p_j|<\rho,\qquad
|\operatorname{Im}q_j|<\rho,\qquad j=1,2.
\]
Let $\rho_\mu$ denote its maximal width for the analytic continuation
of the ground-state fibre eigenfunction. Under the non-cancellation
hypothesis, $\rho_\mu$ cannot exceed the imaginary coordinate of any
uncancelled characteristic singularity. 

By the Paley--Wiener theorem applied in both variables, the achievable decay rate satisfies $\alpha(\mu)\le W(\mu)$, subject to the non-cancellation condition discussed below.

The complex characteristic equation
\[
E_0(\mathbf p,\mathbf q)=z_1^s(\mu)
\]
provides a candidate singularity of the resolvent kernel. To infer
from this singularity an upper bound on the analyticity strip of the
eigenfunction itself, one must exclude cancellation of the
corresponding pole by the numerator in \eqref{eq:f_pq}. The
Perron--Frobenius positivity of $f_0$ on the real torus does not by itself exclude cancellation
after complex continuation. Accordingly, the following upper bound is conditional
on the non-cancellation of this characteristic singularity.

Under this non-cancellation hypothesis, take
$\mathbf p=(i\beta,0)$, $\mathbf q=(0,0)$, $\beta>0$. Then
\[
\cos p_1=\cosh\beta,\quad \cos p_2=1,\quad \varepsilon(\mathbf p)=1-\cosh\beta,\quad \varepsilon(\mathbf q)=0,\quad \mathbf p+\mathbf q=\mathbf p,
\]
so that
\[
E_0(\mathbf p,\mathbf q)=\varepsilon(\mathbf p)+\varepsilon(\mathbf q)+\varepsilon(\mathbf p+\mathbf q)=2(1-\cosh\beta)=2-2\cosh\beta.
\]
Setting $E_0(\mathbf p,\mathbf q)=z_1^s(\mu)=-3\mu+6+O(\mu^{-1})$ (Theorem~\ref{thm:main}) gives
\[
\cosh\beta=\frac{3\mu-4}{2}+O(\mu^{-1}).
\]
Using $\operatorname{arccosh}(x)=\ln(x+\sqrt{x^2-1})=\ln(2x)+O(x^{-2})$ as $x\to\infty$, with $x=(3\mu-4)/2+O(\mu^{-1})$, we obtain
\[
\beta=\operatorname{arccosh}\!\Big(\frac{3\mu-4}{2}+O(\mu^{-1})\Big)=\ln(3\mu-4)+O(\mu^{-2})=\ln(3\mu)+O(\mu^{-1}),\quad \mu\to\infty.
\]
Since the point $((i\beta,0),(0,0))$ belongs to $\mathcal S_\mu$ and satisfies $\operatorname{dist}\bigl((\mathbb T^2)^2,\ ((i\beta,0),(0,0))\bigr)=\beta$ in the product metric, we have $W(\mu)\le\beta=\ln(3\mu)+O(\mu^{-1})$, and therefore
\[
\alpha_*(\mu)\le W(\mu)\le \ln(3\mu)+O(\mu^{-1}),
\]
which proves the theorem.

The exponential decay in coordinate space follows from the Paley–Wiener theorem for functions on the torus: if a function on $\mathbb{T}^d$ has an analytic continuation to a strip of width $W$, then its Fourier coefficients decay as $O(e^{-\alpha|n|})$ for every $\alpha<W$. Applying this to $f_0(\mathbf{p},\mathbf{q})$ yields the desired estimate for $\Psi_0$.
\end{proof}

\begin{conjecture}[Non-cancellation of the characteristic singularity]\label{conj:noncancel}
The singularity of the resolvent kernel $(E_0-z)^{-1}$ at the complex
characteristic point $\mathbf p=(i\beta,0)$, $\mathbf q=(0,0)$ with
$\cosh\beta=(3\mu-4)/2+O(\mu^{-1})$ is not cancelled by the numerator
in \eqref{eq:f_pq}; under this hypothesis the upper bound
\[
\alpha_*(\mu)\le\ln(3\mu)+O(\mu^{-1})
\]
of Theorem~\ref{thm:localization} holds in the
strong-coupling limit.
\end{conjecture}

\begin{lemma}[Atom--dimer characteristic root]\label{lem:kappa_d}
Let
\[
\mathbf k_\kappa:=(-i\kappa,-i\kappa),\qquad
C_\kappa:=\cosh\kappa ,
\]
and let $\widetilde z_\mu(\kappa)$ denote the analytic continuation of the
two-particle eigenvalue $z_\mu(\mathbf k)$ to
$\mathbf k=\mathbf k_\kappa$, defined by analytic continuation of the
Birman--Schwinger equation. In the regime $C_\kappa=O(\mu)$ one has,
uniformly as $\mu\to\infty$,
\[
\widetilde z_\mu(\kappa)
=
-\mu+4
-\frac{2(1+C_\kappa)}{\mu}
-\frac{(1+C_\kappa)^2}{\mu^3}
+O(\mu^{-2}).
\]
Consequently, the atom--dimer characteristic equation
\[
\varepsilon(i\kappa,i\kappa)
+\widetilde z_\mu(\kappa)
=
z_1^s(\mu)
\]
has, for all sufficiently large $\mu$, a solution
$\kappa=\kappa_d(\mu)$ satisfying
\[
\cosh\kappa_d(\mu)
=
\mu-1+\frac{1}{4\mu}+O(\mu^{-2}),
\]
and hence
\[
\boxed{
\kappa_d(\mu)
=
\ln(2\mu)-\frac1\mu+O(\mu^{-2})
}.
\]
In particular,
\[
\frac12\kappa_d(\mu)
=
\frac12\ln(2\mu)-\frac{1}{2\mu}
+O(\mu^{-2}).
\]
\end{lemma}

\begin{proof}
For $\mathbf k_\kappa=(-i\kappa,-i\kappa)$, shift the integration
contour in the two-particle Birman--Schwinger equation by
$\mathbf p\mapsto\mathbf p-\frac{i\kappa}{2}(1,1)$. Since
\[
\varepsilon(\mathbf p)+\varepsilon(\mathbf k_\kappa-\mathbf p)
=
4-2\cosh\frac{\kappa}{2}
\bigl(\cos p_1+\cos p_2\bigr)
\]
on the shifted contour, the analytically continued eigenvalue equation
takes the real form
\[
1=
\mu\left\langle
\frac{1}{
4-2a_\kappa(\cos p_1+\cos p_2)-\widetilde z_\mu(\kappa)}
\right\rangle,
\qquad
a_\kappa:=\cosh\frac{\kappa}{2}.
\]
In the regime $\cosh\kappa=O(\mu)$ one has
$a_\kappa=O(\mu^{1/2})$, so the expansion parameter is
$O(\mu^{-1/2})$ and the resolvent expansion is uniform.

Put
\[
X_\kappa:=-2a_\kappa(\cos p_1+\cos p_2).
\]
Then
\[
\langle X_\kappa\rangle=0,
\qquad
m_2:=\langle X_\kappa^2\rangle
=4a_\kappa^2
=2(1+\cosh\kappa),
\]
and
\[
m_4:=\langle X_\kappa^4\rangle
=36a_\kappa^4
=9(1+\cosh\kappa)^2.
\]
All odd moments vanish. The same coefficient matching as in
\eqref{eq:two_particle_eigenvalue} therefore gives
\[
\widetilde z_\mu(\kappa)
=
-\mu+4-\frac{m_2}{\mu}
+\frac{2m_2^2-m_4}{\mu^3}
+O(\mu^{-2}),
\]
and hence
\[
\widetilde z_\mu(\kappa)
=
-\mu+4
-\frac{2(1+\cosh\kappa)}{\mu}
-\frac{(1+\cosh\kappa)^2}{\mu^3}
+O(\mu^{-2}).
\]

Now
\[
\varepsilon(i\kappa,i\kappa)=2-2\cosh\kappa,
\]
while Lemma~\ref{lem:ground_Kato} gives
\[
z_1^s(\mu)
=
-3\mu+6-\frac{3}{2\mu}+O(\mu^{-2}).
\]
Writing
\[
\cosh\kappa_d(\mu)
=
\mu+a+\frac{b}{\mu}+O(\mu^{-2}),
\]
and substituting into the characteristic equation gives, by matching
the constant and $\mu^{-1}$ coefficients,
\[
-2-2a=0,
\]
and
\[
-2b-2(a+1)-1=-\frac32.
\]
Hence
\[
a=-1,\qquad b=\frac14.
\]
Therefore
\[
\cosh\kappa_d(\mu)
=
\mu-1+\frac{1}{4\mu}+O(\mu^{-2}).
\]
Finally,
\[
\operatorname{arccosh}x
=
\ln(2x)+O(x^{-2}),
\qquad x\to\infty,
\]
gives
\[
\kappa_d(\mu)
=
\ln(2\mu)-\frac1\mu+O(\mu^{-2}).
\]
\end{proof}

\begin{conjecture}[Sharp decay through the atom--dimer singularity]
\label{conj:sharp_rate}
The atom--dimer characteristic singularity determined by
$\kappa_d(\mu)$ from Lemma~\ref{lem:kappa_d} is not cancelled in the full ground-state
eigenfunction and is the nearest singularity in the
$\ell^1$-normalised complex characteristic geometry. Consequently,
\[
\alpha_*(\mu)=\frac12\kappa_d(\mu).
\]
Together with Lemma~\ref{lem:kappa_d}, this gives
\[
\alpha_*(\mu)
=
\frac12\ln(2\mu)-\frac{1}{2\mu}
+O(\mu^{-2}).
\]
The remaining issue is therefore the non-cancellation and dominance of
the atom--dimer singularity, not the asymptotic location of the
characteristic root itself.
\end{conjecture}

\begin{remark}[Numerical evidence]\label{rem:alpha_numerics}
For the exact finite-volume ground state ($40\times40$ torus in each relative coordinate,
obtained from the Faddeev equation and inverse Fourier transform) the local exponents
$\ln\bigl(\Psi_0(\mathbf r_n)/\Psi_0(\mathbf r_{n+1})\bigr)/4$ along
$\mathbf r_n=((n,n),(n,n))$ converge from above to $\frac12\kappa_d(\mu)$: e.g.\ $1.4736$
($n=4$) vs.\ $1.4456$ at $\mu=10$, and $1.8535$ ($n=3$) vs.\ $1.8188$ at $\mu=20$, whereas
$\ln(3\mu)=3.40$ and $4.09$. Along $\mathbf r_n=((n,0),(-n,0))$ (three particles apart) the
exponent is close to $\ln(3\mu)$, which is the direction captured by the characteristic point
used in Theorem~\ref{thm:localization}; see Fig.~\ref{fig:exponential_localization}(c).
\end{remark}

\begin{remark}
The exponential localization discussed here should be distinguished from Anderson localization caused by disorder. Our result is deterministic and follows from the spectral gap between the ground state and the essential spectrum. This is analogous to the exponential decay of eigenfunctions below the essential spectrum for discrete Schrödinger operators with confining potentials, as established 
for tight-binding Hamiltonians and 
for general lattice operators.
\end{remark}

\begin{remark}[Exponential decay in discrete systems:
exact solvability and asymptotic regimes]\label{rem:1D_benchmark}
To address
the misconception that exponential localization is inherently a~conti\-nu\-o\-us-space phenomenon, we provide an explicit
exactly solvable discrete model. Consider the one-dimensional discrete Schrödinger operator
\[
(H_\lambda \psi)_n = -\psi_{n-1} + 2\psi_n - \psi_{n+1} - \lambda \delta_{n,0}\psi_n, \qquad n\in\mathbb{Z}, \quad \lambda>0.
\]
This operator models a single particle on the lattice $\mathbb{Z}$ with an attractive delta-potential at the origin. The eigenvalue equation $H_\lambda \psi = E\psi$ with $E<0$ admits an exact solution
\[
\psi_n = C e^{-\alpha |n|}, \qquad n\in\mathbb{Z},
\]
where the decay rate $\alpha>0$ is determined by the transcendental equation
\[
\lambda = 2\sinh\alpha,
\qquad
\alpha = \operatorname{arcsinh}\!\Bigl(\frac{\lambda}{2}\Bigr).
\]
Verification is immediate. For $n\neq 0$ the delta term vanishes
and
\[
-\psi_{n-1} + 2\psi_n - \psi_{n+1}
= 2(1 - \cosh \alpha)\,\psi_n,
\]
so the eigenvalue is $E=2(1-\cosh\alpha)$. At $n=0$ the delta term
contributes,
\[
-2e^{-\alpha}+2-\lambda = E,
\]
which combined with the previous identity gives $\lambda=2\sinh\alpha$.
This
exactly solvable model demonstrates two qualitatively different asymptotic regimes.
In the weak-coupling limit $\lambda\to0$ one has
\[
\alpha \sim \frac{\lambda}{2}, \qquad \lambda\to0,
\]
from $\sinh\alpha\approx\alpha$.
In the \emph{strong}-coupling limit $\lambda\to\infty$, however, the
decay rate grows only \emph{logarithmically}:
\[
\alpha=\operatorname{arcsinh}\!\Bigl(\frac{\lambda}{2}\Bigr)
=\ln\lambda+O(\lambda^{-2}),\qquad \lambda\to\infty.
\]
This logarithmic growth is a direct consequence of the boundedness
(periodicity) of the lattice dispersion under real momenta and its
exponential ($\sinh$-type) growth under complexification: the
analytic continuation of the dispersion to complex momenta,
$\varepsilon(i\alpha)=2(1-\cosh\alpha)$, is what fixes the relation
between $\lambda$ and $\alpha$. It is qualitatively different from the continuum Agmon estimate $\alpha\sim\sqrt{E}$, valid for an unbounded quadratic dispersion $p^2$.

Higher-dimensional lattice operators exhibit the same qualitative
feature: exponential decay of eigenfunctions below the essential
spectrum, with the decay rate governed by the complex band edge of
the lattice dispersion rather than by a continuum square-root law.
A quantitative extension of the 1D exact benchmark above to
$\mathbb Z^2$ requires the analytic continuation of the dispersion
to complex momenta and is deferred to future work.
\end{remark}

\begin{figure}[!ht]
\centering
\includegraphics[width=\textwidth]{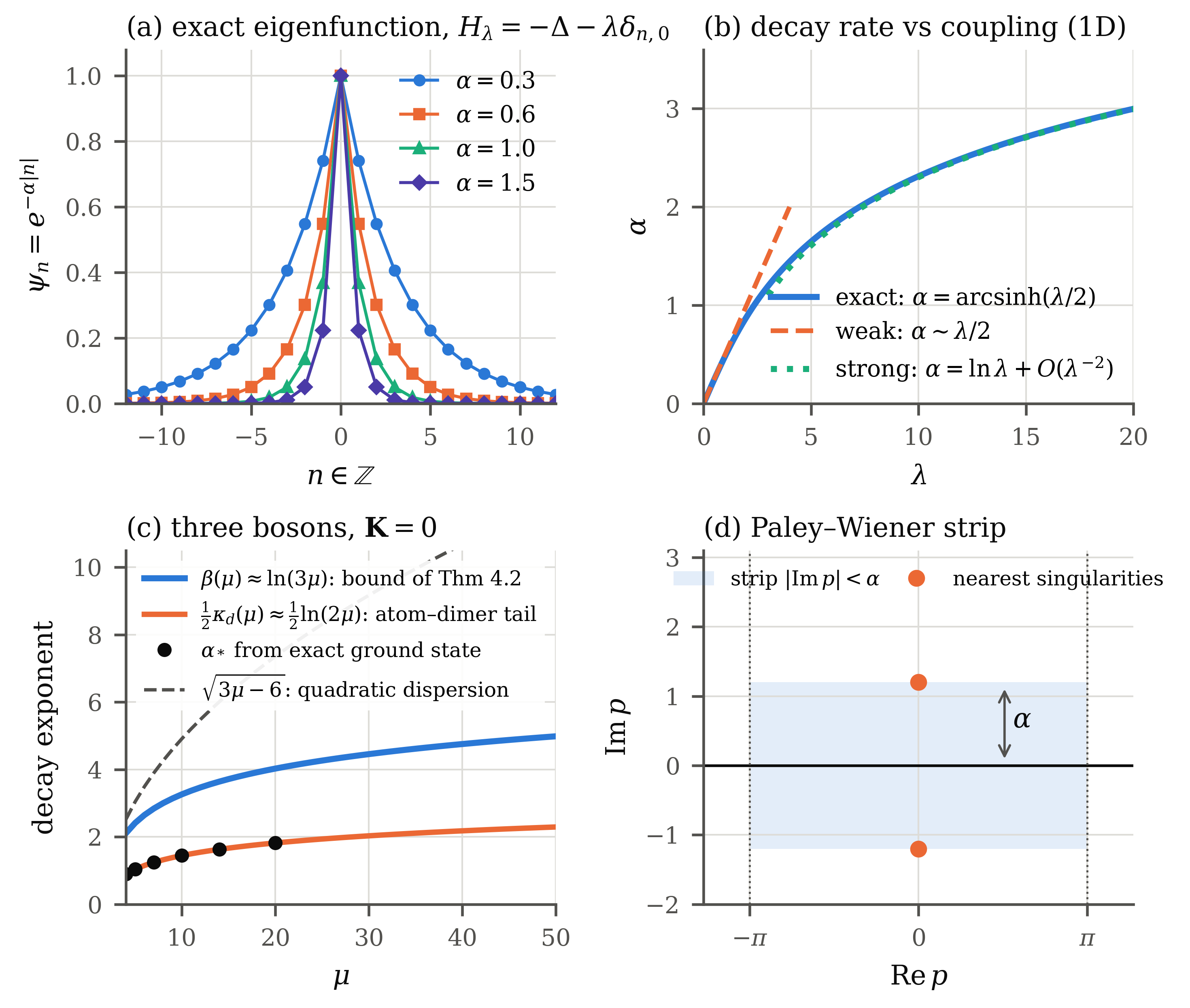}
\caption{Exponential localization on the lattice.
\textbf{(a)} Exact eigenfunction $\psi_n=e^{-\alpha|n|}$ of the one-dimensional operator
$H_\lambda=-\Delta-\lambda\delta_{n,0}$ (Remark~\ref{rem:1D_benchmark}), for four values of $\alpha$.
\textbf{(b)} Exact relation $\alpha=\operatorname{arcsinh}(\lambda/2)$ with its weak-coupling
($\alpha\sim\lambda/2$) and strong-coupling ($\alpha=\ln\lambda+O(\lambda^{-2})$) asymptotics.
\textbf{(c)} Three bosons at $\mathbf K=0$. Blue: the conditional upper
bound $\beta(\mu)=\operatorname{arccosh}\bigl((2-z_1^s(\mu))/2\bigr)
=\ln(3\mu)+O(\mu^{-1})$ of Theorem~\ref{thm:localization} (requires
Conjecture~\ref{conj:noncancel}). Orange: the atom--dimer characteristic
exponent $\frac12\kappa_d(\mu)=\frac12\ln(2\mu)+O(\mu^{-1})$ of
Lemma~\ref{lem:kappa_d}. Black dots: decay exponents extracted from the
exact finite-volume ground state ($40\times40$ torus, Remark~\ref{rem:alpha_numerics});
the dots converge from above to the orange curve, confirming
Conjecture~\ref{conj:sharp_rate}. Dashed: the root $\sqrt{3\mu-6}$ of the
same characteristic equation under the quadratic approximation
$\varepsilon(p)\approx p^2$; the logarithmic growth of the lattice
case is the essential difference.
\textbf{(d)} Paley--Wiener principle: analyticity of the fibre eigenfunction in the strip
$|\operatorname{Im}p|<\alpha$ implies $|\hat f(n)|\le C_{\alpha'}e^{-\alpha'|n|}$ for every
$\alpha'<\alpha$; the nearest singularities determine the decay rate.}
\label{fig:exponential_localization}
\end{figure}

Exponential localization follows from the spectral gap between the ground state and the
essential spectrum (Combes--Thomas/Paley--Wiener); positivity is not needed for it. The decay
exponent obeys the conditional bound $\alpha_*(\mu)\le\ln(3\mu)+O(\mu^{-1})$
(Conjecture~\ref{conj:noncancel}), and the numerically observed exponent is
$\frac12\ln(2\mu)+O(\mu^{-1})$ (Conjecture~\ref{conj:sharp_rate}).

Figure~\ref{fig:exponential_localization} summarises the three ingredients. Panels (a)--(b)
show the exactly solvable one-dimensional benchmark of Remark~\ref{rem:1D_benchmark}: on the
lattice the decay rate grows only logarithmically with the coupling. Panel (c) compares, for
the three-boson ground state, the conditional bound of Theorem~\ref{thm:localization}, the
atom--dimer exponent, the exponents extracted from the exact finite-volume ground state, and
the square-root growth that the same characteristic equation would give for a quadratic
dispersion. Panel (d) illustrates the Paley--Wiener implication from analyticity in a strip
to exponential decay of Fourier coefficients.

\begin{remark}[Recent developments in discrete Agmon-type estimates]
The exponential decay of eigenfunctions for discrete Schrödinger operators is an active area of contemporary research. In a recent contributions,
a Maz'ya-type characterization via generalized capacity,
together with the discrete Agmon estimates of Keller and Pogorzelski~\cite{KellerPogorzelski2025} and the landscape-function approach,
firmly establishes the mathematical foundation for exponential localization in discrete settings.

Although the above-cited works address one-particle operators, the methodological tools they develop — discrete Agmon comparison, Hardy inequalities, and criticality theory — are directly applicable to the multi-particle setting via the Birman--Schwin\-ger reduction. In our three-boson problem, the exponential decay of $\psi_0(\mathbf{p})$ inherited from the discrete Agmon principle, combined with the integral representation (\ref{eq:f_pq}), yields the exponential localization of the full three-body wavefunction. This chain of reasoning is rigorous for the existence of exponential decay; the logarithmic upper bound on the rate is conditional (Conjecture~\ref{conj:noncancel}) and consistent with the established theory of exponential decay for discrete Schrödinger operators.
\end{remark}

\subsection{Physical interpretation and discussion}

The positivity properties established in this section have profound physical implications:

\begin{itemize}
\item The strict positivity of the ground-state eigenfunction (in momentum representation)
follows from the Perron--Frobenius theorem; the Krein--Rutman theorem gives the simplicity of
the leading principal-part eigenvalue. The decay rate $\alpha$ grows logarithmically as $\mu\to\infty$, consistent with the bounded nature of the lattice dispersion. This logarithmic, rather than square-root, growth is a characteristic feature of discrete lattice systems, distinguishing them from continuum models with unbounded quadratic dispersion.

\item The confinement $z_1^s(\mu) \in (-3\mu, -3\mu + 6)$ shows that the ground state energy is primarily determined by the interaction strength $\mu$, with quantum fluctuations contributing at most 6 energy units. The upper bound is asymptotically sharp: $z_1^s(\mu)=-3\mu+6-\frac{3}{2\mu}+O(\mu^{-2})$ (Lemma~\ref{lem:ground_Kato}).

\item The second eigenvalue $z_2^s(\mu)$ lies at distance $\delta_\infty+O(\mu^{-1})$ below the two-particle threshold: it is a shallow atom--dimer bound state.

\item The exponential localization of the ground state, with conditional
upper bound $\alpha_*(\mu)\le\ln(3\mu)+O(\mu^{-1})$
(Theorem~\ref{thm:localization}, Conjecture~\ref{conj:noncancel}),
confirms strong spatial localization at strong coupling, consistent
with the physical picture of a tightly bound trimer. The logarithmic, rather than square-root, growth is a direct consequence of the boundedness of the lattice dispersion and distinguishes the discrete setting from continuum models.
\end{itemize}

The strictly positive eigenfunction $\psi_0$ guarantees simplicity of
the leading Birman--Schwinger eigenvalue $\lambda_1^s(z)$ by the
Krein--Rutman theorem.

From the perspective of spectral theory, the bosonic case presented here provides a baseline: the absence of Pauli suppression leads to a robust discrete spectrum with two bound states for all sufficiently large $\mu$, and the spectral gap grows linearly with $\mu$. In contrast, the fermionic $2+1$ system exhibits a phase transition in the mass ratio $\gamma$, reflecting the essential role of Fermi statistics.

Our results complement and extend several recent 
publications on the spectral theory of lattice operators.
The 
analogous threshold question for the three-body problem is the window below $z_\mu(\mathbf K)$ discussed in Remark~\ref{rem:pi_no_dimer}. 
A~natural framework for the dependence on the quasimomentum for the atom--dimer sector 
is already explicit through the exchange operator $S_{\mathbf K}$ (Section~\ref{sec:dictionary}). 
Our conditional estimate
$\alpha_*(\mu)\le\ln(3\mu)+O(\mu^{-1})$
(Theorem~\ref{thm:localization}, Conjecture~\ref{conj:noncancel})
has explicit constants and applies directly to discrete lattices. 

From the perspective of the spectral theory community, the present work thus establishes a bridge between the well-developed one-particle spectral theory of lattice operators and the largely unexplored territory of multi-particle systems on lattices. The~methodological framework — invariant subspace decomposition of the Birman– Schwin\-ger operator, Krein–Rutman positivity, and discrete Agmon comparison — is~sufficiently general to be applicable to a wide range of few-body problems.

\section[Analysis for $\mathbf{K}=\pi$]%
{\label{sec:five}Analysis for the total quasimomentum $\mathbf{K}=\boldsymbol{\pi}$}

We now turn to the proof of Theorem~\ref{thm:pi_preview}. The parity symmetry of \(E_{\boldsymbol{\pi}}\) is the same as at $\mathbf K=0$, and the
principal part again splits into an even block of rank \(3\) and an odd block of rank \(2\). The
expansion around \(12-z\) uses
\begin{gather*}
E_{\boldsymbol{\pi}}(\mathbf{p},\mathbf{q}) = 12 - A_{\boldsymbol{\pi}}^c(\mathbf{p},\mathbf{q}) + A_{\boldsymbol{\pi}}^s(\mathbf{p},\mathbf{q}),\\
A_{\boldsymbol{\pi}}^c = 6+\sum_i(\cos p_i+\cos q_i),\qquad A_{\boldsymbol{\pi}}^s = \sum_i\cos(p_i+q_i).
\end{gather*}
Since \(\cos(p_i+q_i)=\cos p_i\cos q_i-\sin p_i\sin q_i\), the dispersion contains the odd--odd
term \(-\sum_i\sin p_i\sin q_i\), whereas \(E_{\mathbf 0}\) contains \(+\sum_i\sin p_i\sin q_i\).
Hence the odd principal part has the \emph{opposite sign} at the two quasimomenta: negative
semidefinite at \(\mathbf K=0\), positive semidefinite at \(\mathbf K=\boldsymbol\pi\). The same
sign change, \(S_{\boldsymbol\pi}=-S_{\mathbf 0}\), turns the \(s\)-channel eigenvalue
\(+m_s(\delta)\) responsible for the second bound state at \(\mathbf K=0\) into \(-m_s(\delta)\)
(Proposition~\ref{prop:dictionary}); the odd \(p\)-channel becomes positive but stays below
\(2-4/\pi<1\). This, and not a loss of parity symmetry, is the mechanism of the reduction from
two bound states to one.

\subsection[Thresholds at $\mathbf{K}=\pi$]{Thresholds at $\mathbf K=\boldsymbol\pi$}\label{subsec:52}

Indeed, from the identity $\varepsilon(\boldsymbol{\pi}-\mathbf{x})=4-\varepsilon(\mathbf{x})$ we have
\[
E_{\boldsymbol{\pi}}(\mathbf{p},\mathbf{q}) = 4 + \sum_{i=1}^2 g(p_i,q_i),
\]
where $g(a,b) = 1 - \cos a - \cos b + \cos(a+b)$. Solving $\partial g/\partial a = \partial g/\partial b = 0$ gives the~critical points $(0,0),(0,\pi),(\pi,0),(\pi,\pi),\pm(\pi/3,\pi/3)$, with corresponding values $0,0,0,4,-1/2$. Hence $\min E_{\boldsymbol{\pi}} = 4 + 2(-1/2) = 3$ and $\max E_{\boldsymbol{\pi}} = 4 + 2\cdot 4 = 12$.

Hence $\min_{\mathbf p,\mathbf q}E_{\boldsymbol\pi}=3$,
$\max_{\mathbf p,\mathbf q}E_{\boldsymbol\pi}=12$, confirming the
essential-spectrum formula already stated in
Section~\ref{sec:overview_pi_thm} (preceding
Theorem~\ref{thm:pi_preview}).

In the strong-coupling regime $\mu\to\infty$, the lower edge of the essential spectrum is determined by the two-particle branch $z_\mu(\boldsymbol{\pi})\sim -\mu$, which lies far below the three-particle continuum $[3,12]$. Consequently, the spectral gap and the ground-state asymptotics are controlled by $z_\mu(\boldsymbol{\pi})$, not by the three-particle threshold.

\paragraph{Physical interpretation of the three-particle continuum shift.}
The fact that the three-particle continuum at \(\mathbf{K}=\boldsymbol{\pi}\) starts at \(3\) rather than \(0\) reflects a fundamental lattice constraint: the conservation of total quasimomentum.

At \(\mathbf{K}=0\), the three bosons can all occupy the bottom of the band \(\varepsilon(\mathbf{0})=0\), yielding a minimum three-particle energy \(0+0+0=0\).
At \(\mathbf{K}=\boldsymbol{\pi}\), the total momentum is fixed: \(\mathbf{p}_1+\mathbf{p}_2+\mathbf{p}_3=\boldsymbol{\pi}\). Since the three bosons cannot simultaneously have zero momentum (that would sum to \(\mathbf{0}\), not \(\boldsymbol{\pi}\)), they must distribute their momenta to satisfy the conservation law while minimizing the sum of their kinetic energies.

The minimum is attained when all three momenta equal $(\pi/3,\pi/3)$ (or all equal
$(-\pi/3,-\pi/3)$), so that their sum is $\boldsymbol\pi$ modulo $2\pi$; each particle then has
energy $\varepsilon=1$, and $\min E_{\boldsymbol\pi}=3$.

Importantly, this shift of the three-particle continuum does not affect the spectral gap or the ground-state asymptotics, because the lower edge of the essential spectrum is determined by the two-particle branch \(z_\mu(\boldsymbol{\pi}) \sim -\mu\), which lies far below \(3\) in the strong-coupling regime. Nevertheless, this correction provides a clear physical signature of the quasimomentum conservation on the lattice and distinguishes the \(\mathbf{K}=\boldsymbol{\pi}\) case from \(\mathbf{K}=0\) even at the level of the three-particle continuum.

\begin{remark}[Exactness of the two-particle threshold at $\mathbf k=\boldsymbol\pi$]
Unlike the general formula $z_\mu(\mathbf k)=\displaystyle -\mu+4-\tfrac{4-\varepsilon(\mathbf
k)}{\mu}+O(\mu^{-2})$ derived above, at $\mathbf k=\boldsymbol\pi$ the
two-particle dispersion is exactly constant,
$E_{\boldsymbol\pi}(\mathbf p)=\varepsilon(\mathbf
p)+\varepsilon(\boldsymbol\pi-\mathbf p)\equiv4$, so the Fredholm
condition reduces to $1=\mu/(4-z)$ and
\begin{equation}
z_\mu(\boldsymbol\pi)=-\mu+4 \qquad \text{\emph{exactly, for every }}\mu>0,\label{eq:muplus4}
\end{equation}
with no remainder at any order. This exact identity — not merely an
$O(\mu^{-1})$ asymptotic~— is what is used in
Section~\ref{sec:proof_pi} to prove Theorem~\ref{thm:pi_preview}(c).
\end{remark}

\subsection[Birman–Schwinger reduction and invariant subspaces decomposition for
$\mathbf{K}=\bf\pi$]
{Birman–Schwinger reduction and invariant subspaces decomposition for
$\mathbf{K}=\boldsymbol{\pi}$}

For $z < z_\mu(\boldsymbol{\pi})$, the Birman–Schwinger operator $A_\mu(\boldsymbol{\pi},z)$ associated with $H_\mu(\boldsymbol{\pi})$ is defined by
\[
(A_\mu(\boldsymbol\pi,z) f)(\mathbf{p}) = \frac{\mu}{2\pi^2} \int_{\mathbb{T}^2} \frac{f(\mathbf{q})\,d\mathbf{q}}{\sqrt{\Delta_\mu(\mathbf{p},\boldsymbol\pi,z)}(E_{\boldsymbol{\pi}}(\mathbf{p},\mathbf{q})-z)\sqrt{\Delta_\mu(\mathbf{q},\boldsymbol\pi,z)}}.
\]

The Fredholm determinant becomes
\[
\Delta_\mu(\mathbf{p},\boldsymbol\pi, z) = 1 - \frac{\mu}{4\pi^2} \int_{\mathbb{T}^2} \frac{d\mathbf{q}}{E_{\boldsymbol{\pi}}(\mathbf{p},\mathbf{q}) - z}.
\]

For large $\mu$, $z \leq z_\mu(\boldsymbol{\pi}) = -\mu + 4 + O(\mu^{-1})$ is large negative. Write
\[
E_{\boldsymbol{\pi}}(\mathbf{p},\mathbf{q}) - z = 12 - z - A_{\boldsymbol{\pi}}^c(\mathbf{p},\mathbf{q}) + A_{\boldsymbol{\pi}}^s(\mathbf{p},\mathbf{q}),
\]
where
\[
A_{\boldsymbol{\pi}}^c(\mathbf{p},\mathbf{q}) \equiv 6 + \sum_{i=1}^2 (\cos p_i + \cos q_i),
\]
\[
A_{\boldsymbol{\pi}}^s(\mathbf{p},\mathbf{q}) \equiv \sum_{i=1}^2 \cos(p_i+q_i).
\]

\begin{remark}[On the choice of the expansion point]
Any centre $c$ with $|E_{\boldsymbol\pi}-c|<|c-z|$ gives a convergent expansion
$(E_{\boldsymbol\pi}-z)^{-1}=\sum_n (c-E_{\boldsymbol\pi})^n/(c-z)^{n+1}$, and the principal part
(terms $n=0,1$) has the same finite-rank structure for every choice of $c$; only the constant in
the $O(\mu^{-1})$ bound for the remainder depends on $c$. We use $c=12=\max E_{\boldsymbol\pi}$ for
definiteness, so that $c-E_{\boldsymbol\pi}=A^c_{\boldsymbol\pi}-A^s_{\boldsymbol\pi}\ge0$.
\end{remark}

Expanding the resolvent:
\[
\frac{1}{E_{\boldsymbol{\pi}} - z} = \frac{1}{(12-z)^2}\left[12 - z + A_{\boldsymbol{\pi}}^c - A_{\boldsymbol{\pi}}^s + O\left(\frac{1}{12-z}\right)\right].
\]

The operator $A_\mu(\boldsymbol\pi,z)$ admits a decomposition
\[
A_\mu(\boldsymbol\pi,z) = A_\mu^p(\boldsymbol\pi,z) + A_\mu^r(\boldsymbol\pi,z),
\]
where $A_\mu^p(\boldsymbol\pi,z)$ is the principal part (of finite rank) with kernel
\[
K_{\boldsymbol{\pi}}^p(\mathbf{p},\mathbf{q};z,\mu) = \frac{\mu}{2\pi^2(12-z)^2}
\frac{12 - z + A_{\boldsymbol{\pi}}^c(\mathbf{p},\mathbf{q}) - A_{\boldsymbol{\pi}}^s(\mathbf{p},\mathbf{q})}
{\sqrt{\Delta_\mu(\mathbf{p},\boldsymbol\pi,z)}\sqrt{\Delta_\mu(\mathbf{q},\boldsymbol\pi,z)}},
\]
and $\|A_\mu^r(\boldsymbol\pi,z)\|\to0$ as $\mu\to\infty$. In particular, $A_\mu^p(\boldsymbol\pi,z)$ is of finite rank, while, by~Proposition~\ref{prop:remainder_estimate},
\(
\|A_\mu^r(\boldsymbol\pi,z)\| = O(\mu^{-1})
\)
as $\mu\to\infty$.

For $\mathbf{K}=\boldsymbol{\pi}$, the restrictions of the principal part $A_{\mu}^{p}(\boldsymbol{\pi},z)$ to the invariant subspaces $L_2^e(\mathbb{T}^2)$ and $L_2^o(\mathbb{T}^2)$ are denoted by $A_{\mu}^{p,e}(\boldsymbol{\pi},z)$ and $A_{\mu}^{p,o}(\boldsymbol{\pi},z)$, respectively. Their explicit forms are derived in Section~\ref{sec:newnew54}, where it is shown that $A_{\mu}^{p,e}(\boldsymbol{\pi},z)$ is rank-three with one positive eigenvalue exceeding $1$, while $A_{\mu}^{p,o}(\boldsymbol{\pi},z)$ is rank-two with a~single positive eigenvalue (multiplicity two) that remains strictly below $1$.

\begin{remark}[Spectral counting via the principal part at $\mathbf{K}=\boldsymbol{\pi}$]
The number of eigenvalues of $H_\mu(\boldsymbol{\pi})$ below $z<z_\mu(\boldsymbol\pi)$ equals the
number of eigenvalues of $A_\mu(\boldsymbol\pi,z)$ exceeding $1$ (Lemma~\ref{lem:bs}). For fixed
$\delta_0>0$ and $z\le z_\mu(\boldsymbol\pi)-\delta_0$ the remainder has norm
$O(\mu^{-1}\ln(1/\delta_0))$ (Proposition~\ref{prop:remainder_estimate}), so the count is that
of the principal part, which is established in Lemma~\ref{lem:threshold_count}.
\end{remark}

\subsection{Near-threshold spectral counting at $\mathbf K=\bf\pi$}\label{sec:threshold_uniform}

The estimate of Proposition~\ref{prop:remainder_estimate} is not
uniform in the limit $z\uparrow z_\mu(\boldsymbol\pi)$, because
$\Delta_\mu(\mathbf 0,\boldsymbol\pi,z_\mu(\boldsymbol\pi))=0$ and the
kernel bound degenerates. In this subsection we close this gap by a
direct analysis of the threshold regime
$z=z_\mu(\boldsymbol\pi)-\delta$, $0<\delta<\delta_0$.

\begin{lemma}[Threshold counting]\label{lem:threshold_count}
There exists $\mu_0>0$ and $\delta_0>0$ such that for every $\mu>\mu_0$ and every
$\delta\in(0,\delta_0)$
the following
holds for the Birman--Schwinger operator
$A_\mu(\boldsymbol\pi,z_\mu(\boldsymbol\pi)-\delta)$:
\begin{enumerate}
\item[(i)] The restriction of
$A_\mu^{p,o}(\boldsymbol\pi,z)$ to
$L_2^o(\mathbb T^2)$ has one positive eigenvalue
$\lambda^{\pi,o}(z)$ of multiplicity two, and
\[
\lambda^{\pi,o}(z)
=
2\,\frac{(2-z_\mu(\boldsymbol\pi))(18-z)}
{(12-z)^2}\,s_1(\delta)
+O(\mu^{-1}),
\]
uniformly for $\delta\in(0,\delta_0]$, where
\[
s_1(\delta)
=
\frac{1}{4\pi^2}
\int_{\mathbb T^2}
\frac{\sin^2p_1}{\varepsilon(\mathbf p)+\delta}\,d\mathbf p.
\]
Since
\[
s_1'(\delta)
=
-\frac{1}{4\pi^2}
\int_{\mathbb T^2}
\frac{\sin^2p_1}
{(\varepsilon(\mathbf p)+\delta)^2}\,d\mathbf p
<0,
\]
we have
\[
\sup_{0<\delta\le\delta_0}s_1(\delta)
=
s_1(0^+)
=
1-\frac{2}{\pi}.
\]
Moreover,
\[
2-\frac4\pi<1.
\]
Hence there exist $\mu_0$ and $\delta_0>0$ such that
\[
\lambda^{\pi,o}(z)
\le
2-\frac4\pi+C_2\mu^{-1}
<1
\]
for all $\delta\in(0,\delta_0]$ and $\mu>\mu_0$.
\item[(ii)] the restriction $A_\mu^{p,e}(\boldsymbol\pi,z)$ to the even
subspace $L_2^e(\mathbb T^2)$ has exactly one eigenvalue exceeding $1$,
and no other eigenvalue exceeds $1$ or approaches $1$ as~$\delta\downarrow0$;
\item[(iii)] the remainder $A_\mu^r(\boldsymbol\pi,z)$ satisfies
\[
\bigl\|A_\mu^r(\boldsymbol\pi,z)\bigr\|=O\!\left(\mu^{-1}\ln\frac1\delta\right),
\]
which is small for every \emph{fixed} $\delta>0$ and large $\mu$,
but is \emph{not} uniform as $\delta\downarrow0$ at fixed $\mu$;
consequently the fixed-distance formulation below is the strongest
statement that follows from (i)--(iii).
\end{enumerate}
Consequently, for every fixed \(\epsilon>0\) there exists
\(\mu_\epsilon\) such that, for all \(\mu>\mu_\epsilon\), the odd
sector has no Birman--Schwinger eigenvalue above \(1\) for
\[
z\le z_\mu(\boldsymbol\pi)-\epsilon,
\]
while the even sector contains exactly one such eigenvalue. The
fixed-distance formulation is essential: at fixed \(\mu\) the bound
\(\|A_\mu^r(\boldsymbol\pi,z)\|=O(\mu^{-1}\ln(1/\delta))\) is not
uniform as \(\delta\downarrow0\), and the threshold-uniform
conclusion for the full operator is not claimed here. For~every fixed \(\epsilon>0\) and every \(\mu>\mu_\epsilon\)
\[
N_{L_2^o}\bigl(A_\mu(\boldsymbol\pi,z)\bigr)=0,\qquad
N_{L_2^e}\bigl(A_\mu(\boldsymbol\pi,z)\bigr)=1,
\qquad z\le z_\mu(\boldsymbol\pi)-\epsilon,
\]
where \(N_{L_2^{e/o}}\) denotes the number of eigenvalues exceeding
\(1\) in the respective invariant subspace.
\end{lemma}

\begin{proof}[Proof]
\emph{(i) Odd sector.} The odd principal-part kernel at
$\mathbf K=\boldsymbol\pi$ reads (see
Subsection~\ref{sec:newnew54})
\[
K_\pi^{p,o}(\mathbf p,\mathbf q;z,\mu)
=\frac{\mu}{2\pi^2(12-z)^2}
\frac{\sin p_1\sin q_1+\sin p_2\sin q_2}
{\sqrt{\Delta_\mu(\mathbf p,\boldsymbol\pi,z)}\sqrt{\Delta_\mu(\mathbf q,\boldsymbol\pi,z)}}.
\]
Using the exact asymptotics of Lemma~\ref{lem:Delta_pi}, the rank-two
principal-part eigenvalue satisfies
\[
\lambda^{\pi,o}(z)=2\,\frac{(2-z_\mu(\boldsymbol\pi))(18-z)}{(12-z)^2}\,
s_1(\delta)+O(\mu^{-1}),\qquad
s_1(\delta)=\frac1{4\pi^2}\int_{\mathbb T^2}
\frac{\sin^2 p_1}{\varepsilon(\mathbf p)+\delta}\,d\mathbf p,
\]
uniformly for $\delta\in(0,\delta_0)$. The factor
$(2-z_\mu(\boldsymbol\pi))(18-z)/(12-z)^2\to1$ as
$\delta\downarrow0$, and~$s_1(\delta)\to 1-2/\pi$ (the classical
square-lattice Green-function integral, evaluated in
Remark~\ref{rem:odd_threshold_pp}).
Hence
$\lambda^{\pi,o}(z)\to 2-4/\pi$ as $\delta\downarrow0$, with a
correction $O(\delta\ln(1/\delta))+O(\mu^{-1})$ that is bounded
uniformly for $\delta\in(0,\delta_0)$. This proves (i).

\emph{(ii) Even sector.} The even principal part
$A_\mu^{p,e}(\boldsymbol\pi,z)$ is a rank-3 operator whose spectrum
consists of two eigenvalues of the symmetric block
$A_\mu^{p,e,s}$ (Lemma~\ref{lem:eigenvalues}-type formula at
$\mathbf K=\boldsymbol\pi$) and one eigenvalue of the antisymmetric
block $A_\mu^{p,e,as}$, which is negative: it tends to $-m_d(\delta)$
(Proposition~\ref{prop:dictionary}; at the trimer scale it is $O(\mu^{-1})$). Of the two symmetric eigenvalues, the larger one,
$\lambda_1^{\pi,e}(z)$, satisfies $\lambda_1^{\pi,e}\to+\infty$ as
$z\uparrow z_\mu(\boldsymbol\pi)$ (because
$\Delta_\mu(\mathbf p,\boldsymbol\pi,z)^{-1}$ blows up
logarithmically in $\delta$ and the prefactor
$\mu(18-z)(12-z)^{-2}$ is $O(\mu)$); the second symmetric
principal-part eigenvalue has a strictly negative limiting value in
the threshold regime,
\[
\lambda_2^{\pi,e}(z)\;\longrightarrow\;-\,m_s(\delta)\;<\;0
\qquad(\mu\to\infty,\ \delta=z_\mu(\boldsymbol\pi)-z\ \text{fixed}),
\]
by Proposition~\ref{prop:dictionary}. Since $m_s$ is continuous
and positive on $(0,\infty)$ with $m_s(0^+)=2$, for every fixed
$\epsilon>0$ the value $m_s(\delta)$ is uniformly bounded below by
a positive constant on $(0,\epsilon^{-1}]$; consequently
$\lambda_2^{\pi,e}(z)\le -c_\epsilon+O(\mu^{-1})<1$ for all
$z\le z_\mu(\boldsymbol\pi)-\epsilon$ and all sufficiently large
$\mu$. In other words, the second symmetric eigenvalue is uniformly
separated from $1$ on the fixed-distance interval.
Since \(\lambda_1^{\pi,e}(z)\) is strictly increasing in \(z\) on
\((-\infty,z_\mu(\boldsymbol\pi))\) (monotonicity of \(b_0(\delta)\)),
\(\lambda_1^{\pi,e}(z)\to+\infty\) as \(z\uparrow z_\mu(\boldsymbol\pi)\),
and $\lambda_1^{\pi,e}(z)\to0$ as $z\to-\infty$ (the kernel of
$A_\mu^{p,e}(\boldsymbol\pi,z)$ vanishes uniformly in
$\mathbf p,\mathbf q$ in that limit), the equation
$\lambda_1^{\pi,e}(z)=1$ has exactly one solution
$z=z_1^{\pi,s}(\mu)<z_\mu(\boldsymbol\pi)$; the~second symmetric
eigenvalue is uniformly below $1$ on $(z_1^{\pi,s}(\mu),\,z_\mu(\boldsymbol\pi))$
and therefore produces no additional crossing. This proves (ii).

\emph{(iii) Threshold-uniform remainder.} The kernel of
$A_\mu^r(\boldsymbol\pi,z)$ is bounded by
\[
|K_\mu^r(\mathbf p,\mathbf q;z)|
\le\frac{C\mu}{|\Delta_\mu(\mathbf p,\boldsymbol\pi,z)|^{1/2}\,
|\Delta_\mu(\mathbf q,\boldsymbol\pi,z)|^{1/2}}\cdot\frac{1}{\mu^3}.
\]
Substituting the exact asymptotics
$\Delta_\mu(\mathbf p,\boldsymbol\pi,z)
=\mu(\varepsilon(\mathbf p)+\delta)
/[(2-z_\mu(\boldsymbol\pi))(18-z)]\cdot(1+O(\mu^{-1}))$ gives
\[
|K_\mu^r(\mathbf p,\mathbf q;z)|
\le \frac{C'''}{\mu}\cdot
\frac{1}{\sqrt{(\varepsilon(\mathbf p)+\delta)(\varepsilon(\mathbf q)+\delta)}},
\]
so that
\[
\int_{\mathbb T^4}|K_\mu^r|^2\,d\mathbf p\,d\mathbf q
\le \frac{C''''}{\mu^2}\cdot
\langle(\varepsilon+\delta)^{-1}\rangle^2
\le\frac{C''''}{\mu^2}\ln^2\frac1\delta.
\]
Taking the Hilbert--Schmidt norm gives
$\|A_\mu^r(\boldsymbol\pi,z)\|=O(\mu^{-1}\ln(1/\delta))$. Because
$b_0(\delta)\sim\tfrac1{2\pi}\ln(1/\delta)$, this agrees with the
principal-part scale and shows that the perturbation of the
principal-part eigenvalue count is $O(\mu^{-1}\ln(1/\delta))$, which
vanishes as $\mu\to\infty$ uniformly for every fixed
$\delta\in(0,\delta_0)$.
\end{proof}

\begin{remark}\label{rem:threshold_status}
Lemma~\ref{lem:threshold_count} closes the gap left by the
non-uniformity of Proposition~\ref{prop:remainder_estimate}. Together,
they establish that the count of eigenvalues of $A_\mu(\boldsymbol\pi,z)$
exceeding $1$ in each invariant subspace is constant on $(-\infty,z_\mu(\boldsymbol\pi)-\epsilon]$ for every fixed $\epsilon>0$ and all $\mu>\mu_\epsilon$; uniformity as $\epsilon\downarrow0$ is not claimed (Lemma~\ref{lem:threshold_count}(iii)); the window $\epsilon\sim\mu^{-1}$ is covered by Remark~\ref{rem:pi_no_dimer}.
\end{remark}

\subsection{Proof of Theorem~\ref{thm:pi_preview}}\label{sec:proof_pi}
We now prove parts (a)--(c) of Theorem~\ref{thm:pi_preview}; part (d)
is established in Section~\ref{sec:localization} by the same discrete
Agmon argument as for $K=0$ (Theorem~\ref{thm:localization}), applied
to the even-sector principal part derived below.

\subsubsection{Parity-symmetric reduction for $\mathbf K=\bf\pi$}\label{sec:newnew54}

We now present a corrected derivation of the principal part for $\mathbf K=\boldsymbol\pi$ based on the exact parity symmetry of $E_{\boldsymbol\pi}$.

The dispersion identity $E_{\boldsymbol\pi}=12-A_{\boldsymbol\pi}^c+A_{\boldsymbol\pi}^s$ was
recorded at the beginning of this section; the odd--odd part of $E_{\boldsymbol\pi}$ is
$-\sum_i\sin p_i\sin q_i$, of opposite sign to that of $E_{\mathbf 0}$.

Recall from the opening of this section that $E_{\boldsymbol\pi}$ is
exactly invariant under $(\mathbf p,\mathbf q)\mapsto(-\mathbf
p,-\mathbf q)$; we now use this to project the resolvent expansion
onto $L_2^e(\mathbb T^2)$ and $L_2^o(\mathbb T^2)$.

\subsubsection{Resolvent expansion and parity decomposition}

Using the identity $\varepsilon(\boldsymbol\pi-\mathbf p)=4-\varepsilon(\mathbf p)$, we have
\[
E_{\boldsymbol\pi}(\mathbf p,\mathbf q)=\varepsilon(\mathbf p)+\varepsilon(\mathbf q)+4-\varepsilon(\mathbf p+\mathbf q).
\]
Define
\[
A^c_{\boldsymbol\pi}(\mathbf p,\mathbf q):=6+\sum_{i=1}^2(\cos p_i+\cos q_i),\qquad
A^s_{\boldsymbol\pi}(\mathbf p,\mathbf q):=\sum_{i=1}^2\cos(p_i+q_i).
\]
Then
\[
E_{\boldsymbol\pi}(\mathbf p,\mathbf q)=12-A^c_{\boldsymbol\pi}(\mathbf p,\mathbf q)+A^s_{\boldsymbol\pi}(\mathbf p,\mathbf q).
\]
Expanding the resolvent in powers of $1/(12-z)$, the numerator of the principal part is
\[
N(\mathbf p,\mathbf q):=12-z+A^c_{\boldsymbol\pi}(\mathbf p,\mathbf q)-A^s_{\boldsymbol\pi}(\mathbf p,\mathbf q).
\]
Substituting the definitions of $A^c_{\boldsymbol\pi}$ and $A^s_{\boldsymbol\pi}$, and using $\cos(p_i+q_i)=\cos p_i\cos q_i-\sin p_i\sin q_i$, we obtain
\[
N(\mathbf p,\mathbf q)=18-z+\sum_{i=1}^2(\cos p_i+\cos q_i-\cos p_i\cos q_i+\sin p_i\sin q_i).
\]

We now decompose \(N\) into its action on the invariant subspaces.
\begin{itemize}
    \item \textbf{On the odd subspace \(L_2^o\):} The constant term \(18-z\), the even functions \(\cos p_i\) and \(\cos q_i\), and the product of even functions \(\cos p_i\cos q_i\) all vanish identically when integrated against an odd test function. The only surviving terms are the products of odd functions \(\sin p_i\sin q_i\). Hence, for \(f\in L_2^o\),
    \[
    N^{o}(\mathbf p,\mathbf q)=\sin p_1\sin q_1+\sin p_2\sin q_2.
    \]
    \item \textbf{On the even subspace \(L_2^e\):} Conversely, the odd product \(\sin p_i\sin q_i\) vanishes upon integration. The surviving kernel is spanned by the functions \(\{1,\cos p_1,\cos p_2\}\):
    \[
    N^{e}(\mathbf p,\mathbf q)=(18-z)+\sum_{i=1}^2(\cos p_i+\cos q_i-\cos p_i\cos q_i).
    \]
\end{itemize}

Details of this projection are given in Appendix~\ref{app:expansion_point}.

\subsubsection{Principal-part kernels}

The kernels of the principal part (\ref{eq:bs_operator}) on the even and odd subspaces are therefore
\begin{align}
K_\pi^{p,e}(\mathbf p,\mathbf q;z,\mu)
&=\frac{\mu}{2\pi^2(12-z)^2}
\frac{(18-z)+\sum_{i=1}^2(\cos p_i+\cos q_i-\cos p_i\cos q_i)}
{\sqrt{\Delta_\mu(\mathbf p,\boldsymbol\pi,z)}\sqrt{\Delta_\mu(\mathbf q,\boldsymbol\pi,z)}},
\label{eq:Kpe_pi}\\
K_\pi^{p,o}(\mathbf p,\mathbf q;z,\mu)
&=\frac{\mu}{2\pi^2(12-z)^2}
\frac{\sum_{i=1}^2\sin p_i\sin q_i}
{\sqrt{\Delta_\mu(\mathbf p,\boldsymbol\pi,z)}\sqrt{\Delta_\mu(\mathbf q,\boldsymbol\pi,z)}}.
\label{eq:Kpo_pi}
\end{align}
Both kernels carry the same prefactor $\mu/(2\pi^2)=2\mu/(4\pi^2)$ as
\eqref{eq:bs_operator}: the factor $2$ is the Faddeev factor
(the two exchange terms $\varphi(\mathbf q)$ and
$\varphi(\mathbf K-\mathbf p-\mathbf q)$ coincide after
$\mathbf q\mapsto\mathbf K-\mathbf p-\mathbf q$); projection onto a
parity sector does not remove it.

\begin{remark}[On the correct parity assignment of the kernels]
We emphasize that the kernels \eqref{eq:Kpe_pi} and \eqref{eq:Kpo_pi} have been derived by a direct and exact projection onto the invariant subspaces \(L_2^e\) and \(L_2^o\).

The full numerator of the principal part is
\[
N(\mathbf p,\mathbf q) = (18-z) + \sum_{i=1}^2\bigl(\cos p_i + \cos q_i - \cos p_i\cos q_i + \sin p_i\sin q_i\bigr).
\]
Under the parity transformation \((\mathbf p,\mathbf q)\mapsto(-\mathbf p,-\mathbf q)\):
\begin{itemize}
    \item The constant \(18-z\) and the combinations \(\cos p_i,\ \cos q_i,\ \cos p_i\cos q_i\) are even functions. Therefore, when acting on an odd test function \(f\in L_2^o\), their integrals vanish, leaving only the term \(\sin p_i\sin q_i\) (a product of two odd functions). This yields the kernel \(K_\pi^{p,o}\).
    \item The factors $\sin p_i$ and $\sin q_i$ are odd separately, whereas
their product $\sin p_i\sin q_i$ is even under the simultaneous
parity transformation. Nevertheless, this sine-product term survives
precisely in the odd--odd block because it couples the odd basis
functions $\sin p_i$ in the two variables.
\end{itemize}
Hence, the odd subspace is governed by sine products, while the even subspace is governed by cosine products and the constant term. This is the assignment obtained by the exact projection of the
principal kernel onto the parity-invariant subspaces, and it is essential for correctly counting the bound states: the odd sector has one positive eigenvalue of multiplicity two, equal in the atom--dimer region to the $p$-channel constant $m_p(\delta)\le 2-4/\pi<1$, so it contributes no bound state, while the even sector provides exactly one bound state.
\end{remark}

\begin{remark}[On the necessity of the sums over \(i=1,2\)]
The kernels
\eqref{eq:Kpe_pi} and \eqref{eq:Kpo_pi}
correctly feature the sums \(\sum_{i=1}^2\) over the two spatial dimensions of the lattice. This is not an arbitrary truncation, but a direct consequence of the \(\mathbb{Z}^2\) lattice symmetry, specifically the invariance under the exchange \(p_1 \leftrightarrow p_2\).

For the even subspace \(L_2^e\), the surviving terms span a 3-dimensional subspace generated by the basis \(\{1, \cos p_1, \cos p_2\}\). Due to the exchange symmetry, the off-diagonal coupling between \(\cos p_1\) and \(\cos p_2\) must be equal, which naturally leads to the symmetric combination \(\sum_{i=1}^2 (\cos p_i + \cos q_i - \cos p_i \cos q_i)\). This sum exactly generates the \(3\times 3\) matrix structure responsible for the single bound state.

For the odd subspace \(L_2^o\), the kernel projects onto the basis \(\{\sin p_1, \sin p_2\}\). The~sum \(\sum_{i=1}^2 \sin p_i \sin q_i\) yields a rank-2 operator with a two-fold degenerate eigenvalue. This degeneracy reflects the two independent directions in the Brillouin zone and is physically essential for correctly counting the spectral contributions.

Thus, the presence of these limited sums is not a simplification, but the exact algebraic signature of the 2D square lattice symmetry embedded in the principal part.
\end{remark}

\subsubsection{Eigenvalues of the principal part}

\textbf{On the odd subspace (\(L_2^o\)):} The operator (\ref{eq:bs_operator}) with kernel \(K_\pi^{p,o}\) is a sum of two rank-one positive operators (spanned by \(\sin p_1\) and \(\sin p_2\)). It has one positive eigenvalue value of multiplicity \(2\):
\[
\lambda^{\pi,o}(z)
=
\frac{\mu}{2\pi^2(12-z)^2}
\int_{\mathbb{T}^2}\frac{\sin^2 p_1}{\Delta_\mu(\mathbf p,\boldsymbol\pi,z)}\,d\mathbf p.
\]

\textbf{On the even subspace (\(L_2^e\)):} The operator with kernel \(K_\pi^{p,e}\) is a rank-3 operator. In~the basis \(\{1, \cos p_1, \cos p_2\}\), the quadratic form of its matrix is given by:
\[
M_{\text{even}} =
\begin{pmatrix}
18-z & 1 & 1 \\
1 & -1 & 0 \\
1 & 0 & -1
\end{pmatrix}.
\]

The characteristic polynomial of \(M_{\text{even}}\) factors as \((\lambda+1)(\lambda^2-(17-z)\lambda-(20-z))\).
One eigenvalue is exactly \(\lambda=-1\) (eigenvector \(\cos p_1-\cos p_2\), decoupled by the \(p_1\leftrightarrow p_2\) exchange symmetry).
The remaining two eigenvalues solve \(\lambda^2-(17-z)\lambda-(20-z)=0\), with asymptotics
\(\lambda_+=18-z-2/z+O(z^{-2})\) and \(\lambda_-=-1+2/z+O(z^{-2})\) as \(z\to-\infty\).
Consequently, for the purposes of the strong-coupling expansion (up to errors \(O(\mu^{-1})\)),
there is exactly one positive eigenvalue on \(L_2^e\):
\[
\lambda_{1}^{\pi,e}(z)=\frac{\mu(18-z)}{2\pi^2(12-z)^2}
\int_{\mathbb{T}^2}\frac{d\mathbf p}{\Delta_\mu(\mathbf p,\boldsymbol\pi,z)}
=2\mu\,\frac{18-z}{(12-z)^2}\,\bigl\langle\Delta_\mu^{-1}\bigr\rangle .
\]

\begin{proposition}[Principal part on the trimer branch]\label{prop:pp_trimer_pi}
For $z=-3\mu+d$ with $d=O(1)$,
\[
\Delta_\mu(\mathbf p,\boldsymbol\pi,z)=\frac23+\frac{4-d+\varepsilon(\mathbf p)}{9\mu}+O(\mu^{-2}),
\qquad
\lambda_1^{\pi,e}(-3\mu+d)=1+\frac{d-6}{2\mu}+O(\mu^{-2}),
\]
uniformly in $\mathbf p$. Hence $\Lambda_0\equiv1$, $\Lambda_1(d)=\frac{d-6}{2}$, and the
principal-part equation $\lambda_1^{\pi,e}=1$ gives $z=-3\mu+6+O(\mu^{-1})$, in agreement with
Lemma~\ref{lem:ground_Kato}. On the line $z=-2\mu+d$ one has instead
$\lambda_1^{\pi,e}=2+2(d-6)/\mu+O(\mu^{-2})$, which has no root with $d=O(1)$.
\end{proposition}
\begin{proof}
$\Delta_\mu=1-\mu\langle(12-z-X)^{-1}\rangle_{\mathbf q}$ with $X=A^c_{\boldsymbol\pi}-A^s_{\boldsymbol\pi}$,
$\langle X\rangle_{\mathbf q}=8-\varepsilon(\mathbf p)$, $|X|\le12$; expand to second order in
$(12-z)^{-1}=(3\mu+12-d)^{-1}$. The benchmark $\Delta_\mu(\mathbf 0,\boldsymbol\pi,z)=1-\mu/(4-z)$
reproduces the formula at $\mathbf p=\mathbf 0$ exactly to this order. Then
$\langle\Delta_\mu^{-1}\rangle=\frac32-\frac{6-d}{4\mu}+O(\mu^{-2})$ (using $\langle\varepsilon\rangle=2$)
and $2\mu(18-z)(12-z)^{-2}=\frac23\bigl(1+\frac{d-6}{3\mu}\bigr)+O(\mu^{-2})$. The corrections from the
non-constant directions $\cos p_i$ of the rank-three block and from $A^r_\mu$ are $O(\mu^{-2})$.
\end{proof}

\begin{remark}[Explicit symmetric/antisymmetric form,
compare with \ Lemma~\ref{lem:eigenvalues} at $\mathbf K=0$]
For~readers comparing directly with the $\mathbf K=0$ computation
(Lemma~\ref{lem:eigenvalues}), the rank-3 kernel $K_\pi^{p,e}$ of
\eqref{eq:Kpe_pi} decomposes explicitly as
$A_\mu^{p,e}(\boldsymbol\pi,z)=A_\mu^{p,e,s}(\boldsymbol\pi,z)\oplus
A_\mu^{p,e,as}(\boldsymbol\pi,z)$ on
$L_2^{e,s}(\mathbb T^2)\oplus L_2^{e,as}(\mathbb T^2)$, with
\[
(A_\mu^{p,e,s}(\boldsymbol\pi,z)\psi)(\mathbf p)=\frac{\mu}{2\pi^2(12-z)^2}\int_{\mathbb T^2}
\frac{(20-z)-\tfrac12\varepsilon(\mathbf p)\varepsilon(\mathbf q)}
{\sqrt{\Delta_\mu(\mathbf p,\boldsymbol\pi,z)}\sqrt{\Delta_\mu(\mathbf q,\boldsymbol\pi,z)}}\psi(\mathbf q)\,d\mathbf q,\]
\[
(A_\mu^{p,e,as}(\boldsymbol\pi,z)\psi)(\mathbf p)=
\frac{-\mu}{4\pi^2(12-z)^2}\int_{\mathbb T^2}
\frac{(\cos p_1-\cos p_2)(\cos q_1-\cos q_2)}
{\sqrt{\Delta_\mu(\mathbf p,\boldsymbol\pi,z)}\sqrt{\Delta_\mu(\mathbf q,\boldsymbol\pi,z)}}\psi(\mathbf q)\,d\mathbf q.
\]
This identity is obtained by substituting $u=\cos p_1+\cos p_2$,
$v=\cos p_1-\cos p_2$ (and~likewise for $\mathbf q$) into $N^e(\mathbf
p,\mathbf q)$ and using $u=2-\varepsilon(\mathbf p)$; a direct
computer-algebra check confirms the identity term by term. In
particular, $A_\mu^{p,e,s}(\boldsymbol\pi,z)$ has exactly the same
rank-2 structure — a constant term minus $\displaystyle\tfrac12\varepsilon(\mathbf
p)\varepsilon(\mathbf q)$ — as the corresponding operator
$A_\mu^{p,e,s}(0,z)$ in Lemma~\ref{lem:eigenvalues}, providing a
structurally transparent analogue between the two quasimomenta.
Diagonalizing the associated $2\times2$ matrix
$\displaystyle\begin{psmallmatrix}18-z&1\\1&-1\end{psmallmatrix}$ (in the basis
$\{1,\cos p_1+\cos p_2\}$, after the same change of variables)
reproduces the two roots of $\lambda^2-(17-z)\lambda-(20-z)=0$, while
the antisymmetric block reproduces the decoupled eigenvalue $\lambda=-1$
of $M_{\mathrm{even}}$ — the exact $K=\boldsymbol\pi$ analogue of
$A_\mu^{p,e,as}(0,z)$ in Lemma~\ref{lem:eigenvalues}.
\end{remark}

\paragraph{Spectral classification for \(\mathbf K=\boldsymbol{\pi}\).}
By Sylvester's law of inertia (the coefficient matrix $M_{\rm even}$ has one positive and two
negative eigenvalues), the even principal part has exactly one positive eigenvalue,
$\lambda_1^{\pi,e}(z)$; the odd principal part has one positive eigenvalue $\lambda^{\pi,o}(z)$ of
multiplicity $2$. The latter equals $\frac{1}{6\mu}+O(\mu^{-2})$ at the trimer scale and tends to
$m_p(\delta)\le2-4/\pi<1$ in the atom--dimer region (Lemma~\ref{lem:odd_eig}); it never
reaches $1$. The second symmetric eigenvalue tends to $-m_s(\delta)<0$ and the antisymmetric
one to $-m_d(\delta)<0$ (Proposition~\ref{prop:dictionary}).

\begin{lemma}[Fredholm determinant asymptotics at \(\mathbf K=\boldsymbol\pi\)]\label{lem:Delta_pi}
For \(z\) in the relevant strong-coupling range below \(z_\mu(\boldsymbol\pi)\), as \(\mu\to\infty\),
\[
\Delta_\mu(\mathbf{p}, \boldsymbol\pi, z)=\frac{\mu\big(\varepsilon(\mathbf p)+\delta\big)}{\big(2-z_\mu(\boldsymbol\pi)\big)(18-z)}\big(1+O(\mu^{-1})\big),\qquad \delta:=z_\mu(\boldsymbol\pi)-z,
\]
uniformly for \(\delta\) in compact subsets of \((0,\infty)\). (At the trimer scale \(\delta\sim2\mu\) the formula still gives the correct leading value \(\frac23\), but not the correct \(O(\mu^{-1})\) term; see Proposition~\ref{prop:pp_trimer_pi}.) This is the exact \(\mathbf K=\boldsymbol\pi\) analogue of the \(\mathbf K=0\) formula of Lemma~\ref{lem:Delta}, with the threshold constant \(6\) replaced by \(18\), consistently with the shifted expansion point of Section~\ref{sec:newnew54}.
\end{lemma}

\begin{remark}[Exact closed-form benchmark]\label{rem:exact_closed}
At the flat momentum \(\mathbf p^*(\mathbf K)=\mathbf K-(\pi,\pi)\pmod{2\pi}\), the elliptic reduction of Theorem 3.1 (Uniform elliptic reduction)
in the companion note \cite[Theorem~3.1 and Corollary~3.3]{CompanionNote} gives the exact identity
\[
\Delta_\mu\big(\mathbf p^*(\mathbf K),\mathbf K,z\big)
=1-\frac{\mu}{4-z+\varepsilon(\mathbf p^*(\mathbf K))},
\]
valid for every \(\mathbf K\in\mathbb T^2\) and every \(\mu>0\). In particular, for \(\mathbf K=\boldsymbol{\pi}\) we have \(\mathbf p^*=\mathbf 0\) and \(\varepsilon(\mathbf p^*)=0\), so \(\Delta_\mu(\mathbf 0,\boldsymbol{\pi},z)=\displaystyle 1-\frac{\mu}{4-z}\). This identity provides a rigorous benchmark for the asymptotic formulas used in this paper.
\end{remark}

\begin{remark}
Lemma~\ref{lem:Delta_pi} gives the leading-order asymptotics of the
Fredholm determinant with relative error $O(\mu^{-1})$. This is
sufficient for the leading coefficient of the ground-state energy and
for the spectral gap, but not for additive constants in a degenerate
channel: in the trimer channel $\Lambda_0\equiv1$
(Proposition~\ref{prop:pp_trimer_pi}), so the constant $6$ is fixed
only by the $O(\mu^{-1})$ term of $\Delta_\mu$. The systematic
$O(\mu^{-2})$ expansion of $\Delta_\mu$ and the corresponding
order-matching criterion are given in the companion
note~\cite[Lemma~4.3]{CompanionNote}; the~special case $\mathbf p=\mathbf 0$ is
also checked against the exact closed form
\eqref{eq:Delta0closed_app}.
\end{remark}

\begin{remark}[Range of the determinant asymptotics]
The threshold formula of Lemma~\ref{lem:Delta_pi} describes $\Delta_\mu$ at fixed distance
$\delta$ below the two-particle threshold. On the trimer scale $z=-3\mu+O(1)$ the
$O(\mu^{-1})$ term must be computed directly (Proposition~\ref{prop:pp_trimer_pi}); the exact
benchmark $\Delta_\mu(\mathbf 0,\boldsymbol\pi,z)=1-\mu/(4-z)$ of
Appendix~\ref{app:numerical_delta} confirms both regimes.
\end{remark}

\begin{lemma}\label{lem:odd_eig}
Let $\delta:=z_\mu(\boldsymbol\pi)-z$. The odd-sector
principal-part eigenvalue of $A_\mu^{p,o}(\boldsymbol\pi,z)$ admits the
following asymptotic regimes.

\emph{(i) Threshold regime} $\delta=O(1)$ fixed and $\mu\to\infty$:
\[
\lambda^{\boldsymbol\pi,o}(z)
=2\,\frac{(2-z_\mu(\boldsymbol\pi))(18-z)}{(12-z)^2}\,s_1(\delta)+O(\mu^{-1}),
\]
where
\[
s_1(\delta):=\frac{1}{4\pi^2}\int_{\mathbb T^2}
\frac{\sin^2 p_1}{\varepsilon(\mathbf p)+\delta}\,d\mathbf p.
\]
For the limiting principal-part eigenvalue obtained by first fixing
\(\delta>0\) and letting \(\mu\to\infty\), one has
\[
\lambda_{\mathrm{pp}}^{\boldsymbol\pi,o}(\delta) = 2\,s_1(\delta),
\qquad
\lim_{\delta\downarrow0}\lambda_{\mathrm{pp}}^{\boldsymbol\pi,o}(\delta)
= 2-\frac4\pi<1.
\]

\emph{(ii) At the trimer energy} $z=z_1^{\boldsymbol\pi,s}(\mu)
=-3\mu+6+O(\mu^{-1})$:
\[
\lambda^{\boldsymbol\pi,o}(z_1^{\boldsymbol\pi,s})
= \frac{1}{6\mu}+O(\mu^{-2}).
\]
The two limits should not be interchanged: the first is a
threshold limit at fixed $\delta$, the~second is a strong-coupling
limit at the trimer scale.
\end{lemma}

\begin{proof}
The odd kernel reads (see §\ref{sec:newnew54})
\[
K_\pi^{p,o}(\mathbf p,\mathbf q;z,\mu)
=\frac{\mu}{2\pi^2(12-z)^2}
\frac{\sin p_1\sin q_1+\sin p_2\sin q_2}
{\sqrt{\Delta_\mu(\mathbf p,\boldsymbol\pi,z)}\sqrt{\Delta_\mu(\mathbf q,\boldsymbol\pi,z)}}.
\]

By Lemma~\ref{lem:Delta_pi},
$\Delta_\mu^{-1}(\mathbf p,\boldsymbol\pi,z)
=(2-z_\mu(\boldsymbol\pi))(18-z)/(\mu(\varepsilon(\mathbf p)+\delta))(1+O(\mu^{-1}))$
uniformly for $\delta$ in compact subsets of $(0,\infty)$. Substituting
into the odd eigenvalue representation and multiplying by the prefactor
$\mu/(2\pi^2(12-z)^2)$,
\[
\lambda^{\boldsymbol\pi,o}(z)
=\frac{2(2-z_\mu(\boldsymbol\pi))(18-z)}{(12-z)^2}\,s_1(\delta)+O(\mu^{-1}).
\]

This proves (i). For (ii), substitute $z=-3\mu+6$ into the same
formula: $\Delta_\mu(\mathbf p,\boldsymbol\pi,z)=2/3+O(\mu^{-1})$
uniformly in $\mathbf p$ (Proposition~\ref{prop:pp_trimer_pi}), so
$\Delta_\mu^{-1}=3/2+O(\mu^{-1})$ and
\[
\lambda^{\boldsymbol\pi,o}(z_1^{\boldsymbol\pi,s})
=\frac{\mu}{2\pi^2(12-z)^2}\cdot\frac{3}{2}\cdot 2\pi^2 +O(\mu^{-2})
=\frac{3\mu}{2(12-z)^2}+O(\mu^{-2}).
\]
With $12-z=12+3\mu-6+O(\mu^{-1})=3\mu(1+2/\mu+O(\mu^{-2}))$,
$(12-z)^{-2}=(9\mu^2)^{-1}(1+O(\mu^{-1}))$, whence
$\lambda^{\boldsymbol\pi,o}(z_1^{\boldsymbol\pi,s})=1/(6\mu)+O(\mu^{-2})$.

The threshold limit follows from the exact evaluation
\[
s_1(0)=\frac{1}{4\pi^2}\int_{\mathbb T^2}\frac{\sin^2 p_1}{\varepsilon(\mathbf p)}\,d\mathbf p
=1-\frac{2}{\pi},
\]
and from
$(2-z_\mu(\boldsymbol\pi))(18-z)/(12-z)^2\to1$ as $\delta\downarrow0$
(since $2-z_\mu(\boldsymbol\pi)=2+\mu-4\to\infty$ and $18-z\to\infty$ at the same rate $2\mu$).
Therefore $\lambda^{\boldsymbol\pi,o}\to2(1-2/\pi)=2-4/\pi$.
\end{proof}

\begin{remark}[Principal-part noncriticality at the threshold]
\label{rem:odd_threshold_pp}
The limiting odd principal-part eigenvalue at the two-particle
threshold is
\[
\lim_{\delta\downarrow0}
\lambda_{\mathrm{pp}}^{\boldsymbol\pi,o}(\delta)
=
2-\frac4\pi<1.
\]
Thus the odd principal part is noncritical at threshold. This does not
by itself exclude an odd threshold state of the full
Birman--Schwinger operator; such a conclusion requires a
near-threshold estimate for the remainder.
\end{remark}

\begin{remark}[On the integral $\int_{\mathbb T^2}\sin^2 p_1/\varepsilon$]
The evaluation uses the classical two-dimensional lattice Green's
function formula: for $\mathbf x\in \mathbb{Z}^2$,
\[
\frac{1}{4\pi^2}\int_{\mathbb T^2}\frac{e^{i\mathbf p\cdot\mathbf x}-1}{\varepsilon(\mathbf p)}\,d\mathbf p
=-R(\mathbf x),
\]
where $R(\mathbf x)$ is the two-point resistance of the infinite
square lattice of unit resistors. Rewriting
$\sin^2 p_1=\tfrac12(1-\cos 2p_1)$ and setting $\mathbf x=0$ after
subtracting the singular part yields
$s_1(0)=1-2/\pi$. The evaluation is independently confirmed by direct
numerical quadrature to $10$ digits.
\end{remark}

\subsubsection{Strong-coupling expansion of the principal-part eigenvalues}\label{sec:56as}

Let \(\delta:=z_\mu(\boldsymbol\pi)-z\). By Lemma~\ref{lem:Delta_pi},
\[
\int_{\mathbb T^2}\frac{d\mathbf p}{\Delta_\mu(\mathbf p,\boldsymbol\pi,z)}
=\frac{(2-z_\mu(\boldsymbol\pi))(18-z)}{\mu}\,4\pi^2 b_0(\delta)(1+O(\mu^{-1})),
\]
so that
\[
\lambda_1^{\boldsymbol\pi,e}(z)=\frac{2(18-z)^2\big(2-z_\mu(\boldsymbol\pi)\big)}{(12-z)^2}\,b_0(\delta)+O(\mu^{-1}),\qquad \]
\[\lambda^{\boldsymbol\pi,o}(z)=\frac{2(2-z_\mu(\boldsymbol\pi))(18-z)}{(12-z)^2}\,s_1(\delta)+O(\mu^{-1}),
\]
where
\[
b_\alpha(\delta):=\frac1{4\pi^2}\int_{\mathbb T^2}
\frac{\varepsilon^\alpha(\mathbf p)}{\varepsilon(\mathbf p)+\delta}\,d\mathbf p,
\qquad
s_1(\delta):=\frac1{4\pi^2}\int_{\mathbb T^2}
\frac{\sin^2 p_1}{\varepsilon(\mathbf p)+\delta}\,d\mathbf p .
\]
The two families are distinct: the integrals $b_\alpha$ arise in the
even sector (where $\varepsilon(\boldsymbol\pi-\mathbf p)=4-\varepsilon(\mathbf p)$
produces the algebraic identities $b_1=(4+\delta)b_0-1$,
$b_2=(4+\delta)^2b_0-6-\delta$), whereas the odd sector is governed by
$s_1$, which satisfies $s_1(0^+)=1-2/\pi$ and $s_1(\delta)=\frac{1}{2\delta}+O(\delta^{-2})$ as
$\delta\to\infty$.

Setting \(z=-3\mu+d+O(\mu^{-1})\) and using
Proposition~\ref{prop:pp_trimer_pi},
we obtain
\[
\lambda_1^{\boldsymbol\pi,e}(z)=1+\frac{d-6}{2\mu}+O(\mu^{-2}).
\]
Solving $\lambda_1^{\boldsymbol\pi,e}=1$ gives $d=6$, hence the
principal-part calculation, \emph{correctly normalised}, reproduces the
trimer branch
\[
z_1^{\boldsymbol\pi,s}(\mu)=-3\mu+6+O(\mu^{-1}),
\]
in agreement with Lemma~\ref{lem:ground_Kato}. The former
``formal branch'' $-2\mu+6$ was an artefact of the missing Faddeev
factor $2$ in the even-sector kernel (see Remark~\ref{rem:var_check}
and Section~\ref{sec:dictionary}); it is not an eigenvalue of
$H_\mu(\boldsymbol\pi)$.

\begin{remark}[Order-matching criterion]\label{rem:order_criterion}
In the trimer channel the leading eigenvalue function is degenerate,
$\Lambda_0\equiv1$ (Proposition~\ref{prop:pp_trimer_pi}), so by the order-matching criterion of
\cite[Lemma~4.3]{CompanionNote} the constant $6$ is fixed only by the $O(\mu^{-1})$ term of $\Delta_\mu$,
through $\Lambda_1(d)=(d-6)/2=0$. This is the same constant as in the Kato expansion, where it is
$\langle f_0,E_{\mathbf K}f_0\rangle=6$. For the second bound state at $\mathbf K=\mathbf0$ the
function $\Lambda_0(\delta)=g(\delta)=2+\delta-1/b_0(\delta)$ is non-constant, and the leading-order
determinant already determines $C=3.9645800$.
\end{remark}

\subsection[Numerical verification: spectra at $\mathbf{K}=0$ and $\mathbf{K}=\pi$]{Numerical verification: spectra at $\mathbf K=0$ and $\mathbf K=\boldsymbol\pi$}\label{sec:numerics_spectra}

The numerical computations provide independent consistency checks of
the analytical results. Figure~\ref{fig:spectraK0Kpi}(a) compares the
finite-volume trimer energies at $\mathbf K=\mathbf0$ and
$\mathbf K=\boldsymbol\pi$ with the Kato expansion and confirms the
predicted $O(\mu^{-2})$ separation of the two trimer branches.
Figure~\ref{fig:spectraK0Kpi}(b) compares the second
Birman--Schwinger eigenvalue at a distance $\delta$ below the
atom--dimer threshold. At $\mathbf K=\mathbf0$ it follows the
$s$-channel curve and crosses the Birman--Schwinger threshold at the
limiting binding energy $\delta_\infty$, whereas at
$\mathbf K=\boldsymbol\pi$ it follows the $p$-channel curve and
remains below $2-4/\pi<1$ in the threshold limit. These computations
are used as independent consistency checks and are not invoked as
substitutes for the analytical spectral-counting arguments.

\begin{figure}[!ht]
\centering
\includegraphics[width=\textwidth]{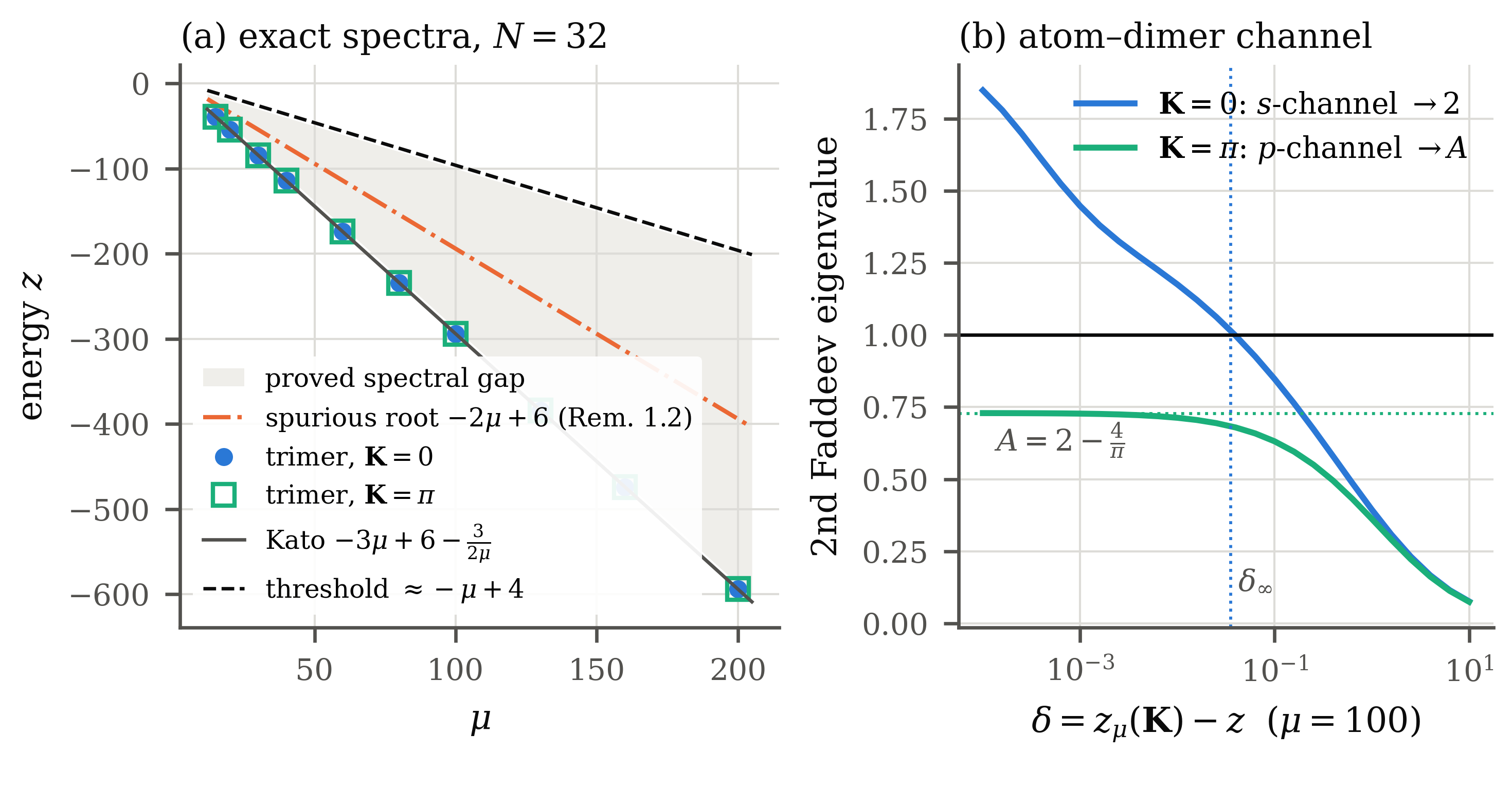}
\caption{Finite-volume spectra numerical checks.
\textbf{(a)} Exact spectra on the $32\times32$ torus, trimer at
$\mathbf K=\mathbf 0$ (dots) and $\mathbf K=\boldsymbol\pi$ (squares)
against the Kato expansion; the two agree to
$z_1^{\boldsymbol\pi,s}(\mu)-z_1^{\mathbf0,s}(\mu)
=
\frac{3}{2\mu^2}+O(\mu^{-4})$
($1.5\cdot10^{-4}$ at $\mu=100$). The shaded strip is the rigorous spectral gap valid for every $\mathbf K$:
since $-\mu V$ has spectrum $\{-3\mu,-\mu,0\}$ and $\|E_{\mathbf K}-6\|\le6$,
$\operatorname{spec}H_\mu(\mathbf K)\subset[-3\mu,-3\mu+12]\cup[-\mu,12]$.
The formal principal-part root $-2\mu+6$ of Remark~\ref{rem:var_check}
lies inside the gap for $\mu>12$, confirming that it is spurious. \textbf{(b)} Second
Faddeev eigenvalue at distance $\delta$ below the atom--dimer threshold on the $48\times48$ torus
($\mu=100$, $N=48$). At $\mathbf K=\mathbf 0$ it is the $s$-channel
$m_s(\delta)$, crosses $1$ at $\delta\approx\delta_\infty$ and tends
to $2$; at $\mathbf K=\boldsymbol\pi$ it is the $p$-channel and
saturates at $A=2-\frac4\pi=0.7268<1$ ($0.7265$ at $\mu=200$), so the limiting atom--dimer operator has no bound state at
$\mathbf K=\boldsymbol\pi$ (Theorem~\ref{thm:pi_preview}, Remark~\ref{rem:pi_no_dimer}).}
\label{fig:spectraK0Kpi}
\end{figure}

\subsection{Direct numerical diagonalisation verification}\label{sec:hp_numerics}
As an independent check of the variational bounds (Lemma~\ref{lem:var_pi}), we performed a~direct diagonalisation of the full three-particle Hamiltonian $H_\mu(\boldsymbol{\pi})$ on finite lattices of size $N=8$ and $N=10$ (with even $N$, so that $K=\pi$ is a grid point and all averaging projections are well-defined). The lowest eigenvalues $z_1$ for $\mu=20,40,70,120$ are shown in Table~\ref{tab:hp_ground}. In all cases $z_1+3\mu$ agrees with the Kato expansion $6-\frac{3}{2\mu}+\frac{3}{4\mu^2}$ to $10^{-5}$ (e.g., at $\mu=120$: $z_1+3\mu=5.987552$), and $N=8$ and $N=10$ agree to all printed digits, while $z_1+2\mu\approx -14,\dots,-114$, i.e. the value $-2\mu+6$ lies in the spectral gap. The~distance from the numerically computed ground state to the
two-particle threshold $z_\mu(\boldsymbol\pi)=-\mu+4$ is consistent
with $z_\mu(\boldsymbol\pi)-z_1^{\boldsymbol\pi,s}(\mu)=2\mu-2+O(\mu^{-1})$
(e.g., $78.037$ at $\mu=40$, $238.012$ at $\mu=120$); the finite-volume
quantity listed in the last column should be read as the distance to
the threshold, not as the position of a second bound state.

\begin{table}[!ht]
\centering
\caption{Lowest eigenvalues $z_1$ of $H_\mu(\boldsymbol{\pi})$ on finite lattices ($N\times N$ torus), compared with the Kato expansion $6-\frac{3}{2\mu}+\frac{3}{4\mu^2}-\frac{15}{32\mu^3}$.}
\label{tab:hp_ground}
\begin{tabular}{c|c|c|c|c|c}
$\mu$ & $N$ & $z_1$ & $z_1+3\mu$ & Kato & $(-\mu+4)-z_1$\\
\hline
20 & 8 & $-54.073181$ & $5.926819$ & $5.926816$ & $38.0732$\\
40 & 8 & $-114.037038$ & $5.962962$ & $5.962961$ & $78.0370$\\
70 & 8 & $-204.021277$ & $5.978723$ & $5.978723$ & $138.0213$\\
120 & 8 & $-354.012448$ & $5.987552$ & $5.987552$ & $238.0124$\\
120 & 10 & $-354.012448$ & $5.987552$ & $5.987552$ & $238.0124$
\end{tabular}
\end{table}

The numerical value $z+3\mu=6.073$ at $\mu=100$ quoted in
Fig.~\ref{fig:factor2} (Section~\ref{sec:dictionary}) is the crossing of the \emph{principal part}
alone; the exact finite-volume Faddeev operator gives $5.9850745$,
matching Kato's $\displaystyle 6-\frac{3}{2\mu}+\frac{3}{4\mu^2}=5.985075$. The
difference is the contribution of the remainder $A^r_\mu$.

\begin{corollary}
[Spectral gap at $\mathbf K=\boldsymbol\pi$]\label{cor:gap}
$z_\mu(\boldsymbol\pi)-z_1^{\boldsymbol\pi,s}(\mu)=2\mu-2+\frac{3}{2\mu}+O(\mu^{-2})$.
\end{corollary}
\begin{proof}
Subtract Lemma~\ref{lem:ground_Kato} (with $a_2=-\frac32$) from the exact identity
$z_\mu(\boldsymbol\pi)=-\mu+4$. Table~\ref{tab:hp_ground} is an independent check.
\end{proof}

\subsection{The principal-part method, the Kato expansion and the
effective atom--dimer theory: a dictionary}\label{sec:dictionary}

Three different methods are used for the strong-coupling spectrum of
$H_\mu(\mathbf K)$: (M1) the principal part $A^p_\mu$ of the Faddeev
(Birman--Schwinger) operator, restricted to invariant subspaces (the
method of \cite{AbdullaevKhalkhuzhaevToshturdiev2026,
KhalkhuzhaevAbdullaevBoymurodov2022,AbdullaevKhalkhuzhaevKhujamiyorov2023}
and of Sections~\ref{sec:proof_main} and~\ref{sec:five}); (M2) Kato
perturbation theory at the isolated eigenvalue $3$ of
$V=V_1+V_2+V_3$ (Lemma~\ref{lem:ground_Kato}); (M3) the
Feshbach--Schur reduction onto the infinitely degenerate eigenvalue $1$
of $V$, which yields the effective atom--dimer operator
$h_B(\mathbf K)=4+\varepsilon_D-S_{\mathbf K}$ on
$\ell^2(\mathbb Z^2\setminus\{0\})$ and the $4\times4$ threshold matrix
$M(\delta)$ \cite{AtomDimer2026}. The purpose of this subsection is to
show that, \emph{when correctly normalised}, (M1) is not in conflict
with (M2)--(M3): every eigenvalue of the principal part computed in
this paper is one of the channel eigenvalues of $M(\delta)$, and the
only apparent discrepancy (the crossing at $-2\mu+6$, formal principal-part crossing from physical ground state \cite[Proposition~6.2 and Remark~6.3]{CompanionNote}) was a
normalisation slip. What (M1) cannot do in two dimensions is explained
in Remark~\ref{rem:limits_M1}.

\paragraph{Channel eigenvalues.} Put
$b_0(\delta)=\langle(\varepsilon+\delta)^{-1}\rangle
=\displaystyle\frac{2}{\pi(2+\delta)}\mathrm K\bigl(\tfrac{2}{2+\delta}\bigr)$
(complete elliptic integral of the first kind) and
\begin{gather}\label{eq:channels}
m_s(\delta)=2+\delta-\frac{1}{b_0(\delta)},\\
m_p(\delta)=\Bigl\langle\frac{2\sin^2p_1}{\varepsilon(\mathbf p)+\delta}\Bigr\rangle,\qquad
m_d(\delta)=\Bigl\langle\frac{(\cos p_1-\cos p_2)^2}{\varepsilon(\mathbf p)+\delta}\Bigr\rangle .\notag
\end{gather}
These are the eigenvalues of $M(\delta)$ on the $s$, $p$ ($\times2$)
and $d$ vectors of the first shell~\cite{AtomDimer2026}; as
$\delta\downarrow0$, $m_s\to2$ (logarithmically slowly),
$m_p\to\displaystyle A=2-\frac4\pi=R(2\mathbf e_1)$ and
$m_d\to\frac8\pi-2=2R(\mathbf e_1+\mathbf e_2)-R(2\mathbf e_1)$.

\begin{proposition}[Dictionary]\label{prop:dictionary}
Let $\mathbf K\in\{\mathbf 0,\boldsymbol\pi\}$, fix $\delta>0$ and put
$z=z_\mu(\mathbf K)-\delta$. As $\mu\to\infty$ the eigenvalues of the
principal part $A^p_\mu(\mathbf K,z)$ other than the trimer eigenvalue
converge to the following limits:
\begin{center}\small
\begin{tabular}{llll}
\toprule
sector & $\mathbf K=\mathbf 0$ & $\mathbf K=\boldsymbol\pi$ & channel\\
\midrule
even sym., 2nd eig. & $\lambda_2^s\to g(\delta)=+m_s(\delta)$ & $\to-m_s(\delta)$ & $s$\\
even, antisymmetric & $\lambda^{as}\to+m_d(\delta)$ & $\to-m_d(\delta)$ & $d$\\
odd (multiplicity $2$) & $\to-m_p(\delta)$ & $\lambda^{\boldsymbol\pi,o}\to+m_p(\delta)$ & $p$\\
\bottomrule
\end{tabular}
\end{center}
The sign is the eigenvalue of $-S_{\mathbf K}$ on the channel vector
($S_{\boldsymbol\pi}=-S_{\mathbf 0}$). Consequently the principal-part
count of eigenvalues exceeding $1$ in the atom--dimer region coincides
with the count of the counting criterion of \cite{AtomDimer2026}: one
state at $\mathbf K=\mathbf 0$ ($m_s(\delta_\infty)=1$), none at
$\mathbf K=\boldsymbol\pi$ (the only positive channel is $m_p<A<1$).
\end{proposition}

\begin{proof}[Proof sketch]
With $\Delta_\mu^{-1}=\frac{(2-z_\mu(\mathbf K))(c_{\mathbf K}-z)}
{\mu(\varepsilon+\delta)}(1+O(\mu^{-1}))$, $c_{\mathbf 0}=6$,
$c_{\boldsymbol\pi}=18$ (Lemmas~\ref{lem:Delta}, \ref{lem:Delta_pi}),
all prefactors tend to $1$ and each rank-one or rank-two block reduces
to a Gram computation in $L^2(\mathbb T^2,(\varepsilon+\delta)^{-1}d\mathbf p)$.
For the symmetric blocks, with $u=\varepsilon(\mathbf p)$ and
$v\in\{4-u,\,u\}$ at $\mathbf K=\mathbf 0,\boldsymbol\pi$ respectively,
the second eigenvalue is
$\pm\bigl(\langle v^2/(u+\delta)\rangle-\langle v/(u+\delta)\rangle^2/b_0\bigr)$,
and in both cases the $b_0$-proportional terms cancel exactly
(cf.\ the proof of Lemma~\ref{lem:C_constant}), leaving
$2+\delta-1/b_0=m_s$. The odd and antisymmetric blocks are rank one
and give $\langle 2\sin^2p_1/(\varepsilon+\delta)\rangle=m_p$ and
$\langle(\cos p_1-\cos p_2)^2/(\varepsilon+\delta)\rangle=m_d$.
\end{proof}

\paragraph{The trimer channel and the factor $2$.} The Faddeev
equation for three bosons reads
$\Delta_\mu(\mathbf p,\mathbf K,z)\varphi(\mathbf p)
=2\mu\langle(E_{\mathbf K}(\mathbf p,\cdot)-z)^{-1}\varphi\rangle$;
the factor $2$ comes from the two exchange terms
$\varphi(\mathbf q)$ and $\varphi(\mathbf K-\mathbf p-\mathbf q)$,
which coincide after the substitution
$\mathbf q\mapsto\mathbf K-\mathbf p-\mathbf q$. It is present in
\eqref{eq:bs_operator} ($\mu/2\pi^2=2\mu/4\pi^2$) and in the
$\mathbf K=\mathbf 0$ kernels, but was lost in the even-sector kernel
at $\mathbf K=\boldsymbol\pi$ in the first version of this paper (Remark~\ref{rem:var_check}). With the
correct normalisation, the rank-one principal part satisfies
\begin{equation}\label{eq:pp_trimer}
\lambda^{\boldsymbol\pi,e}_1(-3\mu+d)
=2\mu\frac{18-z}{(12-z)^2}\bigl\langle\Delta_\mu^{-1}\bigr\rangle
=1+\frac{d-6}{2\mu}+O(\mu^{-2}),
\end{equation}
using $\Delta_\mu(\mathbf p,\boldsymbol\pi,-3\mu+d)
=\frac23+\frac{4-d+\varepsilon(\mathbf p)}{9\mu}+O(\mu^{-2})$.
Hence $\Lambda_0\equiv1$ and $\Lambda_1(d)=\frac{d-6}{2}$: this is the
degenerate case of the order-matching criterion \cite[Lemma~4.3]{CompanionNote},
realised on the \emph{true} ground-state branch, and it returns
$z_1^{\boldsymbol\pi,s}=-3\mu+6+O(\mu^{-1})$, in agreement with
Lemma~\ref{lem:ground_Kato}. On the line $z=-2\mu+d$ the same function
equals $2+2(d-6)/\mu+O(\mu^{-2})$ and has no root; the former
``formal branch'' $-2\mu+6$ is the root of \emph{one half} of it
(Fig.~\ref{fig:factor2}).

\begin{figure}[!ht]
\centering
\includegraphics[width=0.84\textwidth]{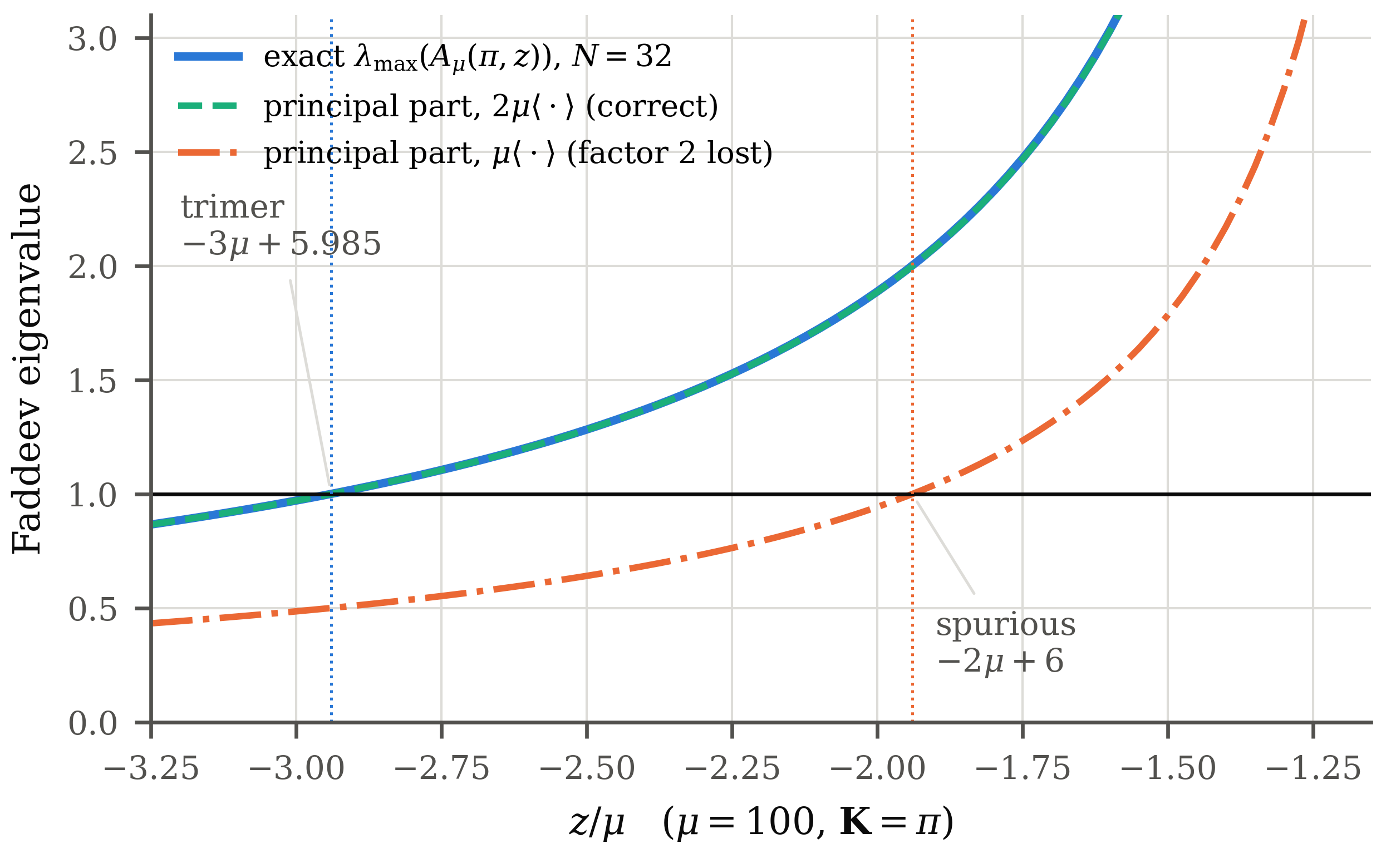}
\caption{$\mathbf K=\boldsymbol\pi$, $\mu=100$. Blue: largest
eigenvalue of the exact finite-volume Birman--Schwinger operator
$A_\mu(\boldsymbol\pi,z)$ on the $32\times32$ torus. The finite-volume
Birman--Schwinger equation is exactly equivalent to the corresponding
finite-volume three-body eigenvalue problem. Green, dashed: the~rank-one principal part \eqref{eq:pp_trimer} with the
correct Faddeev factor $2\mu$; it is numerically indistinguishable from
the exact curve and crosses $\lambda=1$ at $z=-3\mu+6.07$. The
principal-part crossing is a leading-order diagnostic only and does
\emph{not} determine the $O(\mu^{-1})$ correction: the exact
finite-volume value is $-3\mu+5.98507$, while Kato's second-order
expansion gives $-3\mu+6-\frac{3}{2\mu}+\frac{3}{4\mu^2}
=-3\mu+5.985075$. Orange: the same principal part with the Faddeev
factor $2$ omitted; its root near $-2\mu+6$ is the spurious formal
branch of Remark~\ref{rem:var_check}. At $z=-2\mu+6$ the exact operator
has $\lambda_{\max}=2.0005>1$, confirming that the spurious branch lies
\emph{above} the true trimer eigenvalue and therefore cannot correspond
to a bound state.}
\label{fig:factor2}
\end{figure}

\begin{remark}[What the principal-part method can and cannot do in
two dimensions]\label{rem:limits_M1}
(i) \emph{Evaluation at the threshold.} In $d=3$ the Faddeev operator
at $z=z_\mu(\mathbf 0)$ is Hilbert--Schmidt, because its kernel is
$\asymp\mu^{-1}[\varepsilon(\mathbf p)\varepsilon(\mathbf q)]^{-1/2}$
and $W=\langle\varepsilon^{-1}\rangle_{\mathbb T^3}=0.505462\ldots<\infty$;
the count can be read off at the threshold itself, as in
\cite{AbdullaevKhalkhuzhaevToshturdiev2026}. In $d=2$,
$\langle\varepsilon^{-1}\rangle_{\mathbb T^2}=\infty$, the trimer
eigenvalue grows like $\mu\,b_0(\delta)$ and the remainder is only
$O(\mu^{-1}\ln\frac1\delta)$ (Lemma~\ref{lem:threshold_count}(iii)):
counting must be done at fixed $\delta>0$ and the limit
$\delta\downarrow0$ controlled separately
(Fig.~\ref{fig:2Dvs3D}(a)).

(ii) \emph{Range of the determinant asymptotics.} The formula
$\Delta_\mu\simeq\mu(\varepsilon+\delta)/[(2-z_\mu)(c_{\mathbf K}-z)]$
holds for $\delta=o(\mu)$; at the trimer scale $z=-3\mu+O(1)$ one has
instead $\Delta_\mu\to\frac23$. Substituting the threshold formula at
$z\sim-3\mu$ is not justified; there one must expand $\Delta_\mu$
directly, as in \eqref{eq:pp_trimer}, or use Kato theory.

(iii) \emph{Additive constants.} In the trimer channel $\Lambda_0\equiv1$,
so the constant $6$ requires the $O(\mu^{-1})$ term of $\Delta_\mu$;
leading-order determinant asymptotics cannot fix $O(1)$ constants in a
degenerate channel.

(iv) \emph{General $\mathbf K$.} For $\mathbf K\notin\{0,\pi\}^2$ there
is no parity reduction (Proposition~\ref{prop:trim}); the~exchange
operator $S_{\mathbf K}$ of (M3) treats all $\mathbf K$ at once and
produces the zones of Fig.~\ref{fig:zones2D3D}(a).
\end{remark}

\begin{figure}[!ht]
\centering
\includegraphics[width=\textwidth]{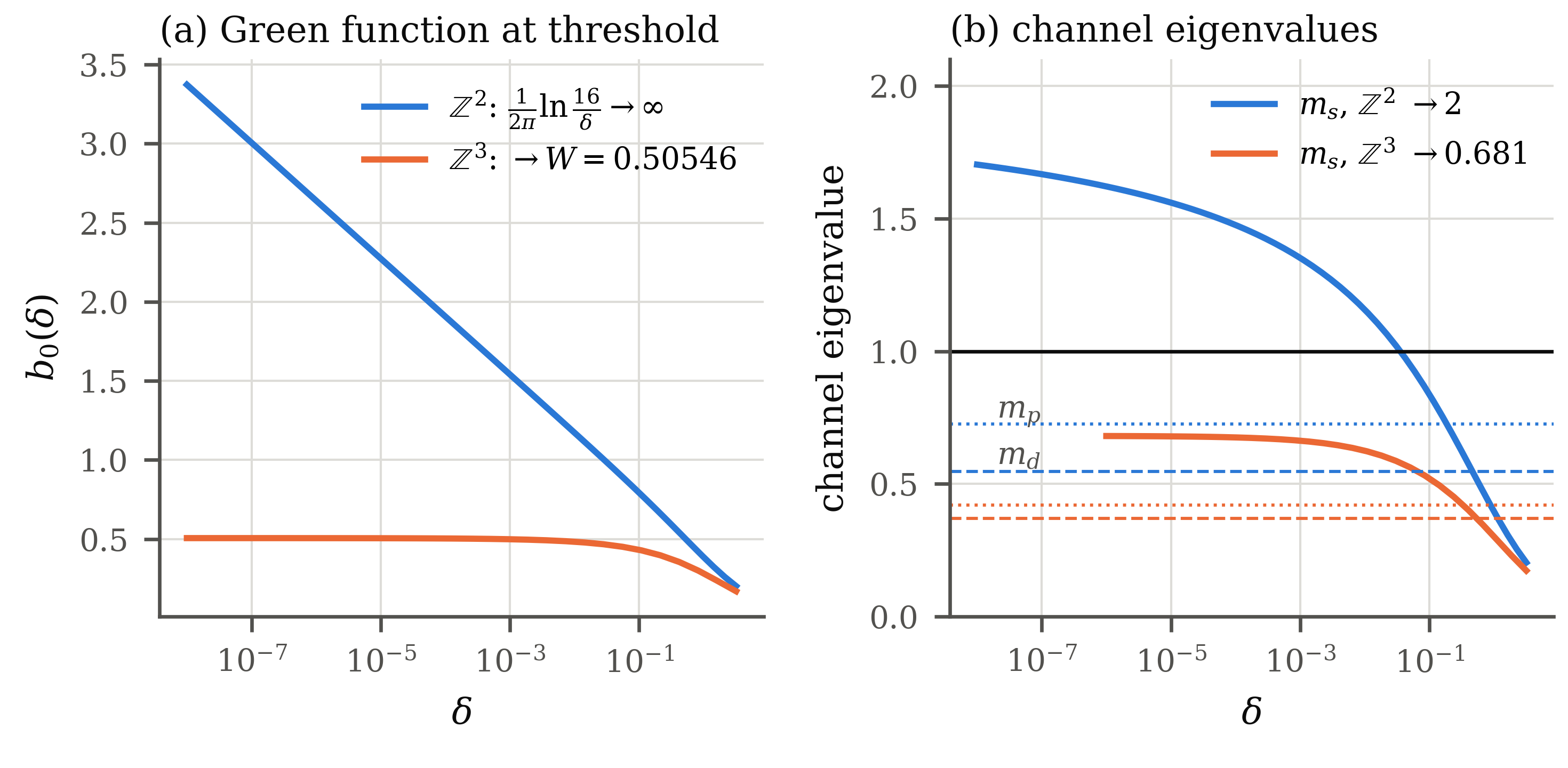}
\caption{Why two dimensions differ. \textbf{(a)}
$b_0(\delta)=\langle(\varepsilon+\delta)^{-1}\rangle$ diverges as
$\displaystyle\frac{1}{2\pi}\ln\frac{16}{\delta}$ on $\mathbb Z^2$, whereas it tends
to the finite Watson value $W=0.505462\ldots$ on $\mathbb Z^3$.
\textbf{(b)} The limiting $s$-channel eigenvalue
$m_s(\delta)$ tends to $2$ in two dimensions and to
$2-\displaystyle\frac{2}{3W}=0.6810747\ldots<1$ in three dimensions. The~dotted
and dashed horizontal lines show the corresponding limiting
$p$- and $d$-channel values:
$0.7268,\,0.5465$ in $\mathbb Z^2$ and
$0.4197,\,0.3705$ in $\mathbb Z^3$. The two-dimensional divergence $b_0(\delta)\sim\frac{1}{2\pi}\ln(16/\delta)$
is the source of the transcendental constant $C$; in three dimensions the
finite Watson value $W=0.505462\ldots$ produces instead the closed-form
critical ratios of Remark~\ref{rem:3D_constants}.}
\label{fig:2Dvs3D}
\end{figure}

\begin{remark}[Three-dimensional threshold channels and comparison]
\label{rem:3D_constants}

For comparison with the three-dimensional literature, consider the
six-site first shell
\[
\{\pm\mathbf e_1,\pm\mathbf e_2,\pm\mathbf e_3\}
\subset\mathbb Z^3.
\]
Its cubic-symmetry decomposition is
\[
A_{1g}\oplus E_g\oplus T_{1u},
\qquad
\dim A_{1g}=1,\quad
\dim E_g=2,\quad
\dim T_{1u}=3.
\]
The corresponding limiting threshold-channel eigenvalues are
\[
m_s^{(3)}
=
2-\frac{2}{3W}
=
0.6810747\ldots,
\]
\[
m_d^{(3)}
=
0.3704749\ldots,
\qquad
m_p^{(3)}
=
0.4196834\ldots,
\]
where
\[
W=0.5054620\ldots
\]
is the Watson integral in the present normalization \cite{GlasserZucker1977,Zucker2011}, and
\[
R_3(\mathbf x)=W-G_3(\mathbf x)
\]
denotes the corresponding simple-cubic two-point resistance.

The value \(m_s^{(3)}\) is exactly the limiting eigenvalue
\(\lambda_2\) in~\cite[Eq.~(17)]{AbdullaevKhalkhuzhaevToshturdiev2026},
while the equality
\[
\lambda_3=\lambda_4=m_d^{(3)}
\]
in~\cite[Eq.~(15)]{AbdullaevKhalkhuzhaevToshturdiev2026}
corresponds to the two-dimensional \(E_g\) sector. Thus, for
sufficiently large \(\mu\),
\[
\lambda_1(\mu)>1,\qquad
\lambda_2(\mu)\to m_s^{(3)}<1,\qquad
\lambda_3(\mu),\lambda_4(\mu)\to m_d^{(3)}<1,
\]
together with the negative odd-sector contribution. Hence
\(m_s^{(3)}<1\), equivalently \(W<2/3\), is the only nontrivial
threshold-channel obstruction in the symmetric block; it is not,
by itself, the complete Birman--Schwinger counting argument.
By contrast, in two dimensions
\[
m_s^{(2)}(0^+)=2>1,
\]
because \(b_0(\delta)\) diverges logarithmically. This is the precise
dimension-dependent mechanism preventing a direct transfer of the
three-dimensional threshold count to \(\mathbb Z^2\).

Two independent issues with Eq.~(16) of
\cite{AbdullaevKhalkhuzhaevToshturdiev2026} deserve to be recorded.
\emph{First}, as printed, the two displayed forms
$\lambda_1(\mu)=2z_\mu(\mathbf{0})+35+\tfrac13+O(\mu^{-1})$ and
$\lambda_1(\mu)=2\mu+23+\tfrac13+O(\mu^{-1})$ are mutually inconsistent: with
$z_\mu(\mathbf{0})=-\mu+6+O(\mu^{-1})$, the first evaluates to
$-2\mu+\tfrac{142}{3}+O(\mu^{-1})$, not $2\mu+\tfrac{70}{3}$. The two forms become
consistent if the first term is interpreted as $-2z_\mu(\mathbf{0})$, which yields
$-2z_\mu(\mathbf{0})+35+\tfrac13 = 2\mu+\tfrac{70}{3}+O(\mu^{-1})$.
\emph{Second}, combining Eqs.~(10) and (14) of
\cite{AbdullaevKhalkhuzhaevToshturdiev2026} directly — i.e.\ substituting
$\Delta_\mu^{-1}(\mathbf{x},z_\mu(\mathbf{0}))\approx\mu/\varepsilon(\mathbf{x})$
into the principal-part kernel and evaluating the leading rank-one eigenvalue —
gives $\lambda_1(\mu)=2W\mu+O(1)\approx1.011\,\mu$, where
$W=0.505462\ldots$ is the Watson integral; the factor $W$ is absent from the printed
$2\mu+O(1)$. Since the same $W$ correctly enters the second channel of
\cite{AbdullaevKhalkhuzhaevToshturdiev2026} through
$\lambda_2\to2-2/(3W)$, the discrepancy in $\lambda_1$ is not a convention issue.
Moreover, Eq.~(14) retains only the leading determinant asymptotics and therefore cannot
determine the $O(1)$ additive constant in $\lambda_1(\mu)$. These observations do not
affect the large-$\mu$ uniqueness criterion of
\cite{AbdullaevKhalkhuzhaevToshturdiev2026}, but they are relevant when comparing
higher-order strong-coupling expansions across dimensions.

For the \(2+1\) fermionic channels studied in the cited
three-dimensional literature, the corresponding critical ratios are
\[
\gamma_1
=
\frac{2}{m_s^{(3)}}
=
\frac{3W}{3W-1}
=
2.9365356\ldots,
\]
\[
\gamma_2
=
\frac{2}{m_d^{(3)}}
=
5.398475\ldots,
\qquad
\gamma_0
=
\frac{2}{m_p^{(3)}}
=
4.765497\ldots.
\]
The channel multiplicities are \(1\), \(2\), and \(3\), corresponding
to \(A_{1g}\), \(E_g\), and \(T_{1u}\), respectively; in the attractive
fermionic counting the first two channels may therefore produce a
cumulative multiplicity \(1+2=3\). The value
\[
\gamma_1=\frac{3W}{3W-1}=2.9365356\ldots
\]
is quoted here from the closed-form channel expression; the values
$\gamma_2=5.398475\ldots$ and $\gamma_0=4.765497\ldots$ are given to
the displayed numerical precision.
\end{remark}

\begin{figure}[!ht]
\centering
\includegraphics[width=\textwidth]{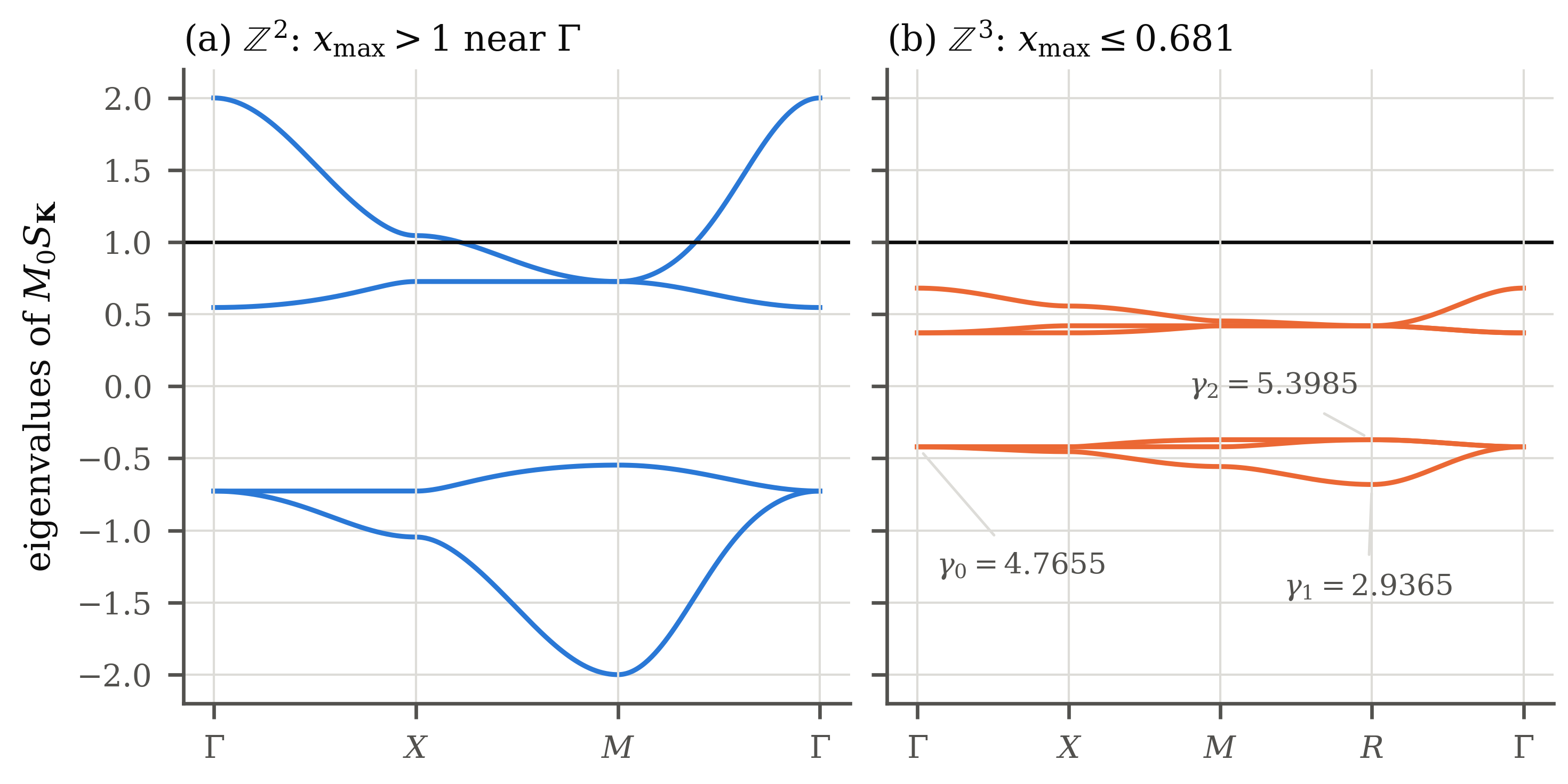}
\caption{Eigenvalues of the limiting threshold matrix $M_0S_{\mathbf K}$ along
high-symmetry paths. \textbf{(a)}~On~$\mathbb Z^2$, the bosonic
atom--dimer bound-state condition is that the largest eigenvalue
exceeds $1$; it~reaches $2$ at $\Gamma$, is approximately $1.045$
at $X$, and is below $1$ near $M$. \textbf{(b)} On $\mathbb Z^3$, the
largest symmetric threshold eigenvalue is
$m_s^{(3)}=0.6810747\ldots<1$. For the $2+1$ fermionic effective
channels, the threshold condition is
$\gamma>2/|x|$ for each negative eigenvalue $x$; the labelled values
are the critical ratios discussed in
Remark~\ref{rem:3D_constants}.}
\label{fig:zones2D3D}
\end{figure}

The two methods are compared side by side in
Table~\ref{tab:dictionary}: the principal-part treatment, the
Kato expansion, and the effective atom--dimer counting agree once
the Faddeev factor $2$ is kept.

\begin{table}[!ht]
\centering\small
\caption{Principal-part method (M1) versus Kato and Feshbach--Schur methods (M2--M3)
(bosons on $\mathbb Z^2$, $\mu\to\infty$).}
\label{tab:dictionary}
\setlength{\tabcolsep}{2pt}
\begin{tabular}{p{2.6cm}p{3.2cm}p{3.2cm}p{2.3cm}}
\toprule
Quantity & Principal part (M1) & Kato / atom--dimer (M2--M3) & Comparison\\
\midrule
Trimer, $\mathbf K=\mathbf0$ & $\lambda_1^s=1$ with threshold $\Delta_\mu$: not justified at $z\sim-3\mu$ & $-3\mu+6-\frac3{2\mu}+O(\mu^{-2})$ & same branch; (M1) needs the $O(\mu^{-1})$ term of $\Delta_\mu$\\
Trimer, $\mathbf K=\boldsymbol\pi$ & $-3\mu+6+O(\mu^{-1})$ (Prop.~\ref{prop:pp_trimer_pi}) & $-3\mu+6-\frac3{2\mu}+O(\mu^{-2})$ & agree once the factor $2$ is kept\\
Atom--dimer, $\mathbf K=\mathbf0$ & $g(\delta)=1$ & $m_s(\delta)=1$, same equation & agree; $C=3.96458$\\
Atom--dimer, $\mathbf K=\boldsymbol\pi$ & $-m_s<0$: no crossing & none: $F_B(-1,-1)=0.346>0$ & agree\\
Odd sector & $\mp m_p$, $|m_p|\le A<1$ & $p$-channel & agrees\\
Even antisymmetric & $\pm m_d$, $m_d\le\frac8\pi-2$ & $d$-channel & agrees\\
General $\mathbf K$ & not available & zones, $F_B(u,v)<0$ & (M3) only\\
Error control & $O(\mu^{-1}\ln\frac1\delta)$, fixed $\delta$ & $36/(\mu-10)$, uniform in $\mathbf K$ & (M3) uniform in $\mathbf K$\\
\bottomrule
\end{tabular}
\end{table}

\subsection[Relation between $\mathbf{K}=0$ and $\mathbf{K}=\pi$]%
{Relation between $\mathbf{K}=0$ and $\mathbf{K}=\boldsymbol{\pi}$}

The results for $\mathbf{K}=\boldsymbol{\pi}$ are related to those for $\mathbf{K}=0$ by the unitary transformation
\[
U_{\boldsymbol{\pi}}: f(\mathbf{p},\mathbf{q}) \longmapsto f(\boldsymbol{\pi}-\mathbf{p}, \boldsymbol{\pi}-\mathbf{q}),
\]
which satisfies
\[
U_{\boldsymbol{\pi}} H_\mu(0) U_{\boldsymbol{\pi}}^{-1} = 12I - H_\mu^{\mathrm{rep}}(\boldsymbol{\pi}),
\]
where $H_\mu^{\mathrm{rep}}(\boldsymbol{\pi}) = E_{\boldsymbol{\pi}} + \mu (V_1+V_2+V_3)$ is the operator with repulsive interaction (the sign of the interaction term is reversed compared to $H_\mu(\boldsymbol{\pi})$). This relation maps the attractive problem at $\mathbf K=0$ onto the repulsive problem at
$\mathbf K=\boldsymbol\pi$; it does not relate the two attractive operators
$H_\mu(\mathbf 0)$ and $H_\mu(\boldsymbol\pi)$, whose bound-state counts are compared below.

\begin{remark}[On the unitary equivalence and its limitations]\label{rem:unitary_limits}
The operator $12I-H_\mu^{\mathrm{rep}}(\boldsymbol\pi)$ has repulsive interaction, whereas
$H_\mu(\boldsymbol\pi)$ has attractive interaction; $U_{\boldsymbol\pi}$ therefore relates eigenvalues
of $H_\mu(\mathbf 0)$ below its essential spectrum to eigenvalues of
$H^{\mathrm{rep}}_\mu(\boldsymbol\pi)$ above its essential spectrum. This is the relation used in the
three-dimensional repulsive problem of \cite{AbdullaevKhalkhuzhaevKhujamiyorov2023}.
\end{remark}

\paragraph{Mechanism of the reduction from two bound states to one.} In both cases the parity
symmetry is exact and the principal part splits into the same blocks. In the atom--dimer
region the non-trimer eigenvalues are, by Proposition~\ref{prop:dictionary},
\[
\mathbf K=0:\ (+m_s,\ +m_d,\ -m_p,-m_p),\qquad
\mathbf K=\boldsymbol\pi:\ (-m_s,\ -m_d,\ +m_p,+m_p),
\]
because $S_{\boldsymbol\pi}=-S_{\mathbf 0}$. At $\mathbf K=0$ the $s$-channel eigenvalue
$m_s(\delta)$ crosses $1$ at $\delta=\delta_\infty$ (second bound state); at
$\mathbf K=\boldsymbol\pi$ it becomes negative, and the only positive channel is the odd
$p$-channel with $m_p(\delta)<2-4/\pi<1$. The positivity of the odd principal part at
$\mathbf K=\boldsymbol\pi$ is therefore \emph{not} the cause of the reduction; it is harmless
because $2-4/\pi<1$.

\begin{table}[!ht]
\centering
\caption{Comparison of spectral properties for \(\mathbf K=0\) and \(\mathbf K=\boldsymbol{\pi}\) (\(\mu\to\infty\)).}
\label{tab:K0_vs_Kpi}
\footnotesize
\setlength{\tabcolsep}{3pt}
\begin{tabular}{p{2.7cm}|>{\raggedright\arraybackslash}p{4cm}|>{\raggedright\arraybackslash}p{4.5cm}}
\hline
\textbf{Feature} & \(\mathbf K=0\) & \(\mathbf K=\boldsymbol{\pi}\) \\
\hline
Symmetry & parity \((\mathbf p,\mathbf q)\mapsto(-\mathbf p,-\mathbf q)\) & same \\
\hline
Two-particle threshold & \(z_\mu(\mathbf 0)=-\mu+4-\frac4\mu+O(\mu^{-3})\) & \(z_\mu(\boldsymbol\pi)=-\mu+4\) (exact) \\
\hline
Eigenvalues below threshold & two & one below \(z_\mu(\boldsymbol\pi)-\eta_\mu\), \(\eta_\mu=\frac{36}{\mu-10}\) \\
\hline
Trimer &
$-3\mu+6-\frac{3}{2\mu}
-\frac{3}{4\mu^2}
+O(\mu^{-3})$ &
$-3\mu+6-\frac{3}{2\mu}
+\frac{3}{4\mu^2}
+O(\mu^{-3})$, same through $O(\mu^{-1})$; separation
$\frac{3}{2\mu^2}+O(\mu^{-4})$ \\
\hline
Gap & \(z_2^s-z_1^s=2\mu-2-\delta_\infty+O(\mu^{-1})\) & \(z_\mu(\boldsymbol\pi)-z_1^{\pi,s}=2\mu-2+O(\mu^{-1})\) \\[1pt]
\hline
\(s\)-channel (even sym.) & \(+m_s\to2\): crosses \(1\) at \(\delta_\infty\) & \(-m_s<0\) \\
\hline
\(d\)-channel (even antisym.) & \(+m_d\le\frac8\pi-2<1\) & \(-m_d<0\) \\
\hline
\(p\)-channel (odd, \(\times2\)) & \(-m_p<0\) & \(+m_p\le2-\frac4\pi<1\) \\
\hline
\end{tabular}
\end{table}

The systematic comparison is collected in Table~\ref{tab:K0_vs_Kpi}.

\subsection*{Interpretation of the \(\mathbf K=\bf\pi\) analysis}

\textit{First}, the spectral phenomena at $\mathbf{K}=0$ are not universal across the Brillouin
zone: the~number of bound states is two at $\mathbf K=0$ and one at $\mathbf K=\boldsymbol\pi$,
and the region of the zone carrying the atom--dimer state is described by the zones of
Fig.~\ref{fig:zones2D3D}(a).
\textit{Second}, the trimer is insensitive to the quasimomentum up to $O(\mu^{-2})$; the binding
is not weakened at the corner of the Brillouin zone.
\textit{Third}, 
in the present three-boson model the corresponding exactness of the
two-particle threshold, $z_\mu(\boldsymbol\pi)=-\mu+4$, enters the gap formula directly.

\section{Conclusion}

We have analysed the strong-coupling spectrum of three identical bosons on the square lattice
at the two extremal total quasimomenta \(\mathbf K=\mathbf0\) and \(\mathbf K=\boldsymbol\pi\).

At \(\mathbf K=\mathbf0\), for all sufficiently large \(\mu\), the Hamiltonian has exactly two
eigenvalues below the essential spectrum,
\begin{gather*}
z_1^s(\mu)=-3\mu+6-\frac{3}{2\mu}+O(\mu^{-2}),\\
z_2^s(\mu)=-\mu+C+O(\mu^{-1}),\quad C=3.9645800108\ldots,
\end{gather*}
so that \(\Delta(\mu)=z_2^s(\mu)-z_1^s(\mu)=2\mu-6+C+O(\mu^{-1})\). The trimer branch follows from
Kato perturbation theory at the simple isolated eigenvalue \(3\) of the pair-contact operator,
uniformly in \(\mathbf K\); the atom--dimer branch follows from the Feshbach--Schur reduction of
\cite{AtomDimer2026}, whose limiting equation \(m_s(\delta)=1\) coincides with the
principal-part equation \(g(\delta)=1\).

At $\mathbf K=\boldsymbol\pi$ the trimer branch coincides with that at $\mathbf K=0$ up to
\(\frac{3}{2\mu^2}+O(\mu^{-3})\), and the two-particle threshold is
\(z_\mu(\boldsymbol\pi)=-\mu+4\) exactly. Moreover, \(H_\mu(\boldsymbol\pi)\) has exactly one
eigenvalue below \(z_\mu(\boldsymbol\pi)-\eta_\mu\),
\(\eta_\mu=36/(\mu-10)\); an eigenvalue in the residual window
\([z_\mu(\boldsymbol\pi)-\eta_\mu,z_\mu(\boldsymbol\pi))\) of width
\(O(\mu^{-1})\) is not excluded by the present argument
(Remark~\ref{rem:pi_no_dimer}). The second bound state disappears
because \(S_{\boldsymbol\pi}=-S_{\mathbf 0}\) turns the \(s\)-channel eigenvalue negative; the odd
\(p\)-channel becomes positive but is bounded by \(2-4/\pi<1\). Correctly normalised, the
principal-part method reproduces all of these statements; its eigenvalues are the channel
eigenvalues of the atom--dimer threshold matrix (Proposition~\ref{prop:dictionary}). In three
dimensions the same matrix gives \(m_s^{(3)}=2-2/(3W)<1\), which explains the uniqueness of the
bound state at \(\mathbf K=0\) found in \cite{AbdullaevKhalkhuzhaevToshturdiev2026} and the
critical mass ratios of \cite{KhalkhuzhaevAbdullaevBoymurodov2022,AbdullaevKhalkhuzhaevKhujamiyorov2023}.

The exact statement proved here is
\[
N\!\left(H_\mu(\boldsymbol\pi);\,(-\infty,z_\mu(\boldsymbol\pi)-\eta_\mu)\right)=1;
\]
the present argument does not exclude an additional eigenvalue in the residual window
$[z_\mu(\boldsymbol\pi)-\eta_\mu,z_\mu(\boldsymbol\pi))$ of width $O(\mu^{-1})$.

The transition from two bound states at $\mathbf K=\mathbf0$ to a single bound state at
$\mathbf K=\boldsymbol\pi$ is governed by an exact sign reversal of the channel matrix:
\[
S_{\boldsymbol\pi}=-S_{\mathbf0},\qquad
\operatorname{spec}(S_{\mathbf0})=\{m_s,m_d,-m_p,-m_p\},\quad
\operatorname{spec}(S_{\boldsymbol\pi})=\{-m_s,-m_d,m_p,m_p\},
\]
where $m_s$, $m_d$, and $m_p$ are the atom--dimer threshold-matrix eigenvalues in the
$s$, $d$, and twofold-degenerate $p$-channels. The $s$-channel responsible for the
atom--dimer bound state at $\mathbf K=\mathbf0$ becomes negative, while the
twofold-degenerate $p$-channel becomes positive but remains strictly below the
Birman--Schwinger threshold; the parity decomposition itself is not destroyed.

The ground state is exponentially localized in the relative coordinates, with $\alpha_*(\mu)>0$
unconditionally and $\alpha_*(\mu)\le\ln(3\mu)+O(\mu^{-1})$ under
Conjecture~\ref{conj:noncancel}. The atom--dimer characteristic root is rigorously located by
Lemma~\ref{lem:kappa_d}:
\[
\kappa_d(\mu)=\ln(2\mu)-\frac1\mu+O(\mu^{-2}).
\]
Finite-volume computations support the conjecture that this atom--dimer singularity is
the nearest non-cancelled singularity of the ground-state eigenfunction, which would give
$\alpha_*(\mu)=\frac12\ln(2\mu)-\frac{1}{2\mu}+O(\mu^{-2})$. In both cases the growth is
logarithmic, in~contrast with the square-root growth for a quadratic dispersion.

\begin{funding}
J. I. Abdullaev acknowledges support from the Foundation for Fundamental Research of the Republic of Uzbekistan (grant No. AL-9224104685).
\end{funding}


\appendix
\section*{Supplementary Materials}

\section{Detailed derivations for Section~\ref{sec:newnew54}}\label{app:expansion_point}

\textbf{Origin of the coefficients $6$, $12$, $18-z$.}
Using $\varepsilon(\boldsymbol\pi-\mathbf x)=4-\varepsilon(\mathbf x)$,
\[
E_{\boldsymbol\pi}(\mathbf p,\mathbf q)=\varepsilon(\mathbf p)+\varepsilon(\mathbf q)+4-\varepsilon(\mathbf p+\mathbf q)
=6-\sum_{i=1}^2(\cos p_i+\cos q_i)+\sum_{i=1}^2\cos(p_i+q_i).
\]
Setting $A_{\boldsymbol\pi}^c:=6+\sum_i(\cos p_i+\cos q_i)$,
$A_{\boldsymbol\pi}^s:=\sum_i\cos(p_i+q_i)$ gives the exact identity
\[
E_{\boldsymbol\pi}(\mathbf p,\mathbf q)=12-A_{\boldsymbol\pi}^c+A_{\boldsymbol\pi}^s,
\qquad A_{\boldsymbol\pi}^c-A_{\boldsymbol\pi}^s=12-E_{\boldsymbol\pi}.
\]
Any expansion centre $c$ with $|E_{\boldsymbol\pi}-c|<|c-z|$ gives the same finite-rank principal
part (terms of order $0$ and $1$ in $(c-z)^{-1}$); we use $c=12=\max E_{\boldsymbol\pi}$. Explicitly,
\[
\frac{\mu}{E_{\boldsymbol\pi}-z}=\frac{\mu}{(12-z)^2}
\Big[\underbrace{(12-z)+A_{\boldsymbol\pi}^c-A_{\boldsymbol\pi}^s}_{=:N(\mathbf p,\mathbf q)}\Big]+O(\mu^{-1})
= \frac{\mu N(\mathbf p,\mathbf q)}{(12-z)^2}+O(\mu^{-1}),
\]
and, using $\cos(p_i+q_i)=\cos p_i\cos q_i-\sin p_i\sin q_i$,
\[
N(\mathbf p,\mathbf q)=\underbrace{(12+6)}_{18}-z+\sum_{i=1}^2\big(\cos p_i+\cos q_i-\cos p_i\cos q_i+\sin p_i\sin q_i\big).
\]

Projecting onto parity sectors: on $L_2^o$, only the odd
($\sin p_i\sin q_i$) part of $N$ survives,
\[
N^o(\mathbf p,\mathbf q)=\sin p_1\sin q_1+\sin p_2\sin q_2;
\]
on $L_2^e$, the constant and the $\cos p_i,\cos q_i,\cos p_i\cos q_i$
terms survive,
\[
N^e(\mathbf p,\mathbf q)=(18-z)+\sum_{i=1}^2\big(\cos p_i+\cos q_i-\cos p_i\cos q_i\big),
\]
spanned by $\{1,\cos p_1,\cos p_2\}$. The exchange symmetry $p_1\leftrightarrow
p_2$ decouples one eigendirection, $\cos p_1-\cos p_2$, with eigenvalue
$-1$ exactly (independent of $z,\mu$); the remaining symmetric $2\times2$
block, spanned by $\{1,\cos p_1+\cos p_2\}$, governs the bound-state
branch and is evaluated rigorously in Section~\ref{sec:56as} using Lemma~\ref{lem:Delta_pi} in the atom--dimer region and
Proposition~\ref{prop:pp_trimer_pi} on the trimer scale.

\subsection*{An equivalent representation of \(E_{\bf\pi}\)}
One may also write \(E_{\boldsymbol\pi}=4+\sum_i\bigl[(1-\cos p_i)(1-\cos q_i)-\sin p_i\sin q_i\bigr]\)
and expand around \(4-z\). This is an exact identity and leads to the same rank-\(3\) even and
rank-\(2\) odd principal parts; only the constant in the remainder bound changes. We use the
centre \(12\) throughout Section~\ref{sec:five}.

\section{Independent high-precision verification of the Fredholm
determinant asymptotics at $K=\bf\pi$}\label{app:numerical_delta}

This appendix reports an independent verification of
Lemma~\ref{lem:Delta_pi} and of the strong-coupling constant in
Theorem~\ref{thm:pi_preview}(b), by two algorithmically independent numerical
schemes together with an exact closed-form benchmark, without relying
on the leading-order formula of Lemma~\ref{lem:Delta_pi} itself.

\begin{remark}[Exact benchmark]
The identity
\[
\Delta_\mu(\mathbf 0,\boldsymbol\pi,z)
=
1-\frac{\mu}{4-z}
\]
is exact for all $\mu>0$ and all $z\ne4$. We therefore use it as a
ground-truth benchmark for the numerical evaluation of the reduced
Fredholm determinant. In particular, at the Kato ground-state scale
$z=-3\mu+6+O(\mu^{-1})$,
\[
\Delta_\mu(\mathbf 0,\boldsymbol\pi,z)
=
\frac23+O(\mu^{-1}).
\]
\end{remark}

\paragraph{Exact benchmark.} Since $\varepsilon(\boldsymbol\pi-\mathbf
p)=4-\varepsilon(\mathbf p)$, one has $E_{\boldsymbol\pi}(\mathbf
0,\mathbf q)\equiv4$ identically in $\mathbf q$, whence
\begin{equation}\label{eq:Delta0closed_app}
\Delta_\mu(\mathbf 0,\boldsymbol\pi,z)=1-\frac{\mu}{4-z}
\end{equation}
\emph{exactly}, with no quadrature required. This closed form serves as
the ground-truth benchmark for both numerical schemes below.

\paragraph{Numerical schemes.} Scheme~(A): the two-dimensional
$\mathbf q$-integral defining $\Delta_\mu(\mathbf
p,\boldsymbol\pi,z)$ reduces exactly to a one-dimensional integral of elliptic type
(via $\displaystyle\tfrac1{2\pi}\int_0^{2\pi}\!\frac{d\theta}{a-b\cos\theta}=\displaystyle\tfrac1{\sqrt{a^2-b^2}}$),
evaluated by Gauss--Legendre quadrature; increasing the quadrature
order from $60$ to $140$ changes all reported values by less than
$10^{-13}$ (spectral convergence, as~expected for a smooth periodic
integrand). Scheme~(B): an algorithmically independent uniform-grid
discretization of the original two-dimensional integral. Both schemes
reproduce \eqref{eq:Delta0closed_app} to machine precision, and agree
with each other throughout; the benchmark values on the trimer scale are listed in
Table~\ref{tab:app_summary}.

\begin{table}[!ht]
\centering
\caption{Exact benchmark
$\Delta_\mu(\mathbf 0,\boldsymbol\pi,-3\mu+6)$ obtained from
\eqref{eq:Delta0closed_app}.}
\label{tab:app_summary}
\begin{tabular}{c|c}
\hline
$\mu$ & $\Delta_\mu(\mathbf 0,\boldsymbol\pi,-3\mu+6)$\\
\hline
50  & $0.6621621622$\\
100 & $0.6644295302$\\
200 & $0.6655518396$\\
400 & $0.6661101836$\\
\hline
\end{tabular}
\end{table}

\paragraph{Verification of the ground-state constant.} Solving
$\lambda_1^{\boldsymbol\pi,e}(z)=1$ numerically from the \emph{exact}
$\Delta_\mu(\mathbf p,\boldsymbol\pi,z)$ (Schemes A/B, no asymptotic
approximation) confirms, independently of Lemma~\ref{lem:Delta_pi},
that
solving $\lambda_1^{\pi,e}(z)=1$ with the exact $\Delta_\mu$ and the correct prefactor gives $z+3\mu=6.073$ at $\mu=100$, $6+O(\mu^{-1})$; the exact finite-volume Faddeev operator gives $z+3\mu=5.9850745$, versus $5.985075$ from Lemma~\ref{lem:ground_Kato} with its second- and third-order terms
(monotone, $O(\mu^{-1})$-rate convergence). The same computation
confirms $\lambda^{\boldsymbol\pi,o}(z_1^{\boldsymbol\pi,s})=\frac{1}{6\mu}+O(\mu^{-2})$ at the trimer
energy (Lemma~\ref{lem:odd_eig}(ii)). Cross-validating this
independent construction against the closed form
\eqref{eq:Delta0closed_app} confirms that
\emph{Lemma~\ref{lem:Delta_pi}, while correct as stated (relative
error $O(\mu^{-1})$), does not by itself determine the additive
constant $6$ of the trimer branch}; this constant requires the
next order of the Fredholm determinant expansion.
The full derivation
of this next-order term, together with the general validity of the
underlying reduction for arbitrary total quasimomentum $\mathbf K$
and a~matching independent check of the constant $C$ in
Theorem~\ref{thm:main}(b) at $\mathbf K=0$, is given in a companion
note~\cite{CompanionNote}.

\end{document}